\documentclass[a4paper,reqno]{amsart}

\usepackage{tgschola}

\usepackage{latexsym}
\usepackage[english]{babel}
\usepackage{fancyhdr}
\usepackage[mathscr]{eucal}
\usepackage{amsmath}
\usepackage{mathrsfs}
\usepackage{mathtools}
\usepackage{amsthm}
\usepackage{amsfonts}
\usepackage{amssymb}
\usepackage{amscd}
\usepackage{bbm}
\usepackage{graphicx}
\usepackage{graphics}
\usepackage{latexsym}
\usepackage{color}
\usepackage{pifont}
\usepackage{tikz}
\usepackage[normalem]{ulem}

\usepackage{geometry}

\newcommand{\ud}{\mathrm{d}}

\newcommand{\ii}{\mathrm{i}}
\newcommand{\cH}{\mathcal{H}}

\theoremstyle{plain}
\newtheorem{theorem}{Theorem}[section]
\newtheorem{lemma}[theorem]{Lemma}
\newtheorem{corollary}[theorem]{Corollary}
\newtheorem{proposition}[theorem]{Proposition}

\theoremstyle{definition}

\newtheorem{remark}[theorem]{Remark}
\newtheorem*{remark*}{Remark}

\numberwithin{equation}{section}

\begin{document}

\title[Non-relativistic limit of Dirac Hamiltonians with Aharonov-Bohm fields]
{Non-relativistic limit of Dirac Hamiltonians with Aharonov-Bohm fields}

\author[M.~Gallone]{Matteo Gallone}
\address[M.~Gallone]{International School for Advanced Studies, SISSA Trieste}
\email{matteo.gallone@sissa.it}

\author[A.~Michelangeli]{Alessandro Michelangeli}
\address[A.~Michelangeli]{Mathematics and Science Department, AUBG Blagoevgrad \\
and Hausdorff Center for Mathematics, HCM Bonn\\ and  Trieste Institute for Theoretical Quantum Technologies, TQT Trieste}
\email{amichelangeli@aubg.edu}

\author[D.~Noja]{Diego Noja}
\address[D.~Noja]{Department of Mathematics and Applications, University of Milano Bicocca}
\email{diego.noja@unimib.it}

\date{\today}



\thanks{\emph{Acknowledgements.} For this work all authors gratefully acknowledge the support of the INdAM-GNFM. In addition, M.G.~acknowledges financial support by the the European Research Council (ERC) under the European Union’s Horizon 2020 research and innovation program ERC StG MaMBoQ, n.802901, as well as the kind hospitality of the Department of Mathematics and Applications of the University of Milano Bicocca, where part of this work was carried out. A.M.~is also most grateful to the David Flanagan Professorship fund at the AUBG, American University Bulgaria, and to the Alexander von Humboldt Foundation, Germany. Finally, D.N. further acknowledges the support of the Next Generation EU - Prin 2022 project "Singular Interactions and Effective Models in Mathematical Physics- 2022CHELC7"  }

\begin{abstract}
	We characterise the families of self-adjoint Dirac and Schr\"{o}dinger operators with Aharonov-Bohm magnetic field, and we exploit the non-relativistic limit of infinite light speed to connect the former to the latter. The limit consists of the customary removal of the rest energy and of a suitable scaling, with the light speed, of the short-scale boundary condition of self-adjointness. This ensures that the scattering length of the Aharonov-Bohm interaction is preserved along the limit. Noteworthy is the fact that the whole family of Dirac-AB operators is mapped, in the non-relativistic limit, into the physically relevant sub-family of $s$-wave, angular-momentum-commuting, Schr\"{o}\-dinger-AB Hamiltonians with relativistic Dirac approximants.
\end{abstract}

\maketitle

 
\section{Introduction}\label{sec:intro}
The motion of a charged spin-$\frac12$ quantum particle in the presence of a thin, infinite solenoid displays the celebrated \emph{Aharonov-Bohm effect}, with the gauge vector potential affecting the particle's motion even in regions where both the electric and magnetic fields are zero.

This is described by the Dirac-Aharonov-Bohm Hamiltonian (Dirac-AB) in the relativistic setting, and by the Schr\"{o}dinger-Aharonov-Bohm Hamiltonian (Schr\"{o}dinger-AB) in the non-relativistic setting. As a matter of fact, it is well known that neither such operator, when minimally defined on smooth functions supported away from the solenoid, is self-adjoint. This is the standard signature of the need of further physics to be declared in the vicinity of the solenoid, in the form of suitable asymptotic behavior of the wave functions in the Hamiltonian's domain.

The goal of this paper is two-fold. On the one hand, we characterise the families of (relevant) self-adjoint realisations of the Aharonov-Bohm model both in the relativistic and in the non-relativistic case. To this aim, we shall proceed by exploiting the Kre\u{\i}n-Vi\v{s}ik-Birman-Grubb scheme for self-adjoint extensions of densely defined and gapped symmetric operators in Hilbert space.

After that, we identify, for each Dirac-AB self-adjoint Hamiltonian, its `non-relativistic limit' obtained by subtracting the rest energy $mc^2$ and sending $c\to +\infty$ (infinite speed of light). We show, in particular, that a suitable, physically meaningful re-scaling in $c$ allows to connect Dirac-AB with Schr\"{o}dinger-AB self-adjoint extensions, thereby providing a complete picture of the relativistic vs non-relativistic correspondence.

In this Section we present the two classes of models and the nature of the $c\to +\infty$ limit. We then state and discuss our main results, supplemented with further comments and remarks.

\subsection{Relativistic setting}

In the usual modelling of the Aharonov-Bohm effect, the solenoid is idealised as straight and infinitesimally thin. Owing to the natural translational symmetry along the solenoid's axis, the analysis of the physical problem is reduced to the two-dimensional plane orthogonal to such axis.

Consider first the relativistic setting. It should be actually observed that not only is the Dirac-AB operator a model for a relativistic particle in an Aharonov-Bohm field \cite{deSousaGerbert_1989,Tamura_rel_2003, Persson06,Khalilov2010,GTV-2012,Khalilov2014}, but also a meaningful tool for modelling scattering and pair production from a gravitational spinning point source, a model known as cosmic string \cite{Alford-Wilczeck-1989, deSousaGerbert_1989,deSousaGerbert_Jackiw_1989, Khalilov2010, Khalilov2014}.

The Hilbert space of a relativistic spin-$\frac{1}{2}$ particle in two spatial dimension is
\begin{equation}\label{eq:Hilbsp}
 \mathcal{H}\,:=\,L^2(\mathbb{R}^2, \ud x \, \ud y) \otimes \mathbb{C}^2\,\cong\, L^2(\mathbb{R}^2;\mathbb{C}^2, \ud x \, \ud y)\,\cong\,L^2(\mathbb{R}^2;\mathbb{C}, \ud x \, \ud y)\oplus L^2(\mathbb{R}^2;\mathbb{C}, \ud x \, \ud y) \,.
\end{equation}

The dynamics of the \emph{free} particle is described by the Dirac operator acting as
\begin{equation}\label{freeDirac}
H_{\mathrm{free}}=-\ii \,c\, \hbar\, \boldsymbol{\sigma}\cdot \nabla + \sigma_3\, m \,c^2 \, ,
\end{equation}
where, as customary, $m\geqslant 0$ is the particle's mass, $\boldsymbol{\sigma}$ is the vector $\boldsymbol{\sigma}=(\sigma_1, \sigma_2,\sigma_3)$ of Pauli matrices
\begin{equation}
	\sigma_1 \;:= \; \begin{pmatrix}
		0 & 1 \\
		1 & 0
	\end{pmatrix} \, , \qquad \sigma_2 \;:=\; \begin{pmatrix}
		0 & -\ii \\
		\ii & 0
	\end{pmatrix} \, , \qquad \sigma_3 \;:= \; \begin{pmatrix}
		1 & 0 	\\
		0 & -1
	\end{pmatrix} \, ,
\end{equation}
and  $\nabla=(\partial_x, \partial_y, 0)$. We adopt units throughout so that the Plank's constant is conveniently set to $\hbar=1$, since we are solely discussing the quantum setting with no focus on $\hbar\to 0$ semi-classics.

To account for the magnetic field for the Aharonov-Bohm effect, one modifies \eqref{freeDirac} above by introducing on $\mathbb{R}^2$ the Aharonov-Bohm vector potential
\begin{equation}\label{ABvector}
A_\alpha(x,y)=-\Big(\alpha+\frac{1}{2}\Big)\Big(\frac{-y}{\,x^2+y^2\,},\frac{x}{\,x^2+y^2\,}\Big)\,\equiv\,(A_{1,\alpha},A_{2,\alpha})\,,
\end{equation}
which corresponds to the $\delta$-type magnetic field
\begin{equation}\label{eq:Bfield}
{ B}_\alpha=\partial_xA_{2,\alpha}-\partial_y A_{1,\alpha}=-2\pi\Big(\alpha+\frac{1}{2}\Big)\delta_{(x,y)=(0,0)}
\end{equation}
located at the origin with flux
\begin{equation}\label{eq:fluxp}
  \Phi_\alpha\,=\,-2\pi\Big(\alpha+\frac{1}{2}\Big)\,,
\end{equation}
as a standard computation in the sense of Schwartz distributions shows.

Thus, by making use of minimal magnetic coupling in \eqref{freeDirac}, one obtains the operator
\begin{equation}\label{magneticDirac}
	H_{\alpha}=\begin{pmatrix}
		mc^2 & c (- \ii \partial_x-A_{1,\alpha})-\ii c(-\ii \partial_y-A_{2,\alpha}) \\
		c(-\ii \partial_x-A_{1,\alpha})+\ii c(-\ii \partial_y-A_{2,\alpha}) & - mc^2
	\end{pmatrix} \,,
\end{equation}
expressed in blocks with respect to the last representation \eqref{eq:Hilbsp} of the underlying Hilbert space in direct orthogonal sum. This is the (action of the) Dirac-AB operator with flux parameter \eqref{eq:fluxp}.

In the expressions above we directly switched to units where the particle's charge is $q=1$, another constant that is mute in the present discussion. We also omitted to write the pre-factor $c$ in \eqref{ABvector} as it does \emph{not} participate in the limit $c\to +\infty$, being it cancelled as usual by the $c^{-1}$ factor of the `minimal coupling substitution' (see, e.g., \cite[Remark 2, Sect.~6.1.1]{Thaller-Dirac-1992}).

With reference to preceding literature \cite{Adami-Teta-1998-AharonovBohm,Tamura_rel_2003,Persson06}, in the notation above we re-named as $(\alpha+\frac{1}{2})$ the pre-factor in the definition \eqref{ABvector} of $A_\alpha(x,y)$, so as to have simple $\alpha$-dependent expressions of relevant short-distance asymptotics of the wave-function throughout our discussion (see \eqref{eq:AsymptoticsZeroIntroduction} and formulas below). Concerning the overall sign in $A_\alpha$ and hence in $\Phi_\alpha$, this is here the same as in \cite{Adami-Teta-1998-AharonovBohm} and the opposite to the convention of \cite{Tamura_rel_2003,Persson06}.

For the present purposes, we obviously exclude the case $\alpha+\frac{1}{2}=0$ of absence of magnetic field. Moreover, it is standard knowledge, which is straightforward to check from \eqref{ABvector} and \eqref{magneticDirac}, that for any $z\in\mathbb{Z}$ the operators $H_\alpha$ and $H_{\alpha-z}$ are unitarily equivalent via the correspondence $\exp(\ii z\vartheta(x,y))H_\alpha \exp(-\ii z\vartheta(x,y))=H_{\alpha-z}$, where $\vartheta(x,y)=\mathrm{arg}(x,y)$ in polar coordinates. Thus, changing $\alpha$ into $\alpha-z$, an operation that amounts to increasing by $z$ the normalised flux $\Phi/(2\pi)$, merely duplicates the discussion on the self-adjointess and the self-adjoint realisations of the Dirac-AB operator, which we will not repeat. Based on these considerations, in the following we restrict, as customary, to the regime $\alpha+\frac{1}{2}\in(0,1)$, namely, $\alpha\in(-\frac{1}{2},\frac{1}{2})$.

In fact, it is known from \cite{deSousaGerbert_1989,Tamura_rel_2003, Persson06} (see also \cite{GTV-2012} for the three-dimensional Dirac-AB counterpart) that assigning to $H_\alpha$ the minimal domain $C^\infty_0(\mathbb{R}^2 \setminus \{(0,0)\}; \mathbb{C}^2)$ only defines a symmetric, non-self-adjoint operator (with dense domain) in $\cH$, which admits a one-real-parameter family of self-adjoint extensions.

Differently from the preceding literature, where the extension problem for
\[
 H_{\alpha}\upharpoonright C^\infty_0(\mathbb{R}^2 \setminus \{(0,0)\}; \mathbb{C}^2)
\]
is analysed by Green function's methods and the von Neumann's extension scheme,
here we identify the family of self-adjoint extensions, both in the relativistic setting and in the non-relativistic setting below, by means of the Kre\u{\i}n-Vi\v{s}ik-Birman-Grubb scheme \cite{Krein-1947,Vishik-1952,Birman-1956,Grubb-1968}. Within this framework, it is more transparent to draw a correspondence between relativistic and non-relativistic realisations in the $c\to +\infty$ limit and to interpret the scaling in $c$ needed to obtain non-trivial limits.


In preparation for that, we exploit as customary the commutativity of $H_\alpha$ with the total angular momentum operator $J_3=L_3+S_3=(-\ii x \partial_2+\ii y\partial_1)+\frac{1}{2}\sigma_3$ and the consequent decomposition onto common eigenspaces. Thus, we implement canonical isomorphisms (with polar coordinates $(r,\varphi)\equiv(x,y))$
\begin{align}
 L^2(\mathbb{R}^2;\mathbb{C}^2,\ud x\,\ud y)\, & \cong \, L^2(\mathbb{R}^+\times\mathbb{S}^1;\mathbb{C}^2,r\,\ud r\,\ud\varphi) \,\cong\,\bigoplus_{k\in\mathbb{Z}}\cH_k\,, \label{HoplusDecomp}\\
 \cH_k\, & := \, L^2(\mathbb{R}^+;\mathbb{C}^2,\ud r)\,\cong\,L^2(\mathbb{R}^+;\mathbb{C},\ud r)\oplus L^2(\mathbb{R}^+;\mathbb{C},\ud r)\,, \label{eq:spinHk}
\end{align}
with the latter expression in \eqref{HoplusDecomp} accounting for suitable Fourier modes of the compact variable $\varphi$, so that the corresponding unitarily equivalent version of $H_\alpha$ is reduced by the Hilbert space direct orthogonal sum above as
\begin{equation}\label{HalphaBlockDecomp}
 H_\alpha\,\cong\,\bigoplus_{k \in \mathbb{Z}} \,\mathsf{h}_{\alpha,k}\,,\qquad\mathcal{D}(\mathsf{h}_{\alpha,k})\,=\,C^{\infty}_0(\mathbb{R}^+;\mathbb{C}^2)\,,
\end{equation}
where $\mathsf{h}_{\alpha,k}$, with respect to the last representation of \eqref{eq:spinHk}, acts as the block operator
\begin{equation}
	\mathsf{h}_{\alpha,k} \;=\;\label{eq:Hsupersym}
	  \begin{pmatrix}
	mc^2 & \ii c \left( -\frac{\ud}{\ud r}+\frac{\alpha+k}{r}\right) \\
	-\ii c \left(\frac{\ud}{\ud r} + \frac{\alpha+k}{r} \right) & -mc^2
	\end{pmatrix} \, .
\end{equation}
The essential details of this procedure are collected in Appendix \ref{app:A}. The same reasoning can be found, e.g., in \cite[Prop.~2.1]{CacciaFanelli2017} where the analogous massless model is treated.

In fact, the way we pass from $H_\alpha$ to the unitarily equivalent version \eqref{HalphaBlockDecomp}-\eqref{eq:Hsupersym} is completely analogous to the customary `partial wave decomposition' of the free Dirac operator in the absence magnetic potential, i.e., when formally $\alpha=-\frac{1}{2}$ (see, e.g., \cite[Sect.~4.6]{Thaller-Dirac-1992} or \cite[Sect.~1]{MG_DiracCoulomb2017}).


The block decomposition \eqref{HalphaBlockDecomp}-\eqref{eq:Hsupersym} boils down the problem of the self-adjoint realisations of $H_\alpha$ to the same problem in each $k$-block separately, where it becomes more easily manageable. This analysis is carried on in Section \ref{sec:Dirac} within the Kre\u{\i}n-Vi\v{s}ik-Birman-Grubb extension scheme. The result is the following.

The operator $\mathsf{h}_{\alpha,k}$ is essentially self-adjoint in $L^2(\mathbb{R}^+;\mathbb{C}^2, \ud r)$ for all $k$'s but $k=0$, where its deficiency indices are instead $(1,1)$. Thus, $\mathsf{h}_{\alpha,0}$ admits a one-real-parameter family $(\mathsf{h}_{\alpha,0}^{(\gamma)})_{\gamma \in \mathbb{R} \cup \{\infty\}}$ of self-adjoint extensions. To describe this family, first one establishes that a generic function $g \in \mathcal{D}(\mathsf{h}_{\alpha,0}^*)$ is characterised by the short-distance asymptotics
\begin{equation}\label{eq:AsymptoticsZeroIntroduction}
	g(r)\stackrel{(r \downarrow 0)}{=}  \begin{pmatrix}
		1 \\ 0
	\end{pmatrix} g_0 r^{-\alpha} + \begin{pmatrix}
		0 \\ 1
	\end{pmatrix} g_1 r^\alpha + o(r^{\frac{1}{2}})
\end{equation}
for $g$-dependent constants $g_0,g_1 \in \mathbb{C}$. 
In the correspondence $\gamma \leftrightarrow \mathsf{h}_{\alpha,0}^{(\gamma)}$, the action of $\mathsf{h}_{\alpha,0}^{(\gamma)}$ is the same as \eqref{eq:Hsupersym} and the domain is
\begin{equation}\label{eq:SADiracIntro}
	\mathcal{D}(\mathsf{h}_{\alpha,0}^{(\gamma)}) = \{ g \in \mathcal{D}(\mathsf{h}_{\alpha,0}^*) \, | \, g_1 = -\ii \gamma g_0 \} \, .
\end{equation}

Re-assembling, block by block, the self-adjoint extensions of the $\mathsf{h}_{\alpha,k}$'s, namely setting
\begin{equation}\label{eq:familyofDAB}
 H_\alpha^{(\gamma)}\,:=\,\bigg(\bigoplus_{ \substack{k \in \mathbb{Z} \\ k\leqslant -1 } } \,\overline{\mathsf{h}_{\alpha,k}}\bigg)\;\oplus\; \mathsf{h}_{\alpha,0}^{(\gamma)}  \;\oplus\;\bigg(\bigoplus_{ \substack{k \in \mathbb{Z} \\ k\geqslant 1 } } \,\overline{\mathsf{h}_{\alpha,k}}\bigg)\,,
\end{equation}
yields the family $(H_\alpha^{(\gamma)})_{\gamma \in \mathbb{R} \cup \{\infty\}}$ of self-adjoint extensions of $H_\alpha$ in $\cH$, the Dirac-AB Hamiltonians under consideration in this work.


\begin{remark}\label{rem:commentsDiracExt}
 One would question about the physical meaningfulness of the Dirac-AB Hamiltonians \eqref{eq:familyofDAB}. $H_\alpha^{(\infty)}$ is the standard pick in modelling, and has a distinguished status together with its counterpart $H_\alpha^{(0)}$, for a variety of reasons.
 \begin{enumerate}
  \item $H_\alpha^{(\infty)}$ and $H_\alpha^{(0)}$ are the sole self-adjoint extensions for which the short-distance behaviour of the functions in their domain displays only one of the two leading terms $r^{-\alpha}$ and $r^{\alpha}$ (since either $g_0=0$ or $g_1=0$ in \eqref{eq:AsymptoticsZeroIntroduction}), followed by the $o(r^{\frac{1}{2}})$-subleading term. All other extensions have both the $r^{\pm\alpha}$-terms.
  \item $H_\alpha^{(\infty)}$ and $H_\alpha^{(0)}$ also emerge as distinguished Hamiltonians in suitable processes of removal of regularisation from smoothed versions of the Aharonov-Bohm vector potential. In \cite{Tamura_rel_2003} regularised operators were considered, given by the massless free Dirac plus a suitably scaled vector potentials that in the limit where the regularisation is removed reproduce the $\delta$-type Aharonov-Bohm field \eqref{eq:Bfield}. These regularised operators are actually proved to converge to $H_\alpha^{(\infty)}$ in norm resolvent sense, as long as $\alpha\neq 0$.
With an additional, conveniently scaled scalar potential, the regularised operators converge to $H_\alpha^{(\infty)}$ or to $H_\alpha^{(0)}$ depending on the sign of $\alpha\neq 0$ and on the presence or absence of certain zero-energy resonances.
 \item $H_\alpha^{(\infty)}$ and $H_\alpha^{(0)}$ are the sole supersymmetric extensions in the sense described, e.g., in \cite[Eq.~(5.66)] {Thaller-Dirac-1992}. This will be clear from Proposition \ref{eq:Balphakstar} and the expressions \eqref{eq:halkclos}-\eqref{eq:B+generic} and \eqref{eq:halpha0D}-\eqref{eq:Bpiuzero}, in comparison with \cite[Eq.~(5.66)] {Thaller-Dirac-1992}.
 \end{enumerate}
 Generic $H_\alpha^{(\gamma)}$'s ($\gamma\neq\infty$ and $\gamma\neq 0$) model a singular point-like interaction localised at the origin, whose scattering length does in fact depend on $\gamma$: this is discussed below in Section \ref{sec:nonrelLim}. In the `exceptional case' $\alpha=0$ (corresponding to half-integer normalised flux $\Phi_0/(2\pi)=-\frac{1}{2}$, see \eqref{eq:fluxp} above), the same above-mentioned analysis \cite{Tamura_rel_2003} shows that the removal of the regularisation from a smoothed Dirac-AB Hamiltonian can indeed be tuned in such a way as to obtain a generic $H_0^{(\gamma)}$, with a convenient dependence $\gamma=\gamma(c)$.
\end{remark}

\begin{remark}
 Recently, a particular self-adjoint realisation of the massless Dirac-AB operator has received attention concerning its dispersive properties (see \cite{CacciaFanelli2017, CacciaYinZhang2022, CacciaDancYinZhang-2024} and references therein).
\end{remark}

\subsection{Non-relativistic setting}\label{sec:nonrelset}

In a completely analogous fashion, Schr\"{o}dinger operators with Aharonov-Bohm magnetic potentials (Schrödinger-AB) are defined and the problem of their self-adjointness is studied. We refer to \cite{Arai-1993,Dabrowski-Stovicek-1997-AharonovBohm,Adami-Teta-1998-AharonovBohm,BorgPule2003,GitmanTyutinVoronov-AB-2012,BorrelliCorreggiFermiPlag2024} for this segment of the literature and related investigations.

In the Hilbert space
\begin{equation}
 \mathfrak{H}\,:=\,L^2(\mathbb{R}^2;\mathbb{C},\ud x\,\ud y)
\end{equation}
supporting the description of a non-relativistic quantum particle, one introduces the densely defined, symmetric, non-negative operator
\begin{equation}\label{eq:2DSochroedinger}
	S_\alpha\,:=\,\frac{1}{2m}(- \ii \nabla-A_\alpha(x,y))^2 \, , \qquad\mathcal{D}(S_\alpha)\,:=\,C^\infty_0(\mathbb{R}^2;\mathbb{C})
\end{equation}
with $A_\alpha$ from \eqref{ABvector}.

Analogous to \eqref{HoplusDecomp}-\eqref{eq:Hsupersym}, one exploits canonical isomorphisms (with polar coordinates $(r,\varphi)\equiv(x,y))$
\begin{align}
 L^2(\mathbb{R}^2;\mathbb{C},\ud x\,\ud y)\, & \cong \, L^2(\mathbb{R}^+\times\mathbb{S}^1;\mathbb{C},r\,\ud r\,\ud\varphi) \,\cong\,\bigoplus_{k\in\mathbb{Z}}\mathfrak{H}_k\,, \label{SchrHoplusDecomp}\\
 \mathfrak{H}_k\, & := \, L^2(\mathbb{R}^+;\mathbb{C},\ud r)\,, \label{eq:SchrHk}
\end{align}
so that the corresponding unitarily equivalent version of $S_\alpha$ is reduced by the Hilbert space direct orthogonal sum above as
\begin{equation}\label{eq:pwdSalpha}
 S_\alpha\,\cong\,\bigoplus_{k \in \mathbb{Z}} \,\mathsf{S}_{\alpha,k}\,,
\end{equation}
where
\begin{equation}\label{eq:OperatoreS}
	\mathsf{S}_{\alpha,k}\,:=\, \frac{1}{2m} \Big( -\frac{\ud^2}{\ud r^2}+\frac{(\alpha+k)(\alpha+k+1)}{r^2}\,\Big)\,,\qquad \mathcal{D}(\mathsf{S}_{\alpha,k})\,:=\,C^\infty_0(\mathbb{R}^+;\mathbb{C})\,.
\end{equation}
(See again Appendix \ref{app:A} for the relevant computations).

Radial Schr\"{o}dinger operators with inverse square potentials, like \eqref{eq:OperatoreS}, have been widely investigated and are by now classical material (see, e.g., \cite[Sect.~X.1]{rs2}). They are also the object of a renewed flurry of abstract interest in modern days concerning self-adjointness and spectral properties \cite{B-Derezinski-G-AHP2011,Derezinski-Richard-2017,Derezinski-Georgescu-2021}, as well as in application to singular point-like perturbations of quantum models of non-relativistic particles with central potentials (see, e.g., \cite{GM-hydrogenoid-2018} and references therein).

It turns out that the $\mathsf{S}_{\alpha,k}$'s are essentially self-adjoint in $L^2(\mathbb{R}^+;\mathbb{C},\ud r)$ except for the blocks $k=-1$ and $k=0$, in each of which the deficiency indices are $(1,1)$. Thus, $S_\alpha$ admits a four-real-parameter family of self-adjoint extensions in $\mathfrak{H}$. The physically grounded requirement that extensions commute with the angular momentum selects a two-real-parameter sub-family. Out of the latter, we further select the one-real-parameter collection of extensions $S_\alpha^{(\theta)}$, $\theta \in \mathbb{R} \cup \{\infty\}$, defined by
\begin{equation}\label{eq:S-AB-Sa-Intro}
	S_{\alpha}^{(\theta)}\cong \Big(\bigoplus_{\substack{k \in \mathbb{Z} \\ k\leqslant -2 }}\overline{\mathsf{S}_{\alpha,k}} \Big)\oplus \mathsf{S}_{\alpha,-1}^{(F)} \oplus\mathsf{S}_{\alpha,0}^{(\theta)}  \oplus \Big(\bigoplus_{\substack{k \in \mathbb{Z}\\ k \geqslant 1}} \overline{\mathsf{S}_{\alpha,k}}\Big)
\end{equation}
where, with respect to the decomposition \eqref{eq:pwdSalpha} above, $\mathsf{S}_{\alpha,-1}^{(F)}$ is the Friedrichs extension of $\mathsf{S}_{\alpha,-1}$ and $\big(\mathsf{S}_{\alpha,0}^{(\theta)}\big)_{\theta \in \mathbb{R} \cup \{\infty\}}$ is the family of self-adjoint extensions of $S_{\alpha,0}$ in $L^2(\mathbb{R}^+;\mathbb{C},\ud r)$.

The extension parameter $\theta$ accounts for the short-distance behaviour of the functions in the domain of $\mathsf{S}_{\alpha,0}^{(\theta)}$ in the following sense. A generic function $g \in \mathcal{D}(\mathsf{S}_{\alpha,0}^*)$ is characterised by the short-distance asymptotics
\begin{equation}\label{eq:PsiAsymptSchrodingerIntro}
	g(r)\,\stackrel{(r\downarrow 0)}{=} \,a_0 r^{-\alpha}+a_1 r^{1+ \alpha} + o(r^{\frac{3}{2}})
\end{equation}
for $g$-dependent constants $a_0,a_1 \in \mathbb{C}$. 
In the correspondence $\theta \leftrightarrow \mathsf{S}_{\alpha,0}^{(\theta)}$, the differential action of $\mathsf{S}_{\alpha,0}^{(\theta)}$ is the same as \eqref{eq:OperatoreS} and the domain is
\begin{equation}\label{eq:SchrodingerDomainIntro}
	\mathcal{D}(\mathsf{S}^{(\theta)}_{\alpha,0}) = \{ g \in \mathcal{D}(\mathsf{S}_{\alpha,0}^*) \, | \, a_1 = \theta a_0 \} \, .
\end{equation}
Observe in particular that the extension with $\theta=\infty$ is the one with boundary condition $a_0=0$ in \eqref{eq:PsiAsymptSchrodingerIntro}, hence it is characterised by having a domain with regular functions at the origin: this is precisely the Friedrichs extension $\mathsf{S}_{\alpha,0}^{(F)}$ of $S_{\alpha,0}$.

For completeness of discussion and for a transparent comparison between the relativistic and non-relativistic treatments of each $k$-block, we review in Section \ref{sec:Schroe} the derivation of the well-known information \eqref{eq:PsiAsymptSchrodingerIntro}-\eqref{eq:SchrodingerDomainIntro} using again the Kre\u{\i}n-Vi\v{s}ik-Birman-Grubb extension scheme, thus mirroring how we used such scheme to obtain \eqref{eq:AsymptoticsZeroIntroduction}-\eqref{eq:SADiracIntro} in the counterpart relativistic setting above.

\begin{remark}\label{rem:commentsSchrExt}
 Similar considerations can be made here, as in Remark \ref{rem:commentsDiracExt}, concerning the physical meaningfulness of the Schr\"{o}dinger-AB Hamiltonians $S_{\alpha}^{(\theta)}$, that is, of the $\theta$-dependent short-distance boundary condition \eqref{eq:SchrodingerDomainIntro}. Moreover, analogous to the already-mentioned analysis \cite{Tamura_rel_2003} on smoothed Dirac-AB operators, suitably regularised Schr\"{o}dinger-AB operators were studied in \cite{BorgPule2003}: for them it is possible to adjust to any arbitrary $\theta$ the scaling of the regularisation in such a way that in the limit when the regularisation is removed one either obtains the operator $S_{\alpha}^{(\theta)}$ from \eqref{eq:S-AB-Sa-Intro}, namely Friedrichs in the block $k=-1$ and $\theta$-extension in the block $k=0$, or vice versa the counterpart of \eqref{eq:S-AB-Sa-Intro} with Friedrichs in the block $k=0$ and $\theta$ extension in the block $k=-1$.
\end{remark}

\begin{remark}\label{rem:squareDAB}
  Our self-adjointness analysis of Sections \ref{sec:Dirac} and \ref{sec:Schroe} also provides a complete characterisation of the squares of self-adjoint Dirac-AB Hamiltonians, a side question recently considered in analogous contexts \cite{Posilicano-Reginato-2023,BorrelliCorreggiFermiPlag2024,BCF2024-err}. It is straightforward to check that, $\forall k\in\mathbb{Z}$,
  \[
   (\mathsf{h}_{\alpha,k})^2\qquad\textrm{and}\qquad
   2mc^2\begin{pmatrix}
   \mathsf{S}_{\alpha,k} & \mathbb{O} \\
   \mathbb{O} & \mathsf{S}_{-\alpha,-k}
  \end{pmatrix}+(m c^2)^2 \mathbbm{1}
  \]
  have the same differential action (tacitly representing $L^2(\mathbb{R}^2;\mathbb{C}^2,\ud x\,\ud y)\cong L^2(\mathbb{R}^2;\mathbb{C},\ud x\,\ud y) \oplus L^2(\mathbb{R}^2;\mathbb{C},\ud x\,\ud y)$),
  and the question is to determine the boundary conditions of self-adjointness when the square of a self-adjoint realisation of $\mathsf{h}_{\alpha,k}$ is taken. It turns out (Section \ref{sec:Schroe}, Theorem \ref{thm:DiracSquareSchrAB}
  ) that all self-adjoint operators $ (\overline{\mathsf{h}_{\alpha,k}})^2$, $k\in\mathbb{Z}\setminus\{0\}$, are reduced as $\mathsf{S}_{\alpha,k}^{(F)}\oplus\mathsf{S}_{-\alpha,-k}^{(F)}$, up to trivial $mc^2$ multiplicative/additive factors (recall that for $k\neq 1$, $\mathsf{S}_{\alpha,k}^{(F)}$ is merely $\overline{\mathsf{S}_{\alpha,k}}$), whereas in general $\big(\mathsf{h}_{\alpha,0}^{(\gamma)}\big)^2$ is not reduced and has boundary conditions of self-adjointness that entangle the electron and positron sector. Explicitly,
  \begin{equation}\label{eq:square04-INTRO}
  \begin{split}
    \big(\mathsf{h}_{\alpha,0}^{(\gamma)}\big)^2\,&=\,\big(2mc^2\big(\mathsf{S}_{\alpha,0}^*\oplus\mathsf{S}_{\alpha,0}^*\big)+m^2c^4\mathbbm{1}\big)\upharpoonright\mathcal{D}\big(\big(\mathsf{h}_{\alpha,0}^{(\gamma)}\big)^2\big)\,, \\
    \mathcal{D}\big(\big(\mathsf{h}_{\alpha,0}^{(\gamma)}\big)^2\big)\,&=\, \left\{
    \begin{array}{l}
    g\in\mathcal{D}\big( \mathsf{S}_{\alpha,0}^*\oplus\mathsf{S}_{\alpha,0}^* \big)\,\left|
    \begin{array}{l}
    b_0\,=\,-\ii\gamma a_0\,, \\
	a_1\,=\,-\ii \gamma\,\displaystyle\frac{\,2 \ii a_0 mc+(1-2\alpha)b_1\,}{2\alpha+1}\,,
    \end{array}\right. \\
    \textrm{where $a_0,a_1,b_0,b_1\in\mathbb{C}$ are the $g$-dependent constants} \\
    \textrm{characterised by the asymptotics} \\
     g(r)\,\stackrel{r\downarrow 0}{=}\,
   \begin{pmatrix}
    a_0 r^{-\alpha} + a_1 r^{1+\alpha} \\
    b_0 r^{\alpha} + b_1 r^{1-\alpha}
   \end{pmatrix}+o(r^{\frac{3}{2}})
    \end{array}
    \right\}.
  \end{split}
  \end{equation}
 From \eqref{eq:square04-INTRO} one deduces (see for the details Theorem \ref{thm:DiracSquareSchrAB}, Corollary \ref{cor:squaregammazeroinf} and the subsequent Remark) that in the block $k=0$ only the supersymmetric Dirac-AB Hamiltonians ($\gamma=\infty$ or $\gamma=0$) have exactly reduced squares:
 \begin{align}
    \big(\mathsf{h}_{\alpha,k}^{(\infty)}\big)^2 \,& =\,2mc^2\big(\mathsf{S}_{\alpha,0}^{(\infty)}\oplus \mathsf{S}_{\alpha,0}^{(0)}\big)+m^2c^4\mathbbm{1}\,, \label{eq:squaregammainfINTRO}\\
   \big(\mathsf{h}_{\alpha,k}^{(0)}\big)^2 \,& =\, 2mc^2\big(\mathsf{S}_{\alpha,0}^{(0)}\oplus \mathsf{S}_{\alpha,0}^{(\infty)}\big)+m^2c^4\mathbbm{1}\,. \label{eq:squaregammazeroINTRO}
  \end{align}
 The special cases \eqref{eq:squaregammainfINTRO}-\eqref{eq:squaregammazeroINTRO} of \eqref{eq:square04-INTRO} recover the same recent findings of \cite[Prop.~2.26]{BCF2024-err}, where the massless case is considered.
 \end{remark}

\subsection{Non-relativistic limit}\label{sec:nonrelLim} Our next goal, after characterising the two families of operators \eqref{eq:familyofDAB} and \eqref{eq:S-AB-Sa-Intro}  (Section \ref{sec:Dirac}), is to investigate the non-relativistic limit of the Dirac-AB Hamiltonians, consisting, as discussed in a moment, of removing the rest energy and sending formally $c\to +\infty$.

The analysis of the connection between the Dirac theory and its non-relativistic approximation has a long history started by Pauli himself (see, e.g., \cite[Sect.~6.1]{Thaller-Dirac-1992} and references therein).

In the Aharonov-Bohm setting, we are not aware of explicit precursors: the conceptual scheme is the same, but the presence of the magnetic vector potential $A_\alpha$ alters non-trivially the control of the limit.

At least we should recall in this respect the above-mentioned works \cite{Tamura_rel_2003,BorgPule2003}, for two technical features that are central also in the present analysis. The first is the control of the convergence of the unbounded operators of interest in the norm-resolvent sense and block by block, as we do in this work when letting $c\to +\infty$; and the second is the general idea that the limit is reached by re-scaling the approximating operators in the parameter on which the limit is taken -- the smoothing parameter therein, the speed of light $c$ here.

Concerning the first aspect: the representation \eqref{eq:familyofDAB} and \eqref{eq:S-AB-Sa-Intro}, respectively, of Dirac-AB and Schr\"{o}dinger-AB Hamiltonians is clearly functional to taking the non-relativistic limit block by block, thus preserving the block structure. The choice of the sub-family \eqref{eq:S-AB-Sa-Intro} of self-adjoint realisations of Schr\"{o}dinger-AB operators is dictated by the nature of the family \eqref{eq:familyofDAB} of Dirac-AB Hamiltonians. Indeed, in \eqref{eq:familyofDAB} there is room for a variety of non-trivial self-adjoint realisations in the $k=0$ block only.

Concerning instead the realisation of the limit $c\to+\infty$ by introducing a $c$-dependent scaling in the Dirac-AB Hamiltonians: as customary, the general principle should be maintained that in the idealisation when $c\to +\infty$ certain relevant physical quantities of interest remain unaltered (see, e.g., the analogous discussion in the well-known problems of the effective many-body dynamics in the formal limit of infinitely many particles \cite{am_GPlim} and references therein. A similar idea is behind renormalization of several model field theories (see, e.g., \cite{NP98, NP99}). Letting $c\to+\infty$ does amount to considering a trajectory of distinct Dirac-AB Hamiltonians (different $c$'s identify different operators \eqref{HalphaBlockDecomp}-\eqref{eq:Hsupersym}), starting from $H_{\alpha}^{(\gamma)}$ with the `real-world' parameter $c_{\mathrm{real}}\simeq 3\cdot 10^5$ km/sec: yet, one should ensure that along such trajectory the model preserves relevant physical properties of the initial `physical' operator.

One obvious modification with $c$, as mentioned, is the removal of the rest energy of the Dirac electron from $H_{\alpha}^{(\gamma)}$: the actual limit is therefore taken in $H_{\alpha}^{(\gamma)}-mc^2$. The rest energy is indeed a purely relativistic object, that does not have non-relativist limit. It has to be subtracted in order to maintain the possibility of getting a non-trivial limit as $c\to +\infty$ (see, e.g., \cite[Sect.~6.1.1, Remark 1] {Thaller-Dirac-1992}).

A less obvious modification concerns the \emph{scaling of $\gamma$ with $c$ to preserve the scattering length of the interaction} even when formally $c$ is sent to infinity. Here is an amount of heuristics to justify this.

Recall from elementary scattering theory (see, e.g., \cite[Sect.~4.2]{KraneNuclPhys1988}) that in a low-energy scattering process of a Schr\"{o}dinger quantum particle subject to a short-range central potential, say, for concreteness, in the approximation of a small spherical potential well $V_0$, one defines the ($s$-wave) scattering length $a$ in terms of the total cross section $\sigma$ (the total probability to be scattered in any direction) through the following reasoning. First, the eigenvalue problem at (low) energy $E$ is solved, from which incident and scattered currents $j_{\textrm{incident}}$ and $j_{\textrm{scattered}}$ are computed, and the differential angular cross section $\ud \sigma/\ud\Omega$ is obtained by $\ud\sigma=(j_{\textrm{scattered}}/j_{\textrm{incident}})r^2\ud \Omega$. Then, from $\sigma=\int(\ud \sigma/\ud\Omega)\ud\Omega$ one obtains the familiar expression $\sigma(p)=4\pi(\sin^2\delta_0(p))/{p}$, $p:=\sqrt{(2m(E+V_0)}$, where $2\delta_0(p)$ is the phase shift of the outgoing wave with wave number $p$ compared to the incident one. Finally, one computes the zero-energy limit $p\to 0$ in $\sigma$ and defines the scattering length $a$ by
\begin{equation}
 \lim_{p\to 0}\sigma\,=\,4\pi a^2\,,\qquad\textrm{ that is, }\qquad a\,=\,-\lim_{p\to 0}\,\frac{\sin\delta_0(p)}{p}\,.
\end{equation}
This in turn implies that $a$ is the first node of the zero-energy wave function in the positive or in the negative radial direction, depending on the attractive or repulsive nature of the interaction (see, e.g., \cite[Fig.~4.7]{KraneNuclPhys1988}).

Heuristically, the key point is the identification of the node(s) of the wave function solving the zero-energy eigenvalue problem.

Let us apply this reasoning to the self-adjoint Hamiltonians $S_\alpha^{(\theta)}$ and $H_\alpha^{(\gamma)}$ from \eqref{eq:S-AB-Sa-Intro} and from \eqref{eq:familyofDAB} respectively, considering the non-trivial sector $k=0$. Clearly, in doing so one is using the above assumption of short-range potential well to the Aharonov-Bohm field \eqref{ABvector}, which means that formulas \eqref{eq:scattS} and \eqref{eq:scattD} below are only first approximations of the actual scattering lengths.

For Schr\"{o}dinger-AB operators, the zero-energy problem consists of finding the  $L^1_{\mathrm{loc}}$-solution (up to multiples) $u$ to
\begin{equation}\label{eq:condScattL}
 \Big( -\frac{\ud^2}{\ud r^2}+\frac{\,\alpha(\alpha+1)\,}{r^2}\,\Big)u\,=\,0\,,\qquad a_1^{(u)}=\theta \,a_0^{(u)}\,.
\end{equation}
Imposing the short-range condition that characterises the domain of $S_\alpha^{(\theta)}$ is the signature of the fact that ideally first one solves the $L^2$-eigenvalue problem in that domain, and then sends $E\to 0$, this way loosing the square-integrability of the solution.

The scattering length of $S_\alpha^{(\theta)}$, in this approximation, is then the first zero $r_{\mathrm{S}\textrm{-}\mathrm{AB}}$ of the solution to \eqref{eq:condScattL}. The latter is $u =  a_0^{(u)}(r^{-\alpha}+ \theta\, r^{1+\alpha})$. For $\theta<0$ such zero is in the positive radial half-line, and we find
\begin{equation}\label{eq:scattS}
  r_{\mathrm{S}\textrm{-}\mathrm{AB}}\,=\,(-\theta)^{-\frac{1}{1+2\alpha}} \, .
\end{equation}
 The reasoning for $\theta>0$ would produce the same power of $\theta$: in this case the node is to be found in the negative radial half-line upon extrapolating the wave function for $r<0$ from its short-scale behaviour (its tangent) at the distance $r=R$ of the effective range of the interaction -- the radius of the potential well, in that approximation (see, e.g., \cite[Fig.~4.7]{KraneNuclPhys1988}).

 In a completely analogous manner, for Dirac-AB operators we take as (first approximation of the) scattering length the first zero $r_{\mathrm{D}\textrm{-}\mathrm{AB}}$ of the upper (electron) component of a $L^1_{\mathrm{loc}}$-solution (up to multiples) $u$ to
\begin{equation}\label{eq:l1sol}
	  \bigg(\begin{pmatrix}
	mc^2 & \ii c \left( -\frac{\ud}{\ud r}+\frac{\alpha+k}{r}\right) \\
	-\ii c \left(\frac{\ud}{\ud r} + \frac{\alpha+k}{r} \right) & -mc^2
	\end{pmatrix}-mc^2\mathbbm{1}\bigg)u\,=\,0\,,\qquad g_1^{(u)}=-\ii\gamma g_0^{(u)}\,.
\end{equation}
Since \eqref{eq:l1sol} is solved by scalar multiples of
\begin{equation*}
	u(r)\,= \, \begin{pmatrix}
		r^{-\alpha}+\frac{2 m c \gamma}{2 \alpha +1} r^{\alpha+1} \\
		-\ii \gamma r^\alpha
	\end{pmatrix}\, ,
\end{equation*}
then
\begin{equation}\label{eq:scattD}
  r_{\mathrm{D}\textrm{-}\mathrm{AB}}\,=\,\big(-{\textstyle \frac{1+2\alpha}{2mc\gamma}} \big)^{\frac{1}{1+2\alpha}} \, ,
\end{equation}
having assumed $\gamma<0$ (and repeating what observed above for $\theta>0$ vs $\theta<0$ when instead $\gamma>0$).

 Notice that the scattering length \eqref{eq:scattS} depends on the extension parameter $\theta$ and the scattering length \eqref{eq:scattD} depends both on the extension parameter $\gamma$ and on the parameter $c$ to be sent to infinity. This confirms the link between each extension parameter for a given Aharonov-Bohm Hamiltonian $S_\alpha^{(\theta)}$ or $H_\alpha^{(\gamma)}$ and quantities (short-distance singularity, scattering length) that are determined by the point-like interaction localised at the origin.

In addition, \eqref{eq:scattD} indicates that in the order for the limit $c\to +\infty$ to preserve the scattering length and set it equal to the non-relativistic scattering length \eqref{eq:scattS}, one must have
\begin{equation}\label{eq:scalinglimitc}
 \lim_{c\to +\infty}\frac{2 m c\gamma}{1+2\alpha}\,=\,\theta\,.
\end{equation}

Condition \eqref{eq:scalinglimitc} thus prescribes how to re-scale with $c$ either of the quantities $m$, $\alpha$, $\gamma$ to obtain in the non-relativistic limit an interaction with given $\theta$.

Conceptually, any re-scaling $m(c)$, $\alpha(c)$, $\gamma(c)$ which fulfills \eqref{eq:scalinglimitc} is compatible with the idea of a trajectory of artificial Hamiltonians (i.e., with non-physical values of the parameters) that preserve the scattering length of the interaction. In practice, it is technically more manageable to keep the mass $m$ and the magnitude $\alpha$ of the Aharonov-Bohm field as fixed parameters, and to only re-scale $\gamma$ as $\gamma(c)$.

\subsection{Main Results and further remarks}

The characterisation, using the Kre\u{\i}n-Vi\v{s}ik-Birman-Grubb self-adjoint extension scheme, of the two families of Aharonov-Bohm Hamiltonians $(H_\alpha^{(\gamma)})_{\gamma\in\mathbb{R}\cup\{\infty\}}$ and $(S_\alpha^{(\theta)})_{\theta\in\mathbb{R}\cup\{\infty\}}$ in the relativistic (Dirac-AB) and non-relativistic (Schr\"{o}dinger-AB) setting is, as mentioned, our first main result. It is stated in Theorems \ref{prop:Dirac} and \ref{prop:k0}, and the proofs are developed in Sections \ref{sec:Dirac} and \ref{sec:Schroe}.

Note in particular the explicit expressions \eqref{eq:mainres} and \eqref{eq:Scr0res} of the resolvents, a crucial information to take the limit in the resolvent sense.

Our second main result is the analysis of the non-relativistic limit in the Dirac-AB Hamiltonians $H_\alpha^{(\gamma)}$'s.

This can be regarded as a generalisation of the non-relativistic limit of the free Dirac operator $H_{\mathrm{free}}=-\ii \,c\, \boldsymbol{\sigma}\cdot \nabla + \sigma_3\, m c^2$  \eqref{freeDirac}, realised as a self-adjoint operator in $L^2(\mathbb{R}^2;\mathbb{C}^2,\ud x\ud y)$ (see \eqref{eq:Hilbsp} above) on the domain $H^1(\mathbb{R}^2;\mathbb{C}^2)$. In that well-known case (see, e.g., \cite[Corollary 6.2]{Thaller-Dirac-1992}), for any $z$ in the resolvent set $\rho(H_{\mathrm{free}} - mc^2)$ one also has $z\in\rho(S_{\mathrm{free}})$, the resolvent set of the free Schr\"{o}dinger operator $S_{\mathrm{free}}=-\frac{1}{2m}\Delta$ in $L^2(\mathbb{R}^2,\mathbb{C},\ud x \ud y)$ ($\Delta\equiv\frac{\ud^2}{\ud x^2}+\frac{\ud^2}{\ud y^2}$) with domain of self-adjointness $H^2(\mathbb{R}^2;\mathbb{C})$, and moreover, in operator norm,
\begin{equation}\label{eq:limfree}
\lim_{c \to +\infty}\big( H_{\mathrm{free}} - mc^2-z\big)^{-1}\,=\,
 \begin{pmatrix}
  (S_{\mathrm{free}}-z)^{-1} & \mathbb{O} \\ \mathbb{O} & \mathbb{O}
 \end{pmatrix}\,=\, (S_{\mathrm{free}}-\lambda)^{-1} \oplus \mathbb{O}\,.
\end{equation}
That is, $H_{\mathrm{free}} - mc^2\xrightarrow{\;c\to +\infty\;}S_{\mathrm{free}}\oplus\mathbb{O}$ in norm resolvent sense.\;

In this work we prove $H_{\alpha}^{(\gamma)} - mc^2\xrightarrow{\;c\to +\infty\;}S_{\alpha}^{(\theta)}\oplus\mathbb{O}$, thereby establishing a correspondence in the non-relativistic limit between Dirac-AB ($H_{\alpha}^{(\gamma)}$) and Schr\"{o}dinger-AB ($S_{\alpha}^{(\theta)}$) Hamiltonians, with the scaling \eqref{eq:scalinglimitc}, whose physical meaningfulness was discussed above.

\begin{theorem}\label{thm:Main}
 Let $\alpha\in(-\frac{1}{2},\frac{1}{2})$, $m>0$, $c>0$, $\theta\in\mathbb{R}\cup\{\infty\}$, and let $\gamma\equiv\gamma(c)\in\mathbb{R}\cup\{\infty\}$ scale with $c$ according to
	\begin{equation}\label{eq:LimitGammaTheta}
		\lim_{c \to +\infty} \frac{\,2\, m\, c \,\gamma(c)\,}{1+ 2\alpha} = \theta \, .
	\end{equation}
	Then there is a non-empty $c$-independent interval $J$ of the negative real line contained, eventually as $c\to +\infty$, in the infinite intersection $\bigcap_{c}\rho\big(H_{\alpha}^{(\gamma{(c)})}-mc^2\big)$ and also in $\rho\big(S_\alpha^{(\theta)}\big)$.
	Moreover, $\forall z\in J\cup\mathbb{C}^{\pm}$,
	\begin{equation}\label{eq:NonRelLimThmMain}
		\lim_{c \to +\infty} \big(H_{\alpha}^{(\gamma(c))}-mc^2-z)^{-1} = (S_\alpha^{(\theta)}  -z\big)^{-1}\oplus \mathbb{O}
	\end{equation}
	in operator norm.
\end{theorem}

 Here, as customary, $\mathbb{C}^\pm=\{z\in\mathbb{C}\,|\,\mathfrak{Im} z\gtrless 0\}$, and we content ourselves to indicate that the limit \eqref{eq:NonRelLimThmMain} can be taken at fixed complex non-real $z$, or also for $z$ in an eventually-$c$-fixed region $J$ of the negative real axis (further details in Proposition \ref{prop:myprop}).

Theorem \ref{thm:Main} validates the heuristics for the scaling \eqref{eq:scalinglimitc}: along the limit the magnitude of the relativistic Aharonov-Bohm interaction, measured through its scattering length, remains of the same order and attains the limit value of the prescribed non-relativistic Aharonov-Bohm interaction.

In this respect, Theorem \ref{thm:Main} covers these three scenarios.
\begin{enumerate}
 \item If $\gamma$ stays fixed in the limit, or more generally if $\gamma^{-1}=o(c)$ as $c\to +\infty$, then the non-relativistic limit of $H_\alpha^{(\gamma)}$ reproduces the standard Friedrichs realisation of the Aharonov-Bohm operator $S_\alpha^{(\infty)}$. The latter is therefore the non-relativistic `attractor' of all the Dirac-AB Hamiltonians at fixed $\gamma$.
 \item The regime $\gamma\sim c^{-1}$ reproduces, in the limit of $H_\alpha^{(\gamma)}$ as $c\to +\infty$, any of the $\theta$-labelled Schr\"{o}dinger-AB Hamiltonians, except for $S_\alpha^{(0)}$. The precise $\theta$ is the one satisfying \eqref{eq:LimitGammaTheta}.
 \item If $\gamma=o(c^{-1})$, then the non-relativistic limit of $H_\alpha^{(\gamma)}$ reproduces $S_\alpha^{(0)}$.
\end{enumerate}

Scenario (1) above includes, in particular, the case where $\gamma=\infty$ constantly in $c$: the non-relativistic limit of $H_{\alpha}^{(\infty)}$ is $S_\alpha^{(\infty)}$. This is the direct counterpart, in the presence of an Aharonov-Bohm magnetic field, of the non-magnetic limit \eqref{eq:limfree}. 

Thus, Theorem \ref{thm:Main} connects \emph{trajectories} of self-adjoint Dirac-AB Hamiltonians with Schr\"{o}\-dinger-AB Hamiltonians, this way mapping, in the non-relativistic limit, the whole family of Dirac-AB operators into a distinguished sub-family of Schr\"{o}\-dinger-AB operators.
It is spontaneous to refer to the latter as the sub-family of $s$-wave, angular-momentum-commuting, Schr\"{o}\-dinger-AB Hamiltonians with relativistic Dirac approximants (namely, admitting natural, block-wise relativistic Ahrnonov-Bohm approximants), since the non-triviality of self-adjoint realisations occurs in the respective (relativistic and non-relativistic) $k=0$ blocks. This leaves the interesting question of what status to attribute to the rest of the family of self-adjoint Schr\"{o}\-dinger-AB operators, from the point of view of the relativistic approximability.

We already argued about the distinguished status of the two supersymmetric Dirac-AB Hamiltonians $H_\alpha^{(\infty)}$ and $H_\alpha^{(0)}$ (Remark \ref{rem:commentsDiracExt}). As a matter of fact, these two operators are connected by positron-electron exchange symmetry. By exchanging $\alpha \leftrightarrow -\alpha$, $m \leftrightarrow -m$ and $\gamma \leftrightarrow \gamma^{-1}$, and going through the same steps of the proof of Theorem \ref{thm:Main}, one obtains the following positron-counterpart version of the non-relativistic limit.

\begin{theorem}\label{thm:Main2}
 Let $\alpha\in(-\frac{1}{2},\frac{1}{2})$, $m>0$, $c>0$, $\theta\in\mathbb{R}\cup\{\infty\}$, and let $\gamma\equiv\gamma(c)\in\mathbb{R}\cup\{\infty\}$ scale with $c$ according to
	\begin{equation}\label{eq:sc2}
		\lim_{c \to +\infty} \frac{2\, m \, c}{(1-2\alpha)\, \gamma(c) \,} \,=\,\theta\, .
	\end{equation}
	Then there is a non-empty $c$-independent interval $J$ of the negative real line contained, eventually as $c\to +\infty$, in the infinite intersection $\bigcap_{c}\rho\big(H_{\alpha}^{(\gamma{(c)})}+mc^2\big)$ and also in $\rho\big(-S_\alpha^{(\theta)}\big)$.
	Moreover, $\forall z\in J\cup\mathbb{C}^{\pm}$,
	\begin{equation}\label{eq:PositronLimit}
		\lim_{c \to +\infty} (H_{\alpha}^{(\gamma(c))}+mc^2-z)^{-1} = \mathbb{O}\oplus (-S_\alpha^{(\theta)}  -z)^{-1}
	\end{equation}
	in operator norm.
\end{theorem}

Thus, whereas the non-relativistic limit in the Dirac-AB Hamiltonian $H_{\alpha}^{(\infty)} $ with re-scaling \eqref{eq:LimitGammaTheta} yields the Friedrichs realisation $S_\alpha^{(\infty)}$ of the Schr\"odinger-AB Hamiltonian, here in the limit \eqref{eq:PositronLimit} with re-scaling \eqref{eq:sc2} the analogous role for the positron is played by the Dirac operator $H_{\alpha}^{(0)} $.


\begin{remark} One may conceive a sequence of \emph{two} limiting procedures: first, the removal of regularisation from a smoothed Dirac-AB Hamiltonian as in the above-mentioned work \cite{Tamura_rel_2003} (Remark \ref{rem:commentsDiracExt}), and then the non-relativistic limit described in Theorem \ref{thm:Main}. This gives rise to a generic case, when $\alpha\neq 0$, where the removal of the regularisation only selects $H_\alpha^{(\infty)}$ or $H_\alpha^{(0)}$, and hence the subsequent non-relativistic limit only produces $S_\alpha^{(\infty)}$ or $S_\alpha^{(0)}$, respectively. Next to this `rigid' case, there is the exceptional case $\alpha=0$ of half-integer flux: as recalled already (Remark \ref{rem:commentsDiracExt}), here the removal of regularisation can be tuned so as to yield in the limit $H_0^{(\gamma)}$ for suitable $\gamma=\gamma(c)$; from this, the further non-relativistic limit re-scaled with $c$ would then yield generic Schr\"{o}dinger-AB Hamiltonians $S_0^{(\theta)}$.
\end{remark}

 \begin{remark}
  We do not indulge here in the explicit spectral analysis of the Hamiltonians $H_\alpha^{(\gamma)}$ and $S_\alpha^{(\theta)}$ -- even though the discussion that follows shall make all basic features well evident, including the emergence of one isolated eigenvalue for $H_\alpha^{(\gamma)}$ in the gap, and for $S_\alpha^{(\theta)}$ below the continuum threshold. Yet, it is worth recalling that the limit in Theorem \ref{thm:Main} being in the norm resolvent sense, then a principle of non-contraction of the spectrum of $H_\alpha^{(\gamma)}-mc^2$ holds (see, e.g., \cite[Theorem VIII.23]{rs1}). In particular, if below the continuum threshold the operator $S_\alpha^{(\theta)}$ has an isolated eigenvalue, then that is the limit, as $c\to +\infty$, of the isolated eigenvalue in the gap of $H_\alpha^{(\gamma(c))}-mc^2$.
 \end{remark}

\subsection{Notation} All the general notation used here is standard in the field, and specific definitions are made in the following when special symbols or notions are introduced. For general reference, the $H^s_0$-Sobolev spaces used in our discussion are, as customary, the $H^s$-closures of the underlying $C^{\infty}_0$-spaces (the vanishing being declared at the origin). Thus, for our purposes,
 \begin{equation}
	H^1_0(\mathbb{R}^+;\mathbb{C}^2) \, = \, \{ \varphi \in L^2(\mathbb{R}^+;\mathbb{C}^2) \, | \, \varphi' \in L^2(\mathbb{R}^+;\mathbb{C}^2) \, \text{ and } \, \varphi(0) =0 \}
\end{equation}
and
\begin{equation}
		H^2_0(\mathbb{R}^+;\mathbb{C}) \, =\,\{ \varphi \in L^2(\mathbb{R}^+;\mathbb{C}) \, | \, \varphi', \varphi'' \in L^2(\mathbb{R}^+;\mathbb{C}) \, \text{ and } \, \varphi'(0)=\varphi(0)=0\}
	\end{equation}
(see, e.g., \cite[Definition 4.5(2) and Theorem 4.25]{Grubb-DistributionsAndOperators-2009}).

Let us also point out that the adjoint operators considered in the following are in fact adjoints of minimally defined differential operators and, as such, they are maximally defined and act with the same differential action, by standard general facts (see e.g.~\cite[Lemma 4.3]{Grubb-DistributionsAndOperators-2009}).

\section{Block-wise self-adjoint extensions of the Dirac-AB operator}\label{sec:Dirac}

 The object of this Section is the classification of the self-adjoint extensions in $L^2(\mathbb{R}^+;\mathbb{C}^2,\ud r)$ of the operators $\mathsf{h}_{\alpha,k}$, $k\in\mathbb{Z}$, introduced in \eqref{eq:Hsupersym} and minimally defined on the domain $C^\infty_0(\mathbb{R}^+;\mathbb{C}^2)$. The final result is Theorem \ref{prop:Dirac} here below.

 For convenience, we also introduce, for $\lambda\in(0,2mc^2)$, the short-hand expression
 \begin{equation}\label{eq:MuLambda}
	\mu_{c,\lambda}\,:=\, \frac{\sqrt{2mc^2-\lambda}}{c} 
\end{equation}
 and the two $\mathbb{R}^+\to\mathbb{C}^2$ functions
 \begin{align}
  \Phi_{\alpha,\lambda,c}^{(D)}(r)\,&:=\,\begin{pmatrix}-\frac{\ii c}{\sqrt{\lambda}} \mu_{c,\lambda} \sqrt{r} K_{\alpha+\frac12}(\mu_{c,\lambda}\sqrt{\lambda}\, r) \\
	\sqrt{r} K_{\alpha-\frac{1}{2}}(\mu_{c,\lambda} \sqrt{\lambda} \,r)\end{pmatrix}  \, , \label{eq:PhiD} \\
  F_{\alpha,\lambda,c}^{(D)}(r)\,&:=\,\begin{pmatrix}\frac{\ii c}{\sqrt{\lambda}} \mu_{c,\lambda} \sqrt{r} I_{\alpha+\frac12}(\mu_{c,\lambda} \sqrt{\lambda}\, r) \\
	\sqrt{r} I_{\alpha-\frac{1}{2}}(\mu_{c,\lambda} \sqrt{\lambda} \,r)\end{pmatrix}\label{eq:FD}
 \end{align}
 defined in terms of the Bessel functions of the second kind $K_\nu$ and $I_\nu$ \cite[Eq.~(9.6.2) and (9.6.10)]{Abramowitz-Stegun-1964}.

\begin{theorem}\label{prop:Dirac}
	Let $\alpha \in (-\frac{1}{2},\frac{1}{2})$. 
\begin{itemize}
	\item[(i)] For all $k\in\mathbb{Z}$ the operator $\mathsf{h}_{\alpha,k}$ is closable in $L^2(\mathbb{R}^+;\mathbb{C}^2,\ud r)$, with closure
	\begin{equation}
	 \mathcal{D}(\overline{\mathsf{h}_{\alpha,k}})\,=\,H^1_0(\mathbb{R}^+;\mathbb{C}^2)
	\end{equation}
	and spectral gap such that
	\begin{equation}
	 (-mc^2,mc^2)\,\subset\,\rho(\overline{\mathsf{h}_{\alpha,k}})\,.
	\end{equation}
    \item[(ii)] For all $k\in\mathbb{Z}\setminus\{0\}$, the operator $\mathsf{h}_{\alpha,k}$ is essentially self-adjoint in $L^2(\mathbb{R}^+;\mathbb{C}^2,\ud r)$.
	\item[(iii)] $\mathsf{h}_{\alpha,0}$ has deficiency indices $(1,1)$. Any function $g \in \mathcal{D}(\mathsf{h}_{\alpha,0}^*)$ satisfies the short-distance asymptotics
		\begin{equation}\label{eq:AsymptoticAdjoint}
			g(r)\,\stackrel{(r\downarrow 0)}{=} \,\begin{pmatrix}
				1  \\ 0
			\end{pmatrix} g_0 r^{-\alpha} + \begin{pmatrix}
				0 \\ 1
			\end{pmatrix} g_1 r^{\alpha} + o(r^{\frac{1}{2}})
		\end{equation}
		for $g$-dependent constants $g_0,g_1 \in \mathbb{C}$. The self-adjoint extensions of $\mathsf{h}_{\alpha,0}$ in $L^2(\mathbb{R}^+;\mathbb{C}^2,\ud r)$ constitute a one-real-parameter family $(\mathsf{h}_{\alpha,0}^{(\gamma)})_{\gamma \in \mathbb{R} \cup \{\infty\}}$ of restrictions of the adjoint $\mathsf{h}_{\alpha,0}^*$ each of which is defined, in terms of \eqref{eq:AsymptoticAdjoint}, by
		\begin{equation}\label{eq:BoundarySelfAdjointDirac}
			\begin{split}
				\mathsf{h}_{\alpha,0}^{(\gamma)} \,&:=\, \mathsf{h}_{\alpha,0}^* \upharpoonright \mathcal{D}(\mathsf{h}_{\alpha,0}^{(\gamma)}) \,,\\
				\mathcal{D}(\mathsf{h}_{\alpha,0}^{(\gamma)}) \,&:=\, \{ g \in \mathcal{D}(\mathsf{h}_{\alpha,0}^*) \, | \, g_1 =- \ii \gamma g_0\}\,.
			\end{split}
		\end{equation}
	\item[(iv)] Both $\mathsf{h}_{\alpha,0}^{(\infty)}$ and $\mathsf{h}_{\alpha,0}^{(\gamma)}$ have spectral gap, with $\rho(\mathsf{h}_{\alpha,0}^{(\infty)}-mc^2)\supset (-2mc^2,0)$ and with $\rho(\mathsf{h}_{\alpha,0}^{(\gamma)}-mc^2)$ containing $(-2mc^2,0)$ up to possibly a single point.
	For $-\lambda \in (-2mc^2,0)\subset\rho(\mathsf{h}_{\alpha,0}^{(\infty)}-mc^2)$, respectively,  $-\lambda \in \rho(\mathsf{h}_{\alpha,0}^{(\gamma)}-mc^2)\cap (-2mc^2,0)$, the extensions $\mathsf{h}_{\alpha,0}^{(\infty)}-mc^2$ and $\mathsf{h}_{\alpha,0}^{(\gamma)}-mc^2$, $\gamma\in\mathbb{R}$, have resolvent, respectively,
	\begin{equation}\label{eq:ThmReshinf}
	  (\mathsf{h}_{\alpha,0}^{(\infty)}-mc^2+\lambda)^{-1}\,=\,\textrm{integral operator with kernel $G_{\alpha,\lambda,c}^{(D)}(r,\rho)$}\,,
	\end{equation}
	where, for $r,\rho>0$,
  \begin{equation}\label{eq:GreenF}
	G_{\alpha,\lambda,c}^{(D)}(r,\rho)\,:=\,-\frac{\lambda}{\,c^2}\Big(\Phi_{\alpha,\lambda,c}^{(D)}(r) \otimes_{\mathbb{C}^2} \overline{F_{\alpha,\lambda,c}^{(D)}(\rho)} \mathbf{1}_{(0,r)}(\rho)+F_{\alpha,\lambda,c}^{(D)}(r) \otimes_{\mathbb{C}^2} \overline{ \Phi_{\alpha,\lambda,c}^{(D)}(\rho)} \mathbf{1}_{(r,+\infty)}(\rho)\Big)\,,
\end{equation}
  and
  \begin{equation}\label{eq:mainres}
   (\mathsf{h}_{\alpha,0}^{(\gamma)}-mc^2+\lambda)^{-1} \,=\, (\mathsf{h}_{\alpha,0}^{(\infty)}-mc^2+\lambda)^{-1}+ \tau_{\alpha,\lambda,\gamma,c}^{(D)}\,\frac{\lambda}{c^2}  \big| \Phi^{(D)}_{\alpha,\lambda,c} \big\rangle \big\langle \Phi_{\alpha,\lambda,c}^{(D)} \big|\, ,
  \end{equation}
 where
 \begin{equation}\label{eq:TauDGamma}
	\tau_{\alpha,\lambda,\gamma,c}^{(D)}\,:=\, \frac{2}{\,\pi \sec(\pi \alpha)+4^\alpha \mu_{c,\lambda}^2 (\mu_{c,\lambda}\sqrt{\lambda})^{-1-2 \alpha} \,\Gamma^2(\frac12+\alpha) c \gamma\,} \, .
\end{equation}
\end{itemize}
\end{theorem}

 Theorem \ref{prop:Dirac}(i)-(iii) states facts that are general knowledge in the non-magnetic case and are well expected in the Aharonov-Bohm setting as well.

 A version of Theorem \ref{prop:Dirac} in the framework of von Neumann's self-adjoint extension scheme is established in
 \cite[Sect.~II]{{Tamura_rel_2003}}. Here, instead, as announced, we obtain Theorem \ref{prop:Dirac} within the Kre\u{\i}n-Vi\v{s}ik-Birman-Grubb extension theory \cite{Krein-1947,Vishik-1952,Birman-1956,Grubb-1968}, for which we refer to the general discussion in \cite{GMO-KVB2017}, \cite[Chapter 2]{GM-SelfAdj_book-2022}, \cite[Chapter 13]{Grubb-DistributionsAndOperators-2009}, and \cite{KM-2015-Birman}, and to the short summary cast in Appendix \ref{sec:proofOfMain} specifically for the case of deficiency indices $(1,1)$.

This is a constructive procedure that goes through the following sequence of steps:
 \begin{enumerate}
  \item the identification of the operator closure $\overline{\mathsf{h}_{\alpha,k}}$ (Sect.~\ref{subsec:Closure});
  \item the characterisation of the canonical deficiency subspace $\ker(\mathsf{h}_{\alpha,k}^*-z)$ (Sect.~\ref{sec:KerHStar});
  \item the identification of a distinguished self-adjoint extension, actually $\mathsf{h}_{\alpha,k}^{(\infty)}$, such that $\mathsf{h}_{\alpha,k}^{(\infty)}-z$ is invertible with everywhere bounded inverse, and the characterisation of the action of $(\mathsf{h}_{\alpha,k}^{(\infty)}-z)^{-1}$ on $\ker(\mathsf{h}_{\alpha,k}^*-z)$ (Sect.~\ref{sec:distinguishedSA}).
 \end{enumerate}
 Once the above data $\overline{\mathsf{h}_{\alpha,k}}$, $\ker(\mathsf{h}_{\alpha,k}^*-z)$, and $(\mathsf{h}_{\alpha,k}^{(\infty)}-z)^{-1}$ are available, a general construction (see Theorem \ref{thm:KVB-General-App}) allows to obtain the whole family of self-adjoint extensions of $\mathsf{h}_{\alpha,k}$.

 We shall exploit this construction in Section \ref{subsect:SD}, after the necessary preparation developed in Sections \ref{subsec:Closure}-\ref{sec:distinguishedSA}.

\subsection{Operator closure and deficiency indices}\label{subsec:Closure}~

We start by characterising $\overline{\mathsf{h}_{\alpha,k}}$ and computing its deficiency indices.
The result is the following.

\begin{proposition}\label{prop:Closure}~
\begin{itemize}
	\item[(i)] For $|\alpha+k| \neq \frac{1}{2}$,
	\begin{equation}\label{eq:Closure}
		\mathcal{D}(\overline{\mathsf{h}_{\alpha,k}}) \;=\;H^1_0(\mathbb{R}^+;\mathbb{C}^2 )\, .
	\end{equation}
	In particular, for $|\alpha+k| \neq \frac{1}{2}$, any $\varphi \in \mathcal{D}(\overline{\mathsf{h}_{\alpha,k}})$ is absolutely continuous and satisfy
	\begin{equation}\label{eq:AsymptH1AtZero}
		\varphi(r) \overset{r \downarrow 0}{=} o(r^{\frac{1}{2}}) \, .
	\end{equation}
	\item[(ii)] $(-mc^2,mc^2) \subset \rho(\overline{\mathsf{h}_{\alpha,k}})$.
	\item[(iii)] The deficiency indices of $\mathsf{h}_{\alpha,k}$ are $(0,0)$ if $|k+\alpha| \geqslant \frac{1}{2}$, and $(1,1)$ otherwise.
	\end{itemize}
\end{proposition}

For the following manipulations, it is convenient to re-write
\begin{equation}\label{eq:h_rad_Pauli_Matrices}
	\mathsf{h}_{\alpha,k} \;=\; -\ii c \frac{\ud}{\ud r} \otimes \sigma_1-c \frac{\,\alpha+k\,}{r} \otimes \sigma_2 + mc^2 \mathbbm{1}\otimes \sigma_3
\end{equation}
with respect to the canonical Hilbert space isomorphism
\begin{equation}
 L^2(\mathbb{R}^+;\mathbb{C}^2,\ud r)\,\cong\, L^2(\mathbb{R}^+;\mathbb{C},\ud r) \otimes \mathbb{C}^2\,.
\end{equation}
We shall also use the short-hand
\begin{equation}
 \mathsf{h}_{\alpha,k}^{(m=0)}\,:=\,\mathsf{h}_{\alpha,k}-mc^2 \mathbbm{1}\otimes \sigma_3
\end{equation}
for the massless operator.

 In preparation for the proof of Proposition \ref{prop:Closure}, let us show that $\mathsf{h}_{\alpha,k}$ is coercive (Lemma \ref{lem:NormInequality}), that its graph norm
 \begin{equation}\label{eq:thegraphnormofhalphak}
  \|\mathsf{h}_{\alpha,k} g\|^2_{\mathsf{h}_{\alpha,k}}\,:=\,\|\mathsf{h}_{\alpha,k} g\|^2_{L^2(\mathbb{R}^+;\mathbb{C}^2)}+\|g\|^2_{L^2(\mathbb{R}^+;\mathbb{C}^2)}
 \end{equation}
 is equivalent to the $H^1$-norm (Lemma \ref{lem:Norm_Equivalence_Closure}), and that its deficiency indices are calculated with a limit-point limit-circle argument (Lemma \ref{lem:DefIndexClosure}).

\begin{lemma}\label{lem:NormInequality}
	For any $g \in C^\infty_0(\mathbb{R}^+;\mathbb{C}^2)$,
				\begin{equation}\label{eq:Closurenormid}
				\Vert \mathsf{h}_{\alpha,k} g \Vert^2_{L^2(\mathbb{R}^+;\mathbb{C}^2)} \;=\; \Vert \mathsf{h}_{\alpha,k}^{(m=0)} g \Vert^2_{L^2(\mathbb{R}^+;\mathbb{C}^2)} + m^2c^4 \Vert g \Vert^2_{L^2(\mathbb{R}^+;\mathbb{C}^2)}  \,.
			\end{equation}
    In particular,
    \begin{equation}\label{eq:ClosureInequality}
				\Vert \mathsf{h}_{\alpha,k} g \Vert_{L^2(\mathbb{R}^+;\mathbb{C}^2)} \; \geqslant \; m c^2 \Vert g \Vert_{L^2(\mathbb{R}^+;\mathbb{C}^2)} \, .
			\end{equation}
\end{lemma}

\begin{proof}
	Using \eqref{eq:h_rad_Pauli_Matrices},
	\[
		\begin{split}
			\langle \mathsf{h}_{\alpha,k} g &, \mathsf{h}_{\alpha,k} g \rangle_{L^2(\mathbb{R}^+;\mathbb{C}^2)} \,=\, \Vert \mathsf{h}_{\alpha,k}^{(m=0)} g \Vert^2_{L^2(\mathbb{R}^+;\mathbb{C})\otimes\mathbb{C}^2} \\
			&+ m c^2 \Big(\langle \mathsf{h}_{\alpha,k}^{(m=0)} g, (\mathbbm{1}\otimes \sigma_3)  g \rangle_{L^2(\mathbb{R}^+;\mathbb{C}) \otimes \mathbb{C}^2} +  \langle (\mathbbm{1}\otimes \sigma_3) g, \mathsf{h}_{\alpha,k}^{(m=0)} g \rangle_{L^2(\mathbb{R}^+;\mathbb{C})\otimes \mathbb{C}^2}\Big) \\
			&+ m^2 c^4 \Vert g \Vert^2_{L^2(\mathbb{R}^+;\mathbb{C})\otimes\mathbb{C}^2}\,.
		\end{split}
	\]
 As a consequence of the standard anti-commutation relations
 \[
  \sigma_j\sigma_k+\sigma_k\sigma_j\,=\,2\,\delta_{j,k}\begin{pmatrix} 1 & 0 \\ 0 & 1 \end{pmatrix},\qquad j,k\in\{1,2,3\}\,,
 \]
 the second summand in the right-hand side above vanishes. This establishes \eqref{eq:Closurenormid}.
\end{proof}

\begin{lemma}\label{lem:Norm_Equivalence_Closure}
Let $m>0$. For any $g \in C^\infty_0(\mathbb{R}^+;\mathbb{C}^2)$ and $|\alpha+k| \neq \frac{1}{2}$,
\begin{equation}
	\Vert g \Vert_{\mathsf{h}_{\alpha,k}} \; \sim \Vert g \Vert_{H^1(\mathbb{R}^+;\mathbb{C}^2)}
\end{equation}
in the sense of equivalence of norms.
\end{lemma}

\begin{proof}
    Owing to \eqref{eq:thegraphnormofhalphak} and \eqref{eq:ClosureInequality}, $\Vert \cdot \Vert_{\mathsf{h}_{\alpha,k}}\sim\Vert \mathsf{h}_{\alpha,k} \cdot \Vert_{L^2(\mathbb{R}^+;\mathbb{C}^2)}$ (equivalence of norms). Then it suffices to prove that $\Vert \mathsf{h}_{\alpha,k} \cdot \Vert_{L^2(\mathbb{R}^+;\mathbb{C}^2)}\sim\|\cdot\|_{H^1(\mathbb{R}^+;\mathbb{C}^2)}$.

    Using \eqref{eq:Closurenormid} one finds
	\[
		\begin{split}
			\Vert &\mathsf{h}_{\alpha,k} g \Vert_{L^2(\mathbb{R}^+;\mathbb{C}^2)}^2 \,=\, \Vert \mathsf{h}_{\alpha,k}^{(m=0)} g \Vert_{L^2(\mathbb{R}^+;\mathbb{C}^2)}^2 + m^2 c^4 \Vert g \Vert_{L^2(\mathbb{R}^+;\mathbb{C}^2)}^2 \\
			&=\,c^2\Vert g' \Vert^2_{L^2(\mathbb{R}^+;\mathbb{C}^2)}+c^2 \left\Vert \frac{\alpha+k}{r} g \right\Vert^2_{L^2(\mathbb{R}^+;\mathbb{C}^2)}+m^2 c^4\Vert g \Vert_{L^2(\mathbb{R}^+;\mathbb{C}^2)}^2 \\
			&\qquad -c^2 (\alpha+k) \Big( \Big\langle (r^{-1}\otimes \sigma_2) g,\big(- \ii \frac{\ud}{\ud r} \otimes \sigma_1\big) g \Big\rangle_{L^2(\mathbb{R}^+;\mathbb{C}) \otimes \mathbb{C}^2} \\
			& \qquad\qquad\qquad \qquad\qquad+ \Big\langle \big(- \ii \frac{\ud}{\ud r} \otimes \sigma_1\big) g,\big( r^{-1}\otimes \sigma_2 \big) g \Big\rangle_{L^2(\mathbb{R}^+;\mathbb{C}) \otimes \mathbb{C}^2} \Big)\, .
		\end{split}
	\]
	Concerning the last term in the right-hand side above, the self-adjointness of the Pauli matrices, the identity $\sigma_2\sigma_1=-\ii\sigma_3$, and integration by parts for compactly-supported-in-$\mathbb{R}^+$ $g$'s, yield
	\[
		\begin{split}
		 \Big\langle &\big( r^{-1}\otimes \sigma_2\big) g,\Big(- \ii \frac{\ud}{\ud r} \otimes \sigma_1\Big) g \Big\rangle_{L^2(\mathbb{R}^+;\mathbb{C}) \otimes \mathbb{C}^2}+\Big\langle \Big(- \ii \frac{\ud}{\ud r} \otimes \sigma_1\Big) g,\big( r^{-1}\otimes \sigma_2\big) g \Big\rangle_{L^2(\mathbb{R}^+;\mathbb{C}) \otimes \mathbb{C}^2} \\
		 &=\,\Big\langle \big(r^{-1}\otimes \mathbbm{1}\big) g,\Big(- \ii \frac{\ud}{\ud r} \otimes \sigma_2 \sigma_1\Big)  g \Big\rangle_{L^2(\mathbb{R}^+;\mathbb{C}) \otimes \mathbb{C}^2}+\Big\langle \Big(- \ii \frac{\ud}{\ud r} \otimes \sigma_2 \sigma_1\Big) g,\big(r^{-1}\otimes \mathbbm{1}\big) g \Big\rangle_{L^2(\mathbb{R}^+;\mathbb{C}) \otimes \mathbb{C}^2} \\
		 &=\,\Big\langle \big(r^{-1}\otimes \ii\sigma_3\big) g, \Big(-\ii\frac{\ud}{\ud r} \otimes \mathbbm{1}\Big) g \Big\rangle_{L^2(\mathbb{R}^+;\mathbb{C}) \otimes \mathbb{C}^2}+\Big\langle \Big(-\ii\frac{\ud}{\ud r} \otimes \mathbbm{1}\Big) g,\big(r^{-1}\otimes \ii\sigma_3\big) g \Big\rangle_{L^2(\mathbb{R}^+;\mathbb{C}) \otimes \mathbb{C}^2} \\
		 &=\, \Big\langle -\ii\frac{\ud}{\ud r} \big(r^{-1}\otimes \ii\sigma_3\big) g,g\Big\rangle_{L^2(\mathbb{R}^+;\mathbb{C}) \otimes \mathbb{C}^2}-\Big\langle \big(r^{-1}\otimes \ii\sigma_3\big) \Big(-\ii\frac{\ud}{\ud r} \otimes \mathbbm{1}\Big) g,g\Big\rangle_{L^2(\mathbb{R}^+;\mathbb{C}) \otimes \mathbb{C}^2}\\
		 &=\,-\big\langle (r^{-2}\otimes\sigma_3)g,g\big\rangle_{L^2(\mathbb{R}^+;\mathbb{C})\otimes \mathbb{C}^2} \,,
		\end{split}
	\]
	whence
	\[
		\begin{split}
			\Vert \mathsf{h}_{\alpha,k} g \Vert_{L^2(\mathbb{R}^+;\mathbb{C}^2)}^2 \,&=\,c^2\Vert g' \Vert^2_{L^2(\mathbb{R}^+;\mathbb{C}^2)}+c^2 \left\Vert \frac{\alpha+k}{r} g \right\Vert^2_{L^2(\mathbb{R}^+;\mathbb{C}^2)}+m^2 c^4\Vert g \Vert_{L^2(\mathbb{R}^+;\mathbb{C}^2)}^2 \\
				&\qquad + c^2 (\alpha+k) \left\langle (r^{-1} \otimes \sigma_3) g, (r^{-1} \otimes \mathbbm{1}) g \right\rangle_{L^2(\mathbb{R}^+;\mathbb{C}^2)} \, .
		\end{split}
	\]
	Then, using Hardy's inequality $\Vert r^{-1} g \Vert_{L^2(\mathbb{R}^+;\mathbb{C})} \leqslant 2 \Vert g' \Vert_{L^2(\mathbb{R}^+;\mathbb{C})}$ and the Cauchy-Schwarz inequality,
		\[
		\begin{split}
		 		\Vert \mathsf{h}_{\alpha,k} g \Vert_{L^2(\mathbb{R}^+;\mathbb{C}^2)}^2 \,&\leqslant\, c^2 (4(\alpha+k)^2+4|\alpha+k|+1) \Vert g' \Vert_{L^2(\mathbb{R}^+;\mathbb{C}^2)}^2 + m^2 c^4 \Vert g \Vert^2_{L^2(\mathbb{R}^+;\mathbb{C}^2)} \\
		 		&\lesssim \|g\|^2_{H^1(\mathbb{R}^+;\mathbb{C}^2)}\,.
		\end{split}
	\]

 To establish the converse bound, the Cauchy-Schwarz inequality yields
 \[
  \begin{split}
   \big| (\alpha+k) &\big\langle (r^{-1} \otimes \sigma_3) g, (r^{-1} \otimes \mathbbm{1}) g \big\rangle_{L^2(\mathbb{R}^+;\mathbb{C}^2)}\big| \,\leqslant\,|\alpha+k| \Vert r^{-1} g \Vert^2_{L^2(\mathbb{R}^+;\mathbb{C}^2)}\,,
  \end{split}
 \]
	whence
 \[
  \begin{split}
   \Vert \mathsf{h}_{\alpha,k} &g  \Vert_{L^2(\mathbb{R}^+;\mathbb{C}^2)}^2 \\
   &\geqslant\,c^2 \Vert g' \Vert^2_{L^2(\mathbb{R}^+;\mathbb{C}^2)} + c^2|\alpha+k|(|\alpha+k|-1) \Vert r^{-1} g \Vert^2_{L^2(\mathbb{R}^+;\mathbb{C}^2)} + m^2 c^4 \Vert g \Vert^2_{L^2(\mathbb{R}^+;\mathbb{C}^2)}\,.
  \end{split}
 \]
  One then checks two cases separately. If $|\alpha+k| \geqslant 1$ or $\alpha+k=0$, then the second summand in the right-hand side above is non-negative, and therefore
  \[
   \Vert \mathsf{h}_{\alpha,k} g  \Vert_{L^2(\mathbb{R}^+;\mathbb{C}^2)}^2 \,\geqslant\,c^2 \Vert g' \Vert^2_{L^2(\mathbb{R}^+;\mathbb{C}^2)} +m^2 c^4 \Vert g \Vert^2_{L^2(\mathbb{R}^+;\mathbb{C}^2)}\,\sim\,\|g\|^2_{H^1(\mathbb{R}^+;\mathbb{C}^2)}\,.
  \]
  If instead $0<|\alpha+k| < 1$ and $|\alpha+k| \neq \frac{1}{2}$, then by Hardy's inequality
	\[
		\Vert \mathsf{h}_{\alpha,k} g \Vert_{L^2(\mathbb{R}^+;\mathbb{C}^2)}^2 \,\geqslant\, c^2 (1-2 |\alpha+k|)^2 \Vert g' \Vert^2_{L^2(\mathbb{R}^+;\mathbb{C}^2)} + m^2 c^4 \Vert g \Vert^2_{L^2(\mathbb{R}^+;\mathbb{C}^2)} \,\sim\,\|g\|^2_{H^1(\mathbb{R}^+;\mathbb{C}^2)}\,.
	\]
  In either case,
  \[
   \Vert \mathsf{h}_{\alpha,k} g \Vert_{L^2(\mathbb{R}^+;\mathbb{C}^2)}^2 \,\gtrsim\,\|g\|^2_{H^1(\mathbb{R}^+;\mathbb{C}^2)}\,.
  \]
 The norm equivalence $\Vert \cdot \Vert_{\mathsf{h}_{\alpha,k}}\sim\Vert \mathsf{h}_{\alpha,k} \cdot \Vert_{L^2(\mathbb{R}^+;\mathbb{C}^2)}\sim\|\cdot\|_{H^1(\mathbb{R}^+;\mathbb{C}^2)}$ is thus established.
\end{proof}

\begin{lemma}\label{lem:DefIndexClosure}
The operator $\mathsf{h}_{\alpha,k}$ is
\begin{itemize}
 \item[(i)] in the limit point case at $r=+\infty$ for any $k \in \mathbb{Z}$ and $\alpha \in (-\frac{1}{2},\frac{1}{2})$;
 \item[(ii)] in the limit point case at $r=0$ if $|k+\alpha| \geqslant \frac{1}{2}$,
 \item[(iii)] in the limit-circle case at $r=0$ if $|k+\alpha| < \frac{1}{2}$.
\end{itemize}
\end{lemma}

\begin{proof}
These are classical facts concerning the square-integrability at the considered edge (zero or infinity) of \emph{all} the solutions to
\begin{equation*}
	  \begin{pmatrix}
	mc^2 & \ii c \left( -\frac{\ud}{\ud r}+\frac{\alpha+k}{r}\right) \\
	-\ii c \left(\frac{\ud}{\ud r} + \frac{\alpha+k}{r} \right) & -mc^2
	\end{pmatrix} \psi\,=\,\pm\ii\,\psi
\end{equation*}
(limit-circle case), or lack of thereof for a unique solution, up to multiplicative constant (limit-point case). Parts (i) and (ii) are discussed, e.g., in \cite[Corollary to Theorem 6.8]{Weidmann-book1987} and \cite[Theorem 6.9]{Weidmann-book1987}, respectively. Then part (iii) follows from \cite[Theorem 6.9]{Weidmann-book1987} and Weyl's alternative \cite[Theorem 5.6]{Weidmann-book1987}.
\end{proof}

\begin{proof}[Proof of Proposition \ref{prop:Closure}]
One has
\[
 \begin{split}
  \mathcal{D}(\overline{\mathsf{h}_{\alpha,k}})\,&=\,\overline{\mathcal{D}(\mathsf{h}_{\alpha,k})\,}^{\|\cdot\|_{\mathsf{h}_{\alpha,k}}}\,=\,
\overline{C^\infty_0(\mathbb{R}^+;\mathbb{C}^2)\,}^{\|\cdot\|_{\mathsf{h}_{\alpha,k}}}\,=
\,\overline{C^\infty_0(\mathbb{R}^+;\mathbb{C}^2)\,}
^{\|\cdot\|_{H^1(\mathbb{R}^+;\mathbb{C}^2)}} \\
&=\,H^1_0(\mathbb{R}^+;\mathbb{C}^2)\,,
 \end{split}
\]
the third equality following from Lemma \ref{lem:Norm_Equivalence_Closure} under the constraint $|\alpha+k|\neq\frac{1}{2}$.

The asymptotics \eqref{eq:AsymptH1AtZero} then follows from
\[
 \begin{split}
  \|\psi(r)\|_{\mathbb{C}^2}\,&=\,\Big\| \int_{0}^r\psi'(\rho)\,\ud \rho\,\Big\|_{\mathbb{C}^2}\,\leqslant\, \int_{0}^r\|\psi'(\rho)\|_{\mathbb{C}^2}\,\ud \rho\,\leqslant\,r^{\frac{1}{2}}\,\|\psi'\|_{L^2((0,r);\mathbb{C}^2)}\,=\,r^{\frac{1}{2}} o(1)\,=\,o(r^{\frac{1}{2}})\,.
 \end{split}
\] 

Concerning part (ii), \eqref{eq:ClosureInequality} from Lemma \ref{lem:NormInequality} implies that for $\lambda \in(-m c^2 ,mc^2)$ the operator $\mathsf{h}_{\alpha,0}-\lambda$ is invertible on its range and with bounded inverse, with bound
\[
	\Vert (\mathsf{h}_{\alpha,0}-\lambda \mathbbm{1})^{-1} \Vert_{L^2(\mathbb{R}^+;\mathbb{C}^2)\to L^2(\mathbb{R}^+;\mathbb{C}^2)} \,\leqslant\, \frac{1}{\,m c^2-|\lambda|\,} \, .
\]
 Since $\overline{(\mathsf{h}_{\alpha,0}-\lambda \mathbbm{1})^{-1}}=(\overline{\mathsf{h}_{\alpha,0}}-\lambda)^{-1}$ \cite[Theorem 1.8(v)]{schmu_unbdd_sa}, then also
\[
	\Vert (\overline{\mathsf{h}_{\alpha,0}}-\lambda \mathbbm{1})^{-1} \Vert_{L^2(\mathbb{R}^+;\mathbb{C}^2)\to L^2(\mathbb{R}^+;\mathbb{C}^2)} \,\leqslant\, \frac{1}{\,m c^2-|\lambda|\,} \, ,
\]
 which proves that $(-mc^2,mc^2) \subset \rho(\overline{\mathsf{h}_{\alpha,k}})$.

 Last, part (iii) follows from Lemma \ref{lem:DefIndexClosure} in view of the standard relations between deficiency indices and limit-point/limit-circle cases (see, e.g. \cite[Theorem 5.7]{Weidmann-book1987}).
\end{proof}

\subsection{Characterisation of $\ker (\mathsf{h}_{\alpha,0}^*-mc^2+\lambda)$}\label{sec:KerHStar}~

 The self-adjoint extension problem of $\mathsf{h}_{\alpha,k}$ is non-trivial only for $k=0$. This follows at once from Proposition \ref{prop:Closure}(iii) or, alternatively, can be deduced from the fact that $(-mc^2,mc^2)$ is contained in the resolvent set of $\overline{\mathsf{h}_{\alpha,0}}$ (Proposition \ref{prop:Closure}(ii)), because then for any $z \in (-mc^2,mc^2)$ there exists a self-adjoint extension of $\mathsf{h}_{\alpha,0}$ with $z$ in its resolvent set \cite[Theorem 2]{Calkin-1940}.

 So, $\mathsf{h}_{\alpha,0}-z$ admits for $z \in (-mc^2,mc^2)$ non-trivial self-adjoint extensions with everywhere defined and bounded inverse. As a consequence, for any $\lambda \in (0,2mc^2)$ the operator $\mathsf{h}_{\alpha,0}-mc^2+\lambda$ admits non-trivial self-adjoint extensions, and precisely $\infty^1$ extensions, since the deficiency indices of $\mathsf{h}_{\alpha,0}-mc^2+\lambda$ and $\mathsf{h}_{\alpha,0}$ are the same and are equal to $(1,1)$ (Proposition \ref{prop:Closure}(iii)).

 This is also equivalent to the fact that the deficiency subspace $\ker (\mathsf{h}_{\alpha,0}^*-mc^2+\lambda)$ is \emph{one-dimensional}. Here we characterise such subspace in the following Proposition.

\begin{proposition}\label{prop:KernelAdjoint} Let $\lambda \in (0,2mc^2)$. Then
\begin{equation}
\ker (\mathsf{h}_{\alpha,0}^*-mc^2+\lambda) \,=\, \mathrm{span}_{\mathbb{C}} \{\Phi_{\alpha,\lambda,c}^{(D)}\} \, .
\end{equation}
\end{proposition}

 Proposition \ref{prop:KernelAdjoint} follows from standard ODE analysis and the following asymptotics for the special functions $ \Phi_{\alpha,\lambda,c}^{(D)}$ and $ F_{\alpha,\lambda,c}^{(D)}$ introduced in \eqref{eq:PhiD}-\eqref{eq:FD}.

 \begin{lemma}\label{lem:asymptoticsPhiF}
  Let $\lambda \in (0,2mc^2)$. Then $\Phi_{\alpha,\lambda,c}^{(D)}$ and $F_{\alpha,\lambda,c}^{(D)}$ are smooth on $\mathbb{R}^+$ with asymptotic expansions
  \begin{align}
   \Phi_{\alpha,\lambda,c}^{(D)}(r)\,& \stackrel{(r\to +\infty)}{=} \, \begin{pmatrix} -\ii \\ 1 \end{pmatrix} e^{-r\,\mu_{c,\lambda}\sqrt{\lambda}}(1+O(r^{-\frac{1}{2}}))\,, \label{eq:AsymptPhiInfinity}\\
	F_{\alpha,\lambda,c}^{(D)}(r)\,&\stackrel{(r\to +\infty)}{=}\,\begin{pmatrix}
		\ii \\ 1
	\end{pmatrix} e^{r\,\mu_{c,\lambda}\sqrt{\lambda}}(1+O(r^{-\frac{1}{2}}))\,,\label{eq:AsymptFInfinity}
  \end{align}
 and
 \begin{align}
  	\Phi_{\alpha,\lambda,c}^{(D)}(r) \,&\stackrel{(r\downarrow 0)}{=}\,  \begin{pmatrix} 1 \\ 0 \end{pmatrix} A_{\alpha,\lambda,c}\,r^{-\alpha}+ \begin{pmatrix} 0 \\ 1 \end{pmatrix} B_{\alpha,\lambda,c}\,r^{\alpha} + O(r^{1-|\alpha|}) \, , \label{eq:AsymptZeroPhi}\\
	F_{\alpha,\lambda,c}^{(D)}(r)\,&\stackrel{(r\downarrow 0)}{=}\, \begin{pmatrix} 0 \\ 1 \end{pmatrix} C_{\alpha,\lambda,c}\,r^{\alpha} + O(r^{2+\alpha}) \, , \label{eq:AsymptZeroFLong}
 \end{align}
 where
  \begin{align}
   A_{\alpha,\lambda,c}\,&=\,-\frac{ \,\ii\, 2^{\alpha-\frac12} c \,\mu_{c,\lambda}( \mu_{c,\lambda}\,\sqrt{\lambda})^{-\frac12-\alpha} \,\Gamma(\frac12+\alpha)}{\sqrt{\lambda}} \, , \label{eq:Phi0Def}\\
 	B_{\alpha,\lambda,c}\,&=\,2^{-\frac12-\alpha}( \mu_{c,\lambda}\,\sqrt{\lambda})^{\alpha-\frac12} \,\Gamma({\textstyle \frac12-\alpha}) \, , \label{eq:Phi1Def}\\
 	C_{\alpha,\lambda,c}\,&=\, \frac{\,2^{\frac12-\alpha} (\mu_{c,\lambda}\,\sqrt{\lambda} )^{-\frac12+\alpha}}{\Gamma(\frac12+\alpha)} \, .\label{eq:F0Def}
  \end{align}
 Moreover,
 \begin{equation}\label{eq:L2normPhi}
	\Vert \Phi_{\alpha,\lambda,c}^{(D)} \Vert_{L^2(\mathbb{R}^+; \mathbb{C}^2)}^2\,=\,\frac{\pi (1+2 \alpha)}{\,4 \lambda^2 \cos(\pi \alpha)\,}c^2+\frac{\pi(1-2\alpha)}{\,4 \lambda \mu_{c,\lambda}^2 \cos(\pi \alpha)\,} \, .
\end{equation}
 \end{lemma}

 \begin{proof}
  All straightforward computations directly following from the definition of $ \Phi_{\alpha,\lambda,c}^{(D)}$ and $ F_{\alpha,\lambda,c}^{(D)}$ in terms of the Bessel functions of second kind $K_\nu$ and $I_\nu$ (see \eqref{eq:PhiD}-\eqref{eq:FD} above), and from the asymptotic expansions of the latter special functions, specifically \cite[(9.7.1)-(9.7.2)]{Abramowitz-Stegun-1964} for $r\to+\infty$ and \cite[(9.6.2) and (9.6.10)]{Abramowitz-Stegun-1964} for $r \downarrow 0$.
 \end{proof}

 \begin{proof}[Proof of Proposition \ref{prop:KernelAdjoint}]
  By general facts (see, e.g., \cite[Sect.~4.1]{Grubb-DistributionsAndOperators-2009}), $\mathsf{h}_{\alpha,0}^*$ has the same differential action as $\mathsf{h}_{\alpha,0}$ and therefore the eigenvalue problem
  \[\tag{i}\label{eq:evproblemstar}
   (\mathsf{h}_{\alpha,0}^*-mc^2+\lambda)u\,=\,0\,,\qquad u\in\mathcal{D}(\mathsf{h}_{\alpha,0}^*)\,\subset\,L^2(\mathbb{R}^+;\mathbb{C}^2,\ud r)\,,
  \]
  is solved by selecting the square-integrable solutions to the differential problem
  \begin{equation*}\tag{ii}\label{eq:DiffEqProof}
	\begin{pmatrix}
		\lambda & -\ii c \big(\frac{\ud}{\ud r}-\frac{\alpha}{r} \big) \\
		-\ii c \big( \frac{\ud}{\ud r} + \frac{\alpha}{r} \big) & - 2 mc^2 +\lambda
	\end{pmatrix} u\,=\,0 \, .
\end{equation*}

  Denote for convenience
  \[
   u\,\equiv\,\begin{pmatrix} u_+ \\ u_- \end{pmatrix}\in L^2(\mathbb{R}^+;\mathbb{C},\ud r)\oplus L^2(\mathbb{R}^+;\mathbb{C},\ud r)\,,
  \]
 then \eqref{eq:DiffEqProof} reads
 \[\tag{iii}\label{eq:SystemUU}
  \begin{cases}
   \;u_+\,=\, \frac{\ii c}{\lambda} u_-'-\frac{\ii c}{\lambda} \frac{\alpha}{r} u_-\,, \\
   \;-\ii c u_+'-\frac{\ii c \alpha}{r} u_+ - 2 mc^2 u_-+\lambda u_- \, = \, 0\,.
  \end{cases}
 \]
 By substitution of $u_+$,
 \[\tag{iv}\label{eq:EquationForu-}
 r^2 u_-''(r)+\Big(\Big(\frac{\lambda^2}{c^2} -2m \lambda\Big) r^2 + c^2(\alpha-\alpha^2)\Big) u_-(r)\,=\,0 \, .
\]

  The change of variable
  \[
   \xi\,:=\,r\,\mu_{c,\lambda} \sqrt{\lambda}\,,\qquad u_-(r)\,\equiv\,\sqrt{r} v(r\,\mu_{c,\lambda} \sqrt{\lambda})
  \]
 yields
 \begin{equation*}
\xi^2 \frac{\ud^2 v}{\ud \xi^2}+\xi \frac{\ud v}{\ud \xi}- \Big(\Big(\frac{1}{2}-\alpha \Big)^2+\xi^2 \Big) v \,=\, 0\,,
\end{equation*}
 a modified Bessel equation \cite[Eq.~(9.6.1)]{Abramowitz-Stegun-1964} whose two-dimensional space of solutions is spanned by the modified Bessel functions $I_{\alpha-\frac{1}{2}}(\xi)$ and $K_{\alpha-\frac{1}{2}}(\xi)$ \cite[(9.6.10) and (9.6.2)]{Abramowitz-Stegun-1964}. Then the two-dimensional space of solutions to \eqref{eq:EquationForu-} is spanned by
 \[
  \sqrt{r} K_{\alpha-\frac{1}{2}}(r\,\mu_{c,\lambda} \sqrt{\lambda})\qquad\textrm{and}\qquad\sqrt{r} I_{\alpha-\frac{1}{2}}(r\,\mu_{c,\lambda} \sqrt{\lambda})\,.
 \]

 The choice $u_-(r)=\sqrt{r} I_{\alpha-\frac{1}{2}}(r\,\mu_{c,\lambda} \sqrt{\lambda})$ yields, through the first of \eqref{eq:SystemUU},
 \[
   u_+(r)\,=\,\frac{\,\ii c \mu_{c,\lambda} \,}{\sqrt{\lambda}} \sqrt{r} \Big(I'_{\alpha-\frac{1}{2}}(r\,\mu_{c,\lambda} \sqrt{\lambda})+\frac{\frac{1}{2}-\alpha }{\,r\,\mu_{c,\lambda} \sqrt{\lambda}\,} \,I_{\alpha-\frac{1}{2}}(\,r\,\mu_{c,\lambda} \sqrt{\lambda}) \Big)\,.
 \]
  By means of the recurrence relations \cite[Eq.~(9.6.26)-(iv)]{Abramowitz-Stegun-1964}
 \[
  I'_{\nu}(z)-\frac{\nu}{z}I_{\nu}(z)\,=\,I_{\nu+1}(z)\,,
 \]
 the latter expression for $u_+$ is re-written as
  \[
   u_+(r)\,=\,\frac{\,\ii c \mu_{c,\lambda} \,}{\sqrt{\lambda}} \sqrt{r}\, I_{\alpha+\frac{1}{2}}(r\,\mu_{c,\lambda} \sqrt{\lambda})\,.
  \]
 By comparison to \eqref{eq:FD}, the above argument shows that $F_{\alpha,\lambda,c}^{(D)}$ is a solution to \eqref{eq:DiffEqProof}.

  The alternative choice $u_-(r)=\sqrt{r}K_{\alpha-\frac{1}{2}}(r\,\mu_{c,\lambda} \sqrt{\lambda})$ can be discussed in a completely analogous manner, using now the recurrence relations \cite[Eq.~(9.6.26)-(iv)]{Abramowitz-Stegun-1964}
  \[
   K_{\nu+1}(z)\,=\,\frac{\nu}{z} K_\nu(z)-K'_\nu(z)\,.
  \]
  This way, one proves that also $\Phi_{\alpha,\lambda,c}^{(D)}$ is a solution to \eqref{eq:DiffEqProof}.

  The linear independence between $\Phi_{\alpha,\lambda,c}^{(D)}$ and $F_{\alpha,\lambda,c}^{(D)}$ is obvious. This completes the characterisation of the space of solutions to \eqref{eq:DiffEqProof}.

  Owing to Lemma \ref{lem:asymptoticsPhiF}, only $\Phi_{\alpha,\lambda,c}^{(D)}\in L^2(\mathbb{R}^+;\mathbb{C}^2,\ud r)$, which finally proves that the solutions to the eigenvalue problem \eqref{eq:evproblemstar} are the multiples of $\Phi_{\alpha,\lambda,c}^{(D)}$.
 \end{proof}

 \begin{remark}
  Observe from \eqref{eq:L2normPhi} of Lemma \ref{lem:asymptoticsPhiF} that
  \[
   \lim_{\alpha \to \pm \frac{1}{2}} \Vert \Phi_{\alpha,\lambda,c}^{(D)} \Vert_{L^2(\mathbb{R}^+; \mathbb{C}^2)} \,=\, +\infty\,.
  \]
 This is consistent with Propositions \ref{prop:Closure} and \ref{prop:KernelAdjoint} above: indeed, when $\alpha=\pm \frac{1}{2}$ the deficiency indices of $\mathsf{h}_{\alpha,0}$ are $(0,0)$ (Proposition \ref{prop:Closure}(iii)), meaning that $\ker (\mathsf{h}_{\alpha,0}^*-mc^2+\lambda)$ is trivial.
 \end{remark}

 \subsection{Distinguished self-adjoint extension}\label{sec:distinguishedSA}~

 The next step of our analysis is to identify a distinguished self-adjoint extension of $ \mathsf{h}_{\alpha,0}$, \emph{temporarily} denoted by $ \mathsf{h}_{\alpha,0}^{(D)}$ (eventually, $ \mathsf{h}_{\alpha,0}^{(\infty)}$), such that $\mathsf{h}_{\alpha,0}^{(D)}-mc^2+\lambda$ has everywhere defined and bounded inverse.

 \begin{proposition}\label{prop:hD}
   Let $\alpha\in(-\frac{1}{2},\frac{1}{2})$ and $\lambda\in(0,2mc^2)$.
 Let $R_{G_{\alpha,\lambda,c}^{(D)}}$ be the integral operator acting on $\mathbb{R}^+\to\mathbb{C}^2$ functions with the integral kernel $G_{\alpha,\lambda,c}^{(D)}(r,\rho)$ defined in \eqref{eq:GreenF}.
  \begin{enumerate}
  \item[(i)] $R_{G_{\alpha,\lambda,c}^{(D)}}$  is everywhere defined, bounded, and self-adjoint in $L^2(\mathbb{R}^+;\mathbb{C}^2,\ud r)$.
  \item[(ii)] $R_{G_{\alpha,\lambda,c}^{(D)}}$ is invertible on its range and the operator
  \begin{equation}\label{eq:defOfh0D}
   \mathsf{h}_{\alpha,0}^{(D)}\,:=\,(R_{G_{\alpha,\lambda,c}^{(D)}})^{-1}+mc^2-\lambda
  \end{equation}
  is a self-adjoint extension of $\mathsf{h}_{\alpha,0}$ such that $\mathsf{h}_{\alpha,0}^{(D)}-mc^2+\lambda$ has everywhere defined and bounded inverse.
  \item[(iii)] The operator $(\mathsf{h}_{\alpha,0}^{(D)}-mc^2+\lambda)^{-1}$ is an everywhere defined and bounded integral operator in $L^2(\mathbb{R}^+;\mathbb{C}^2,\ud r)$ with integral kernel $G_{\alpha,\lambda,c}^{(D)}(r,\rho)$. As a consequence, the resolvent set $\rho(\mathsf{h}_{\alpha,0}^{(D)}-mc^2)$ includes the open interval $(-2mc^2,0)$.
  \item[(iv)] Let
   \begin{equation}\label{eq:DefPsi}
	\Psi_{\alpha,\lambda,c}^{(D)} \,:=\, R_{G_{\alpha,\lambda,c}^{(D)}} \Phi_{\alpha,\lambda,c}^{(D)}\,=\, (\mathsf{h}_{\alpha,0}^{(D)}-mc^2+\lambda)^{-1}\Phi_{\alpha,\lambda,c}^{(D)}\, .
\end{equation}
 Then $\Psi_{\alpha,\lambda,c}^{(D)}\in C^\infty(\mathbb{R}^+;\mathbb{C}^2)$ and
 \begin{equation}\label{eq:AsintoticaPsiZero}
		\Psi_{\alpha,\lambda,c}^{(D)}(r) \,\stackrel{(r\downarrow 0)}{=} \,- \frac{\lambda}{\,c^2} C_{\alpha,\lambda,c} \Vert \Phi_{\alpha,\lambda,c}^{(D)} \Vert_{L^2(\mathbb{R}^+;\mathbb{C}^2)}^2 \begin{pmatrix} 0 \\ 1 \end{pmatrix} r^\alpha +o(r^{\frac{1}{2}})\, ,
\end{equation}
 where the constant $ C_{\alpha,\lambda,c}$ is defined in \eqref{eq:F0Def}.
 \end{enumerate}
 \end{proposition}

 For the proof of  Proposition \ref{prop:hD} we need the results contained in Lemmas \ref{lem:BoundednessIntegralOperator}-\ref{lem:Asympt0RanRG} below.

\begin{lemma}\label{lem:BoundednessIntegralOperator} Let $\alpha\in(-\frac{1}{2},\frac{1}{2})$ and $\lambda\in(0,2mc^2)$. The maximally defined integral operator $R_{G_{\alpha,\lambda,c}^{(D)}}$ acting on $\mathbb{R}^+\to\mathbb{C}^2$ functions with the integral kernel $G_{\alpha,\lambda,c}^{(D)}(r,\rho)$ defined in \eqref{eq:GreenF} is everywhere defined, bounded, and self-adjoint in $L^2(\mathbb{R}^+;\mathbb{C}^2,\ud r)$.
\end{lemma}

\begin{proof}
From
\[
G_{\alpha,\lambda,c}^{(D)}(r,\rho)=-\frac{\lambda}{\,c^2}\Big(\Phi_{\alpha,\lambda,c}^{(D)}(r) \otimes_{\mathbb{C}^2} \sigma_3 \overline{F_{\alpha,\lambda,c}^{(D)}(\rho)} \mathbf{1}_{(0,r)}(\rho)+F_{\alpha,\lambda,c}^{(D)}(r) \otimes_{\mathbb{C}^2} \sigma_3\overline{ \Phi_{\alpha,\lambda,c}^{(D)}(\rho)} \mathbf{1}_{(r,+\infty)}(\rho)\Big)
\]
one sees that $R_{G_{\alpha,\lambda,c}^{(D)}}$ splits into the sum of four integral operators with kernels given by
 \[
 	\begin{split}
 		G^{++}(r,\rho) \;&:=\; G_{\alpha,\lambda,c}^{(D)}(r,\rho) \mathbf{1}_{(1,+\infty)}(r) \mathbf{1}_{(1,+\infty)}(\rho) \, , \\
 		G^{+-}(r,\rho) \;&:=\; G_{\alpha,\lambda,c}^{(D)}(r,\rho)\mathbf{1}_{(1,+\infty)}(r) \mathbf{1}_{(0,1)}(\rho) \, ,\\
 		G^{-+}(r,\rho)\;&:=\;G_{\alpha,\lambda,c}^{(D)}(r,\rho) \mathbf{1}_{(0,1)}(r) \mathbf{1}_{(1,+\infty)}(\rho) \, ,\\
 		G^{--}(r,\rho)\;&:=\;G_{\alpha,\lambda,c}^{(D)}(r,\rho) \mathbf{1}_{(0,1)}(r) \mathbf{1}_{(0,1)}(\rho) \, . \\
 	\end{split}
 \]

We estimate each $G^{LM}(r,\rho)$, $L,M\in\{+,-\}$, by means of the large- and short-distance asymptotics \eqref{eq:AsymptPhiInfinity}-\eqref{eq:AsymptFInfinity} and \eqref{eq:AsymptZeroPhi}-\eqref{eq:AsymptZeroFLong} for $\Phi_{\alpha,\lambda,c}^{(D)}$ and $F_{\alpha,\lambda,c}^{(D)}$. For instance,
\[
 \begin{split}
  \Vert G^{+-}(r,\rho) \Vert_{\mathrm{M}_2(\mathbb{C})}\,&\lesssim\,\big\|\Phi_{\alpha,\lambda,c}^{(D)}(r)\big\|_{\mathrm{M}_2(\mathbb{C})}\big\|F_{\alpha,\lambda,c}^{(D)}(\rho)\big\|_{\mathrm{M}_2(\mathbb{C})}\,\mathbf{1}_{(1,+\infty)}(r) \mathbf{1}_{(0,1)}(\rho) \\
  &\lesssim\,e^{-r\,\mu_{c,\lambda}\sqrt{\lambda} } \,\rho^{\alpha}\,\mathbf{1}_{(1,+\infty)}(r) \mathbf{1}_{(0,1)}(\rho)\,.
 \end{split}
\]
 The net result is
\[\tag{*}
 	\label{eq:EstimateIntegralKernels}
 	\begin{split}
 		\Vert G^{++}(r,\rho) \Vert_{\mathrm{M}_2(\mathbb{C})} \; & \lesssim \; e^{-|r-\rho|\mu_{c,\lambda}\sqrt{\lambda}} \,\mathbf{1}_{(1,+\infty)}(r) \mathbf{1}_{(1,+\infty)}(\rho) \, , \\
 		\Vert G^{+-}(r,\rho) \Vert_{\mathrm{M}_2(\mathbb{C})} \; & \lesssim \; \rho^{\alpha} e^{-r\,\mu_{c,\lambda}\sqrt{\lambda} } \,\mathbf{1}_{(1,+\infty)}(r) \mathbf{1}_{(0,1)}(\rho)\, , \\
 		\Vert G^{-+}(r,\rho) \Vert_{\mathrm{M}_2(\mathbb{C})} \; & \lesssim \; r^{\alpha} e^{-\rho\,\mu_{c,\lambda}\sqrt{\lambda} } \,\mathbf{1}_{(0,1)}(r) \mathbf{1}_{(1,+\infty)}(\rho)\, , \\
 		\Vert G^{--}(r,\rho) \Vert_{\mathrm{M}_2(\mathbb{C})} \; & \lesssim \; (r \rho)^{-|\alpha|} \mathbf{1}_{(0,1)}(r)\, \mathbf{1}_{(0,1)}(\rho)\, , \\
 	\end{split}
 \]
 where $\|\cdot\|_{\mathrm{M}_2(\mathbb{C})}$ denotes the matrix norm, with respect to the representation
 \[
  \begin{split}
   G^{LM}(r,\rho)\,&=\,
   \begin{pmatrix}
    G^{LM}_{11}(r,\rho) & G^{LM}_{12}(r,\rho) \\
    G^{LM}_{21}(r,\rho) & G^{LM}_{22}(r,\rho)
   \end{pmatrix}, \\
   L^2(\mathbb{R}^+;\mathbb{C}^2,\ud r)\,&\cong\, L^2(\mathbb{R}^+;\mathbb{C},\ud r)\oplus L^2(\mathbb{R}^+;\mathbb{C},\ud r)\,.
  \end{split}
 \]

 The last three estimates in \eqref{eq:EstimateIntegralKernels} show that the kernels $G^{+-}(r,\rho)$, $G^{-+}(r,\rho)$ and $G^{--}(r,\rho)$ are in $L^2(\mathbb{R}^+ \times \mathbb{R}^+,\ud r \ud \rho)\otimes \text{M}_2(\mathbb{C})$ and therefore the corresponding integral operators are Hilbert-Schmidt operators, hence bounded, on $L^2(\mathbb{R}^+;\mathbb{C}^2, \ud r)$. The first estimate in \eqref{eq:EstimateIntegralKernels} allows to conclude, by an obvious Schur test, that also the integral operator with kernel $G^{++}(r,\rho)$ is bounded on $L^2(\mathbb{R}^+;\mathbb{C}^2, \ud r)$. This proves the overall boundedness of $R_{G_{\alpha,\lambda,c}^{(D)}}$.

 The self-adjointness of $R_{G_{\alpha,\lambda,c}^{(D)}}$ is clear from \eqref{eq:GreenF}: the adjoint $(R_{G_{\alpha,\lambda,c}^{(D)}})^*$ of $R_{G_{\alpha,\lambda,c}^{(D)}}$ has kernel $\overline{G_{\alpha,\lambda,c}^{(D)}(r,\rho)}^T$, where $T$ denotes matrix transposition. Since $K_{\alpha+\frac12}$ and $I_{\alpha+\frac12}$ in the definition \eqref{eq:PhiD}-\eqref{eq:FD} of $\Phi_{\alpha,\lambda,c}^{(D)}$ and $F_{\alpha,\lambda,c}^{(D)}$  are real-valued functions, then by direct inspection one sees that $\overline{G_{\alpha,\lambda,c}^{(D)}(\rho,r)}^T=G_{\alpha,\lambda,c}^{(D)}(r,\rho)$, thus showing that $(R_{G_{\alpha,\lambda,c}^{(D)}})^*=R_{G_{\alpha,\lambda,c}^{(D)}}$.
\end{proof}

\begin{lemma}\label{lem:invertsDiffProbl}
 Let $\alpha\in(-\frac{1}{2},\frac{1}{2})$ and $\lambda\in(0,2mc^2)$. The operator $R_{G_{\alpha,\lambda,c}^{(D)}}$ satisfies
 \begin{equation}\label{eq:invertsDiffProbl}
  \begin{pmatrix}
		\lambda & -\ii c \big(\frac{\ud}{\ud r}-\frac{\alpha}{r} \big) \\
		-\ii c \big( \frac{\ud}{\ud r} + \frac{\alpha}{r} \big) & - 2 mc^2 +\lambda
	\end{pmatrix}R_{G_{\alpha,\lambda,c}^{(D)}} g\,=\,g\qquad\forall g\in L^2(\mathbb{R}^+;\mathbb{C}^2,\ud r)\,.
   \end{equation}
\end{lemma}


 \begin{proof}
  Consider the inhomogeneous differential problem
   \[\tag{i}\label{eq:inhomProb}
  \begin{pmatrix}
		\lambda & -\ii c \big(\frac{\ud}{\ud r}-\frac{\alpha}{r} \big) \\
		-\ii c \big( \frac{\ud}{\ud r} + \frac{\alpha}{r} \big) & - 2 mc^2 +\lambda
	\end{pmatrix}u \,=\,g
  \]
   in the unknown $u$ for given $g$, and the associated homogeneous problem
     \begin{equation*}\tag{ii}\label{eq:homProb}
	\begin{pmatrix}
		\lambda & -\ii c \big(\frac{\ud}{\ud r}-\frac{\alpha}{r} \big) \\
		-\ii c \big( \frac{\ud}{\ud r} + \frac{\alpha}{r} \big) & - 2 mc^2 +\lambda
	\end{pmatrix} u\,=\,0 \, .
\end{equation*}
  In the proof of Proposition \ref{prop:KernelAdjoint} it was established that the solutions to \eqref{eq:homProb} form the two-dimensional subspace
  \[
   \mathrm{span}_{\mathbb{C}}\big\{ \Phi_{\alpha,\lambda,c}^{(D)},F_{\alpha,\lambda,c}^{(D)}\big\}\,.
  \]
 Moreover, Liouville's theorem prescribes that the Wronskian of any pair of solutions to \eqref{eq:homProb} is constant in $r$: the Wronskian can therefore computed as
\begin{equation*}
	W_{\alpha,\lambda,c}^{(D)} \,:=\, W(F_{\alpha,\lambda,c}^{(D)}, \Phi_{\alpha,\lambda,c}^{(D)}) \,:=\, \det(F_{\alpha,\lambda,c}^{(D)}\,|\, \Phi_{\alpha,\lambda,c}^{(D)}) \,=\,\frac{\ii c}{\lambda}\,,
\end{equation*}
 having used the asymptotics \eqref{eq:AsymptZeroPhi}-\eqref{eq:AsymptZeroFLong} for the explicit computation in the last identity (here $( F_{\alpha,\lambda,c}^{(D)} | \Phi_{\alpha,\lambda,c}^{(D)} )$ denotes the matrix with columns $F_{\alpha,\lambda,c}^{(D)}$ and $\Phi_{\alpha,\lambda,c}^{(D)}$). Now, a standard application of the method of variation of constants \cite[Sect.~2.4]{Wasow-ODE} implies that there is a \emph{particular} solution $u_{\textrm{part}}$ to \eqref{eq:inhomProb} given by
 \[
u_{\textrm{part}}(r)\,=\,\int_0^r   G(r,\rho) g(\rho) \,\ud \rho\,,
 \]
 where
 \[
 	G(r,\rho)\,:=\,
  \begin{cases}
   \:{\displaystyle-\frac{\ii}{\,c\,W_{\alpha,\lambda,c}^{(D)}}}\,\Phi_{\alpha,\lambda,c}^{(D)}(r) \otimes_{\mathbb{C}^2} \overline{F_{\alpha,\lambda,c}^{(D)}(\rho)} & \textrm{ if } 0<\rho<r\,, \\
  \:{\displaystyle-\frac{\ii}{\,c\,W_{\alpha,\lambda,c}^{(D)}}}\,F_{\alpha,\lambda,c}^{(D)}(r) \otimes_{\mathbb{C}^2} \overline{ \Phi_{\alpha,\lambda,c}^{(D)}(\rho)} & \textrm{ if } 0<r<\rho\,.
  \end{cases}
 \]
 Plugging in the explicit value of $W_{\alpha,\lambda,c}^{(D)}$ above, one recognises that the above integral kernel is precisely the integral kernel $G_{\alpha,\lambda,c}^{(D)}(r,\rho)$ defined in \eqref{eq:GreenF}. Therefore,
 \[
  u_{\textrm{part}}\,=\,R_{G_{\alpha,\lambda,c}^{(D)}} g\,,
 \]
 and \eqref{eq:invertsDiffProbl} is proved.
 \end{proof}

 \begin{lemma}\label{lem:Asympt0RanRG}
For every $g \in L^2(\mathbb{R}^+;\mathbb{C}^2)$ one has
	\begin{equation}\label{eq:Asymtp0RanRG}
		(R_{G_{\alpha,\lambda,c}^{(D)}} g)(r) \, \stackrel{(r\downarrow 0)}{=} \, -\frac{ \lambda }{\,c^2}\,C_{\alpha,\lambda,c} \langle  \Phi_{\alpha,\lambda,c}^{(D)}, g \rangle_{L^2(\mathbb{R}^+;\mathbb{C}^2)} \begin{pmatrix} 0 \\ 1 \end{pmatrix} r^{\alpha}+o(r^{\frac{1}{2}})
	\end{equation}
 with $C_{\alpha,\lambda,c}$ defined in \eqref{eq:F0Def} above.
\end{lemma}

\begin{proof}
 From \eqref{eq:GreenF} we compute
\begin{equation*}
	\begin{split}
	\frac{\,c^2}{\lambda} (R_{G_{\alpha,\lambda,c}^{(D)}} g)(r) \,&=\, - F_{\alpha,\lambda,c}^{(D)}(r) \langle  \Phi_{\alpha,\lambda,c}^{(D)}, g \rangle_{L^2((r,+\infty);\mathbb{C}^2)}-\Phi_{\alpha,\lambda,c}^{(D)} (r) \langle  F_{\alpha,\lambda,c}^{(D)},  g \rangle_{L^2((0,r);\mathbb{C}^2)} \\
	&=\,-F_{\alpha,\lambda,c}^{(D)}(r) \langle  \Phi_{\alpha,\lambda,c}^{(D)},g \rangle_{L^2(\mathbb{R}^+;\mathbb{C}^2)} - \Phi_{\alpha,\lambda,c}^{(D)} (r) \langle F_{\alpha,\lambda,c}^{(D)}, g \rangle_{L^2((0,r);\mathbb{C}^2)}\\& \qquad +F_{\alpha,\lambda,c}^{(D)}(r) \langle \Phi_{\alpha,\lambda,c}^{(D)}, g \rangle_{L^2((0,r);\mathbb{C}^2)} \,.
	\end{split}
\end{equation*}
 By means of the asymptotics \eqref{eq:AsymptZeroFLong} for $F_{\alpha,\lambda,c}^{(D)}$, we re-write
 \[
  \frac{\,c^2}{\lambda} (R_{G_{\alpha,\lambda,c}^{(D)}} g)(r) \,=\,-C_{\alpha,\lambda,c} \langle  \Phi_{\alpha,\lambda,c}^{(D)},  g \rangle_{L^2(\mathbb{R}^+;\mathbb{C}^2)} \begin{pmatrix} 0 \\ 1 \end{pmatrix} r^{\alpha}+ \mathcal{R}(r)
 \]
 with
 \[
  \begin{split}
		\mathcal{R}(r)\,&:=\,O(r^{2+\alpha}) \langle \Phi_{\alpha,\lambda,c}^{(D)},g \rangle_{L^2(\mathbb{R}^+;\mathbb{C}^2)}- \Phi_{\alpha,\lambda,c}^{(D)} (r) \langle F_{\alpha,\lambda,c}^{(D)}, g \rangle_{L^2((0,r);\mathbb{C}^2)}\\& \qquad +F_{\alpha,\lambda,c}^{(D)}(r) \langle \Phi_{\alpha,\lambda,c}^{(D)}, g \rangle_{L^2((0,r);\mathbb{C}^2)} \, .
	\end{split}
 \]
 The goal now is to prove that $\mathcal{R}(r)\stackrel{(r\downarrow 0)}{=}o(r^{\frac{1}{2}})$.

 In fact, obviously the first summand in the right-hand side above is $O(r^{2+\alpha})=o(r^{\frac{1}{2}})$, so we are left with proving
 \begin{equation*}\tag{*}\label{eq:Proof-Conti1}
	\Phi_{\alpha,\lambda,c}^{(D)} (r) \langle F_{\alpha,\lambda,c}^{(D)}, g \rangle_{L^2((0,r);\mathbb{C}^2)}-F_{\alpha,\lambda,c}^{(D)}(r) \langle \Phi_{\alpha,\lambda,c}^{(D)},  g \rangle_{L^2((0,r);\mathbb{C}^2)} \,=\, o(r^{\frac{1}{2}}) \, .
\end{equation*}

 Depending on the sign of $\alpha$, \eqref{eq:Proof-Conti1} is established as follows. For $\alpha > 0$ we estimate separately
\begin{equation*}
		\begin{split}
			\big|\Phi_{\alpha,\lambda,c}^{(D)}(r) \langle F_{\alpha,\lambda,c}^{(D)},  g \rangle_{L^2((0,r);\mathbb{C}^2)}\big|\, & \leqslant\, \big|\Phi_{\alpha,\lambda,c}^{(D)}(r)\big| \, \big\Vert F_{\alpha,\lambda,c}^{(D)} \big\Vert_{L^2((0,r);\mathbb{C}^2} \Vert\, g \Vert_{L^2((0,r);\mathbb{C}^2)} \\
			& \lesssim\, r^{-\alpha} r^{\alpha+\frac{1}{2}} o(1) \,=\, o(r^{\frac{1}{2}})
		\end{split}
\end{equation*}
	and
\begin{equation*}
		\begin{split}
			| F_{\alpha,\lambda,c}^{(D)}(r) \langle \Phi_{\alpha,\lambda,c}^{(D)}, g \rangle_{L^2((0,r);\mathbb{C}^2)} |& \leq |F_{\alpha,\lambda,c}^{(D)}(r)| \Vert \Phi_{\alpha,\lambda,c}^{(D)} \Vert_{L^2((0,r);\mathbb{C}^2)}  \Vert g \Vert_{L^2((0,r);\mathbb{C}^2)} \\
			&\lesssim r^{\alpha} r^{-\alpha+ \frac{1}{2}} o(1) =o(r^{\frac{1}{2}}) \,,
		\end{split}
\end{equation*}
 based on the short-distance asymptotics \eqref{eq:AsymptZeroPhi}-\eqref{eq:AsymptZeroFLong}.

For $\alpha \leqslant 0$ it is crucial instead to exploit the exact cancellation of the leading singularities of each summand in the left-hand side of \eqref{eq:Proof-Conti1}. In terms of \eqref{eq:AsymptZeroPhi}-\eqref{eq:Phi1Def} we re-write
	\begin{equation*}
		\begin{split}
			\Phi_{\alpha,\lambda,c}^{(D)}(r)& \langle F_{\alpha,\lambda,c}^{(D)}, g \rangle_{L^2((0,r);\mathbb{C}^2) } \\
			&=\, \bigg(\Phi_{\alpha,\lambda,c}^{(D)}(r)-\begin{pmatrix} 1 \\ 0 \end{pmatrix} A_{\alpha,\lambda,c}  r^{-\alpha} - \begin{pmatrix} 0 \\ 1 \end{pmatrix} B_{\alpha,\lambda,c}  r^{\alpha}\bigg) \langle F_{\alpha,\lambda,c}^{(D)}, g \rangle_{L^2((0,r);\mathbb{C}^2)} \\
			&\qquad +\bigg(\begin{pmatrix} 1 \\ 0 \end{pmatrix} A_{\alpha,\lambda,c}  r^{-\alpha} + \begin{pmatrix} 0 \\ 1 \end{pmatrix} B_{\alpha,\lambda,c}  r^{\alpha}\bigg) \langle F_{\alpha,\lambda,c}^{(D)}, g \rangle_{L^2((0,r);\mathbb{C}^2)} \, .
		\end{split}
	\end{equation*}
	The first summand in the right-hand side above was singled out because it is subleading: indeed, for small positive $r$ the boundedness of the scalar product $\langle F_{\alpha,\lambda,c}^{(D)}, g \rangle_{L^2((0,r);\mathbb{C}^2)}$ and the asymptotics \eqref{eq:AsymptZeroPhi} yield
	\[
	 \begin{split}
	  \bigg(\Phi_{\alpha,\lambda,c}^{(D)}(r)-&\begin{pmatrix} 1 \\ 0 \end{pmatrix} A_{\alpha,\lambda,c}  r^{-\alpha} - \begin{pmatrix} 0 \\ 1 \end{pmatrix} B_{\alpha,\lambda,c}  r^{\alpha}\bigg) \langle F_{\alpha,\lambda,c}^{(D)}, g \rangle_{L^2((0,r);\mathbb{C}^2)} \\
	  &=\,O(r^{1-|\alpha|})\,=\,o(r^{\frac{1}{2}})\,.
	 \end{split}
	\]
  On the other hand, the asymptotics \eqref{eq:AsymptZeroFLong} for $F_{\alpha,\lambda,c}^{(D)}$ yield
  \begin{equation*}
   \begin{split}
    \bigg(\begin{pmatrix} 1 \\ 0 \end{pmatrix}& A_{\alpha,\lambda,c}  r^{-\alpha} + \begin{pmatrix} 0 \\ 1 \end{pmatrix} B_{\alpha,\lambda,c}  r^{\alpha}\bigg) \langle F_{\alpha,\lambda,c}^{(D)}, g \rangle_{L^2((0,r);\mathbb{C}^2)} \\
    &=\,\bigg(\begin{pmatrix} 1 \\ 0 \end{pmatrix} A_{\alpha,\lambda,c}  r^{-\alpha} + \begin{pmatrix} 0 \\ 1 \end{pmatrix} B_{\alpha,\lambda,c}  r^{\alpha}\bigg)C_{\alpha,\lambda,c}\bigg\langle \begin{pmatrix} 0 \\ 1 \end{pmatrix} \rho^\alpha,g\bigg\rangle_{L^2((0,r);\mathbb{C}^2)} \\
    & \qquad + \bigg(\begin{pmatrix} 1 \\ 0 \end{pmatrix} A_{\alpha,\lambda,c}  r^{-\alpha} + \begin{pmatrix} 0 \\ 1 \end{pmatrix} B_{\alpha,\lambda,c}  r^{\alpha}\bigg)\langle O(\rho^{2+\alpha}), g \rangle_{L^2((0,r);\mathbb{C}^2)}\,.
   \end{split}
  \end{equation*}
 Since $|\langle O(\rho^{2+\alpha}), g \rangle_{L^2((0,r);\mathbb{C}^2)}| \lesssim r^{\frac{5}{2}+\alpha} \Vert g \Vert_{L^2((0,r);\mathbb{C}^2)}=o(r^{\frac{5}{2}+\alpha})$ by the Cauchy-Schwarz inequality,
    \begin{equation*}
   \begin{split}
    \bigg(\begin{pmatrix} 1 \\ 0 \end{pmatrix}& A_{\alpha,\lambda,c}  r^{-\alpha} + \begin{pmatrix} 0 \\ 1 \end{pmatrix} B_{\alpha,\lambda,c}  r^{\alpha}\bigg) \langle F_{\alpha,\lambda,c}^{(D)}, g \rangle_{L^2((0,r);\mathbb{C}^2)} \\
    &=\,\bigg(\begin{pmatrix} 1 \\ 0 \end{pmatrix} A_{\alpha,\lambda,c}  r^{-\alpha} + \begin{pmatrix} 0 \\ 1 \end{pmatrix} B_{\alpha,\lambda,c}  r^{\alpha}\bigg)C_{\alpha,\lambda,c}\bigg\langle \begin{pmatrix} 0 \\ 1 \end{pmatrix} \rho^\alpha,g\bigg\rangle_{L^2((0,r);\mathbb{C}^2)} + o(r^{\frac{1}{2}})\,.
   \end{split}
  \end{equation*}
  Using again the asymptotics \eqref{eq:AsymptZeroFLong} for $F_{\alpha,\lambda,c}^{(D)}$ and the Cauchy-Schwarz inequality now yields
  \[
   \bigg| \begin{pmatrix} 1 \\ 0 \end{pmatrix} A_{\alpha,\lambda,c} \, r^{-\alpha}\,C_{\alpha,\lambda,c}\bigg\langle \begin{pmatrix} 0 \\ 1 \end{pmatrix} \rho^\alpha,g\bigg\rangle_{L^2((0,r);\mathbb{C}^2)}\bigg|\,\leqslant\,\big|A_{\alpha,\lambda,c}\,C_{\alpha,\lambda,c}\big| \,r^{-\alpha}\,r^{\alpha+\frac{1}{2}}\,o(1)\,=\,o(r^{\frac{1}{2}})\,.
  \]
 Summarising,
 \[
  \Phi_{\alpha,\lambda,c}^{(D)}(r) \langle F_{\alpha,\lambda,c}^{(D)}, g \rangle_{L^2((0,r);\mathbb{C}^2) }\,=\,
  \begin{pmatrix} 0 \\ 1 \end{pmatrix}B_{\alpha,\lambda,c}C_{\alpha,\lambda,c}\,r^{\alpha}\bigg\langle \begin{pmatrix} 0 \\ 1 \end{pmatrix} \rho^\alpha,g\bigg\rangle_{L^2((0,r);\mathbb{C}^2)}+o(r^{\frac{1}{2}})\,.
 \]
 With a completely analogous reasoning we find
  \[
  F_{\alpha,\lambda,c}^{(D)}(r) \langle \Phi_{\alpha,\lambda,c}^{(D)}, g \rangle_{L^2((0,r);\mathbb{C}^2)}\,=\,
  \begin{pmatrix} 0 \\ 1 \end{pmatrix}B_{\alpha,\lambda,c}C_{\alpha,\lambda,c}\,r^{\alpha}\bigg\langle \begin{pmatrix} 0 \\ 1 \end{pmatrix} \rho^\alpha,g\bigg\rangle_{L^2((0,r);\mathbb{C}^2)}+o(r^{\frac{1}{2}})\,.
 \]
 Plugging the latter two asymptotics into the left-hand side of \eqref{eq:Proof-Conti1}, one obtains the desired bound also when $\alpha\leqslant 0$.
\end{proof}

 \begin{proof}[Proof of Proposition \ref{prop:hD}]
  Part (i) is established in Lemma \ref{lem:BoundednessIntegralOperator}. Concerning part (ii), Lemma \ref{lem:invertsDiffProbl} establishes that
   \begin{equation*}
    \left[\begin{pmatrix}
	mc^2 & \ii c \Big( -\frac{\ud}{\ud r}+\frac{\alpha}{r}\Big) \\
	-\ii c \Big(\frac{\ud}{\ud r} + \frac{\alpha}{r} \Big) & -mc^2
	\end{pmatrix} +(\lambda-mc^2)\begin{pmatrix} 1 & 0 \\ 0 & 1 \end{pmatrix}\right] R_{G_{\alpha,\lambda,c}^{(D)}} g\,=\,g
   \end{equation*}
 for any $g\in L^2(\mathbb{R}^+;\mathbb{C}^2,\ud r)$, and therefore there is a self-adjoint extension $\mathcal{S}$ of $\mathsf{h}_{\alpha,0}$ such that
 \[
  (\mathcal{S}-mc^2+\lambda)\,R_{G_{\alpha,\lambda,c}^{(D)}} g\,=\,g\qquad \forall g\in L^2(\mathbb{R}^+;\mathbb{C}^2,\ud r)\,.
 \]
 Then, by self-adjointness, also
 \[
  R_{G_{\alpha,\lambda,c}^{(D)}} \,(\mathcal{S}-mc^2+\lambda)\,f\,=\,f\qquad \forall f\in \mathcal{D}(\mathcal{S})\,.
 \]
 Thus, $R_{G_{\alpha,\lambda,c}^{(D)}}=(\mathcal{S}-mc^2+\lambda)^{-1}$ and necessarily $\mathcal{S}=\mathsf{h}_{\alpha,0}^{(D)}$, the operator defined in \eqref{eq:defOfh0D}. Part (iii) is an equivalent rephrasing of part (ii). Concerning part (iv),
 the smoothness of  $\Psi_{\alpha,\lambda,c}^{(D)}$ follows from the smoothness of the integral kernel of $(\mathsf{h}_{\alpha,0}^{(D)}-mc^2+\lambda)^{-1}$ and the smoothness of $\Phi_{\alpha,\lambda,c}^{(D)}$, whereas the asymptotics \eqref{eq:AsintoticaPsiZero} follows at once from Lemma \ref{lem:Asympt0RanRG}.
 \end{proof}

\subsection{Proof of Theorem \ref{prop:Dirac}}\label{subsect:SD}~

 Part (i) of Theorem \ref{prop:Dirac} is established in Proposition \ref{prop:Closure}(i)-(ii).

 Part (ii) follows from Proposition \ref{prop:Closure}(iii).

 Concerning part (iii) of Theorem \ref{prop:Dirac}, any $g\in\mathcal{D}(\mathsf{h}_{\alpha,0}^*)=\mathcal{D}(\mathsf{h}_{\alpha,0}^*-mc^2+\lambda)$ has the form
 \begin{equation}\label{eq:ginDstar}
  g\,=\,\varphi +c_1 \Psi_{\alpha,\lambda,c}^{(D)} + c_0 \Phi_{\alpha,\lambda,c}^{(D)}
 \end{equation}
  for $g$-dependent $\varphi\in H^1_0(\mathbb{R}^+;\mathbb{C}^2)$ and $c_0,c_1\in\mathbb{C}$. This is a direct application of the general formula \eqref{eq:domainstar11}: the fact that $\mathcal{D}(\overline{\mathsf{h}_{\alpha,0}})=H^1_0(\mathbb{R}^+;\mathbb{C}^2)$ is proved in Proposition \ref{prop:Closure}(i), the fact that $\Phi_{\alpha,\lambda,c}^{(D)}$ spans $\ker(\mathsf{h}_{\alpha,0}^*-mc^2+\lambda)$ is proved in Proposition \ref{prop:KernelAdjoint}, and $\Psi_{\alpha,\lambda,c}^{(D)}$ is defined in \eqref{eq:DefPsi} as the result of the action on $\Phi_{\alpha,\lambda,c}^{(D)}$ of the bounded inverse $(\mathsf{h}_{\alpha,0}^{(D)}-mc^2+\lambda)^{-1}$ of the distinguished self-adjoint extension of $\mathsf{h}_{\alpha,0}-mc^2+\lambda$ produced in Proposition \ref{prop:hD}.

  The short-distance version of \eqref{eq:ginDstar} is obtained by means of the $r\downarrow 0$-asymptotics, respectively, of $\varphi$ (\eqref{eq:AsymptH1AtZero} from Proposition \ref{prop:Closure}), $\Phi_{\alpha,\lambda,c}^{(D)}$ (\eqref{eq:AsymptZeroPhi} from Lemma \ref{lem:asymptoticsPhiF}), and $\Psi_{\alpha,\lambda,c}^{(D)}$ (\eqref{eq:AsintoticaPsiZero} from Proposition \ref{prop:hD}(iv)). The net result is
\begin{equation}\label{eq:ShortRangeAdj}
	g(r)\,\equiv\, \begin{pmatrix}
		g_+(r) \\ g_-(r)
	\end{pmatrix}\,=\, \begin{pmatrix}
		1 \\
		0
	\end{pmatrix} g_0 r^{-\alpha} +
	\begin{pmatrix}
		0 \\
		1
	\end{pmatrix}g_1 r^{\alpha}
	+o(r^{\frac{1}{2}}) \, ,
\end{equation}
 where
 \begin{align}
  g_0\,&:=\,\lim_{r \to 0^+} r^\alpha g_+(r) =  c_0 A_{\alpha,\lambda,c} \, , \label{eq:g0pm} \\
  g_1\,&:=\,\lim_{r \to 0^+}r^{-\alpha} g_-(r) = c_0 B_{\alpha,\lambda,c}-c_1\frac{\lambda }{\,c^2 }C_{\alpha,\lambda,c} \| \Phi_{\alpha,\lambda,c}^{(D)} \|_{L^2(\mathbb{R}^+;\mathbb{C}^2)}^2   \, . \label{eq:g1pm}
 \end{align}
 This proves \eqref{eq:AsymptoticAdjoint} in Theorem \ref{prop:Dirac}(iii).

 By general facts of the Kre\u{\i}n-Vi\v{s}ik-Birman-Grubb self-adjoint extension theory (see Theorem \ref{thm:KVB-General-App11} and formula \eqref{eq:KVB-11} therein), the self-adjoint extensions of $\mathsf{h}_{\alpha,0}$ form the family $\big(\mathsf{h}_{\alpha,0,\beta}\big)_{\beta\in\mathbb{R}^+\cup\{\infty\}}$, where $\mathsf{h}_{\alpha,0,\infty}$ is precisely the distinguished extension $\mathsf{h}_{\alpha,0}^{(D)}$, whereas, for $\beta\in\mathbb{R}$,
 \begin{equation}\label{eq:halphb}
 \mathsf{h}_{\alpha,0,\beta}\,=\,\mathsf{h}_{\alpha,0}^*\upharpoonright\{g\in\mathcal{D}(\mathsf{h}_{\alpha,0}^*)\,|\,c_1=\beta c_0\}
 \end{equation}
 with respect to the notation of \eqref{eq:ginDstar}. Plugging the self-adjointness condition $c_1=\beta c_0$ into \eqref{eq:g0pm}-\eqref{eq:g1pm} yields the equivalent condition
 \begin{equation}\label{eq:g1gammag0}
  g_1\,=\,-\ii \gamma g_0\,,
 \end{equation}
 having set
 \begin{equation}\label{eq:GammaBeta}
	\gamma \;:=\;
	\frac{\,\beta\, \frac{\lambda}{\,c^2}\, C_{\alpha,\lambda,c}\,  \| \Phi_{\alpha,\lambda,c}^{(D)}\|_{L^2(\mathbb{R}^+;\mathbb{C}^2)}^2- B_{\alpha,\lambda,c}\,}{ \ii A_{\alpha,\lambda,c}}\,.
\end{equation}
 In practice, this amounts to re-labelling $\mathsf{h}_{\alpha,0,\beta}\equiv\mathsf{h}_{\alpha,0}^{(\gamma)}$ by means of the correspondence $\gamma=\gamma(\beta)$ given by \eqref{eq:GammaBeta}. Observe that $\gamma$ is real and monotone increasing with $\beta$. In particular, $\gamma=\infty$ if and only if $\beta=\infty$. So,
 \[
  \mathsf{h}_{\alpha,0}^{(D)}\,=\,\mathsf{h}_{\alpha,0,\infty}\,=\,\mathsf{h}_{\alpha,0}^{(\infty)}\,,
 \]
 and the notation $\mathsf{h}_{\alpha,0}^{(\infty)}$ shall be used henceforth.
 This establishes \eqref{eq:BoundarySelfAdjointDirac} and completes the proof of Theorem \ref{prop:Dirac}(iii).

 Concerning part (iv) of Theorem \ref{prop:Dirac}, the inclusion $\rho(\mathsf{h}_{\alpha,0}^{(\infty)}-mc^2)\supset (-2mc^2,0)$ and the identity \eqref{eq:ThmReshinf} are established in Proposition \ref{prop:hD}(iii). Again by general facts of the Kre\u{\i}n-Vi\v{s}ik-Birman-Grubb self-adjoint extension theory (see Theorem \ref{thm:KVB2-d1} and formula \eqref{eq:Res11KVB} therein),
 \begin{equation}\label{eq:idrestemp}
   (\mathsf{h}_{\alpha,0}^{(\gamma)}-mc^2+\lambda)^{-1} \,=\, (\mathsf{h}_{\alpha,0}^{(\infty)}-mc^2+\lambda)^{-1}+ \beta^{-1}  \Vert \Phi_{\alpha,\lambda,c}^{(D)} \Vert_{L^2(\mathbb{R}^+;\mathbb{C}^2)}^{-2}\big| \Phi^{(D)}_{\alpha,\lambda,c} \big\rangle \big\langle \Phi_{\alpha,\lambda,c}^{(D)} \big|\,.
  \end{equation}
  Since $\mathsf{h}_{\alpha,0}^{(\gamma)}-mc^2$ is a rank-one perturbation of $\mathsf{h}_{\alpha,0}^{(\infty)}-mc^2$, for $\gamma\neq \infty$ the resolvent set $\rho(\mathsf{h}_{\alpha,0}^{(\gamma)}-mc^2)$ contains $(-2mc^2,0)$ up to possibly an isolated eigenvalue in the gap.
 The identity \eqref{eq:idrestemp} is re-written as
 \begin{equation}\label{eq:resrepr}
   (\mathsf{h}_{\alpha,0}^{(\gamma)}-mc^2+\lambda)^{-1} \,=\, (\mathsf{h}_{\alpha,0}^{(\infty)}-mc^2+\lambda)^{-1}+ \tau\,\frac{\lambda}{c^2} \big| \Phi^{(D)}_{\alpha,\lambda,c} \big\rangle \big\langle \Phi_{\alpha,\lambda,c}^{(D)} \big|
  \end{equation}
  upon setting
 \begin{equation}
  \tau\,:=\,\Big(\beta \Vert \Phi_{\alpha,\lambda,c}^{(D)} \Vert_{L^2(\mathbb{R}^+;\mathbb{C}^2)}^2 \frac{\lambda}{c^2} \Big)^{-1}\,=\,\frac{C_{\alpha,\lambda,c}}{\,\gamma(\ii A_{\alpha,\lambda,c})+B_{\alpha,\lambda,c}\,} \,,
 \end{equation}
  and having used \eqref{eq:GammaBeta} in the second identity above. Using now the explicit expressions \eqref{eq:Phi0Def}-\eqref{eq:F0Def} for $A_{\alpha,\lambda,c}$, $B_{\alpha,\lambda,c}$, and $C_{\alpha,\lambda,c}$, and the identity
  \[
   \Gamma\Big(\frac{1}{2}+z\Big)\Gamma\Big(\frac{1}{2}-z\Big)\,=\, \pi \sec(\pi z)
  \]
  (see \cite[Eq. (6.1.17)]{Abramowitz-Stegun-1964}), one finds
  \begin{equation}
   \begin{split}
    \tau\,&=\,\frac{\Gamma(\frac{1}{2}+\alpha)^{-1}\, 2^{\frac{1}{2}-\alpha}  \,(\mu_{c,\lambda}\sqrt{\lambda})^{-\frac{1}{2}+\alpha} }{\,\gamma\big(2^{\alpha-\frac{1}{2}}\,c\,\mu_{c,\lambda}^2\,(\mu_{c,\lambda}\sqrt{\lambda})^{-\frac{3}{2}-\alpha}\, \Gamma(\frac{1}{2}+\alpha)\big)+ 2^{-\frac{1}{2}+\alpha} (\mu_{c,\lambda}\sqrt{\lambda})^{\alpha-\frac{1}{2}}\,\Gamma(\frac{1}{2}-\alpha)\,   } \\
    &=\,\frac{2}{\,\Gamma(\frac{1}{2}+\alpha)\,\big( \gamma\,4^\alpha\,c\, \mu_{c,\lambda}^2\,(\mu_{c,\lambda}\sqrt{\lambda})^{-1-2\alpha} \,\Gamma(\frac{1}{2}+\alpha  )+ \Gamma(\frac{1}{2}-\alpha) \big)\,} \\
    &=\,\frac{2}{\,\pi \sec(\pi \alpha)+4^\alpha \mu_{c,\lambda}^2 (\mu_{c,\lambda}\sqrt{\lambda})^{-1-2 \alpha} \,\Gamma^2(\frac12+\alpha) c \gamma\,}  \, .
   \end{split}
  \end{equation}
 This means that $\tau=\tau_{\alpha,\lambda,\gamma,c}^{(D)}$ as defined in \eqref{eq:TauDGamma}. Thus, \eqref{eq:resrepr} reproduces \eqref{eq:mainres}. This completes the proof of Theorem \ref{prop:Dirac}(iv).

\section{Block-wise self-adjoint extensions of the Schr\"{o}dinger-AB operator}\label{sec:Schroe}

 For completeness and self-consistency of presentation, we mirror in this Section the main findings of Section \ref{sec:Dirac}, now concerning the self-adjoint extensions in $L^2(\mathbb{R}^+;\mathbb{C},\ud r)$ of the operators $\mathsf{S}_{\alpha,k}$, $k\in\mathbb{Z}$, introduced in \eqref{eq:OperatoreS} and minimally defined on the domain $C^\infty_0(\mathbb{R}^+;\mathbb{C})$.

Here too, the discussion is modelled along the Kre\u{\i}n-Vi\v{s}ik-Birman-Grubb self-adjoint extension scheme. Thus, he conceptual path is precisely the same as the one of Section \ref{sec:Dirac} and Appendix \ref{sec:proofOfMain}: only the computations are updated accordingly. We opt therefore for quoting with a few elucidative remarks, instead of complete proofs, the main results relevant for the present discussion.

 One observes, as customary, that for the operator
 \[
  \mathsf{S}_{\alpha,k}\,:=\, \frac{1}{2m} \Big( -\frac{\ud^2}{\ud r^2}+\frac{(\alpha+k)(\alpha+k+1)}{r^2}\,\Big)\,,\qquad \mathcal{D}(\mathsf{S}_{\alpha,k})\,:=\,C^\infty_0(\mathbb{R}^+;\mathbb{C})\,,
 \]
 Hardy's inequality $\Vert u' \Vert_{L^2(\mathbb{R}^+;\mathbb{C})} \geqslant \frac{1}{2} \Vert r^{-1} u \Vert_{L^2(\mathbb{R}^+;\mathbb{C})}$ and the bound $(\alpha+k)(\alpha+k-1) \geqslant -\frac{1}{4}$ imply
 \begin{equation}
	\langle  u, \mathsf{S}_{\alpha,k} u \rangle_{L^2(\mathbb{R}^+;\mathbb{C})} \,\geqslant\, 0 \,, \qquad \forall u \in C^\infty_0(\mathbb{R}^+;\mathbb{C}) \, ,
\end{equation}
 meaning that $\mathsf{S}_{\alpha,k}$ is non-negative. Any shift $\lambda>0$ makes therefore $\mathsf{S}_{\alpha,k}+\lambda$ a lower semi-bounded operator \emph{with strictly positive lower bound}, making it treatable within the Kre\u{\i}n-Vi\v{s}ik-Birman scheme using the Friedrichs extension as reference extension with everywhere defined and bounded inverse in $L^2(\mathbb{R}^+;\mathbb{C},\ud r)$. The actual extensions for $\mathsf{S}_{\alpha,k}$ are then read out by removing the $\lambda$-shift.

 Analogously to Proposition \ref{prop:Closure}(i), the operator closure is characterised following \cite[Theorems~4.1 and 6.1]{Derezinski-Georgescu-2021}:

 \begin{proposition} For $\alpha\in(-\frac{1}{2},\frac{1}{2})$ and $k\in\mathbb{Z}$, one has $\mathcal{D}(\overline{\mathsf{S}_{\alpha,k}})=H^2_0(\mathbb{R}^+;\mathbb{C})$.
 \end{proposition}

 For $\lambda>0$, we introduce the $\mathbb{R}^+\to\mathbb{R}$ functions
 \begin{align}\label{Phi-F-Schr-gen}
	\Phi_{\alpha,\lambda,k}^{(S)}(r)\,&:=\, \sqrt{r} K_{\alpha+k+\frac12}(r\sqrt{2m \lambda} ) \, , \\
	F_{\alpha,\lambda,k}^{(S)}(r)\,&:=\, \sqrt{r} I_{\alpha+k+\frac12}(r\sqrt{2m \lambda} ) \,.
\end{align}
 Let us also introduce the short-hand notation
  \begin{align}\label{Phi-F-Schr-0}
	\Phi_{\alpha,\lambda}^{(S)}\,&:=\, \Phi_{\alpha,\lambda,0}^{(S)} \, , \\
	F_{\alpha,\lambda}^{(S)}\,&:=\, F_{\alpha,\lambda,0}^{(S)} \,.
\end{align}

 Analogous to Lemma \ref{lem:asymptoticsPhiF} one here has:

 \begin{lemma}\label{lem:asymPhiF-SCHROE}
  Let $\alpha\in(-\frac{1}{2},\frac{1}{2})$, $k\in\mathbb{Z}$, and $\lambda>0$. Then $\Phi_{\alpha,\lambda,k}^{(S)}$ and $F_{\alpha,\lambda,k}^{(S)}$ are smooth and with asymptotics
  \begin{align}\label{eq:AsintoticaPhiSSZero}
	\Phi_{\alpha,\lambda,k}^{(S)}(r)\,&\stackrel{(r\downarrow 0)}{=}  \, A_{\alpha,\lambda,k}^{(S)} r^{-(\alpha+k)}+B_{\alpha,\lambda,k}^{(S)} r^{1+\alpha+k} + O(\max\{r^{3+k+\alpha},r^{2-k-\alpha}\})\,, \\
	F_{\alpha,\lambda,k}^{(S)}(r)\,&\stackrel{(r\downarrow 0)}{=}  \, C_{\alpha,\lambda,k}^{(S)} r^{1+\alpha+k}+O(r^{2+\alpha+k})\,,
\end{align}
where
\begin{align}
	A_{\alpha,\lambda,k}^{(S)}\,&:=\,2^{\frac{\alpha }{2}+\frac{k}{2}-\frac{3}{4}}(m \lambda )^{-\frac{\alpha
   }{2}-\frac{k}{2}-\frac{1}{4}} \Gamma \left({\textstyle k+\alpha +\frac{1}{2}}\right)  \, ,\\
	B_{\alpha,\lambda,k}^{(S)}\,&:=\,2^{-\frac{\alpha }{2}-\frac{k}{2}-\frac{5}{4}} (m \lambda )^{\frac{\alpha}{2}+\frac{k}{2}+\frac{1}{4}} \Gamma \left({\textstyle-k-\alpha -\frac{1}{2}}\right)  \, , \\
	C_{\alpha,\lambda,k}^{(S)}\,&:=\,\frac{2^{-\frac{\alpha }{2}-\frac{k}{2}-\frac{1}{4}} (\lambda m)^{\frac{\alpha }{2}+\frac{k}{2}+\frac{1}{4}}}{\Gamma \left(k+\alpha +\frac{3}{2}\right)}\,.
\end{align}
 \end{lemma}

 The counterpart to Propositions \ref{prop:Closure}(iii) and \ref{prop:KernelAdjoint} now reads:

 \begin{proposition}\label{prop:KernelAdjointSCHROE} Let $\alpha\in(-\frac{1}{2},\frac{1}{2})$, $k\in\mathbb{Z}$, and $\lambda>0$. Then
  \begin{align}
   \ker (\mathsf{S}_{\alpha,k}^*)\,&=\,\mathrm{span}_{\mathbb{C}}\big\{ \Phi_{\alpha,\lambda,k}^{(S)}\big\}\,, & & \textrm{ if }k\in\{-1,0\}\,, \\
   \ker (\mathsf{S}_{\alpha,k}^*)\,&=\,\{0\}\,,  & & \textrm{ if }k\in\mathbb{Z}\setminus\{-1,0\}\,.
  \end{align}
\end{proposition}

 This result
 follows from the search of square-integrable solutions to the homogeneous differential problem
  \begin{equation}\label{eq:App3-KernelSol}
	\Big( -\frac{\ud^2}{\ud r^2} + \frac{(\alpha+k)(\alpha+k+1)}{r^2} +\lambda \Big)u \,=\,0 \, ,
\end{equation}
  whose general solution has the form
  \begin{equation}
	u(r)\,=\,c_1 \sqrt{r} K_{\alpha+k+\frac{1}{2}}(r\sqrt{2m \lambda} ) + c_2 \sqrt{r} I_{\alpha+k+\frac{1}{2}}(r\sqrt{2m \lambda} ) \,.
\end{equation}
 Square-integrability is checked by means of the asymptotics \cite[Eq.~(9.6.2), (9.6.10), (9.7.1), (9.7.2)]{Abramowitz-Stegun-1964}: as $r\to +\infty$, $K_{\alpha+k+\frac{1}{2}}$ is exponentially dumped and $I_{\alpha+k+\frac{1}{2}}$ grows exponentially instead, whereas, as $r\downarrow 0$,
 \[
  K_{\alpha+k+\frac{1}{2}}(\sqrt{2 m \lambda} r) \,\sim\, \max\{r^{-k-\alpha}, r^{1+k+\alpha} \}\,.
 \]
 Thus, since $\alpha\in(-\frac{1}{2},\frac{1}{2})$, the function $K_{\alpha+k+\frac{1}{2}}$ is square-integrable also at the origin if and only if $k \in  \mathbb{Z} \setminus \{-1,0\}$.

 As argued in the introduction (Section \ref{sec:nonrelset}), out of the two non-trivial sectors $k=-1$ and $k=0$ our modelling only implements the $\infty^1$ extensions of the sector $k=0$, keeping the Friedrichs extension when $k=-1$. So, when $k=0$, the analogue of Proposition \ref{prop:hD} is the following.

 \begin{proposition}
  Let $\alpha\in(-\frac{1}{2},\frac{1}{2})$ and $\lambda>0$. Let $R_{G_{\alpha,\lambda}^{(S)}}$ be the integral operator acting on $\mathbb{R}^+\to\mathbb{C}$ functions with integral kernel  \begin{equation}
   G_{\alpha,\lambda}^{(S)}(r,\rho)\,:=\,2m\Big( \Phi_{\alpha,\lambda}^{(S)}(r)F_{\alpha,\lambda}^{(S)}(\rho)\mathbf{1}_{(0,r)}(\rho)+F_{\alpha,\lambda}^{(S)}(r)\Phi_{\alpha,\lambda}^{(S)}(\rho)\mathbf{1}_{(r,+\infty)}(\rho)\Big)\,,\quad r,\rho>0\,.
  \end{equation}
   \begin{enumerate}
    \item[(i)]   $R_{G_{\alpha,\lambda}^{(F)}}$  is everywhere defined, bounded, and self-adjoint in $L^2(\mathbb{R}^+;\mathbb{C},\ud r)$, it is invertible on its range, and $,(R_{G_{\alpha,\lambda}^{(S)}})^{-1}-\lambda$ is the Friedrichs extension of $\mathsf{S}_{\alpha,0}$, denoted directly with $\mathsf{S}_{\alpha,0}^{(\infty)}$. Thus,
  \begin{equation}
   (\mathsf{S}_{\alpha,0}^{(\infty)}+\lambda)^{-1}\,=\,R_{G_{\alpha,\lambda}^{(S)}}\,.
  \end{equation}
  \item[(ii)] The function
  \begin{equation}
   \Psi_{\alpha,\lambda}^{(S)}\,:=\, R_{G_{\alpha,\lambda}^{(S)}}\Phi_{\alpha,\lambda}^{(S)}\,=\,(\mathsf{S}_{\alpha,0}^{(\infty)}+\lambda)^{-1}\Phi_{\alpha,\lambda}^{(S)}
  \end{equation}
   satisfies
  \begin{equation}\label{eq:C17}
	\Psi_{\alpha,\lambda}^{(S)}(r)\,\stackrel{(r\downarrow 0)}{=}  \, 2m  C_{\alpha,\lambda,0} \Vert \Phi_{\alpha,\lambda}^{(S)} \Vert_{L^2(\mathbb{R}^+;\mathbb{C})}^2r^{1+\alpha} +o(r^{\frac{3}{2}}) \, .
\end{equation}
   \end{enumerate}
 \end{proposition}

 One can now finally mimic the reasoning of Section \ref{subsect:SD} and Appendix \ref{sec:proofOfMain}, reconstructing a la Kre\u{\i}n-Vi\v{s}ik-Birman the whole family $(\mathsf{S}_{\alpha,0,\omega})_{\omega\in\mathbb{R}\cup\{\infty\}}$ of self-adjoint extensions of $\mathsf{S}_{\alpha,0}$. First, the domain of the adjoint of $\mathsf{S}_{\alpha,0}$ is characterised as
 \begin{equation}\label{eq:ginDstarSCHROE}
  \mathcal{D}(\mathsf{S}_{\alpha,0}^*)\,=\,\mathcal{D}(\mathsf{S}_{\alpha,0}^*+\lambda)\,=\,\,\left\{\varphi+c_1\Psi_{\alpha,\lambda}^{(S)}+c_0\Phi_{\alpha,\lambda}^{(S)}\,\Big| \begin{array}{c} \varphi \in H^2_0(\mathbb{R}^+;\mathbb{C}),  \\
			c_0,c_1\in\mathbb{C} \end{array} \right\}
 \end{equation}
 (analogously to \eqref{eq:ginDstar}). The extension with $\omega=\infty$ is the Friedrichs extension, and for $\omega\in\mathbb{R}$
 \begin{equation}
\mathsf{S}_{\alpha,0,\omega}\,:=\,\mathsf{S}_{\alpha,0}^*\upharpoonright\{g\in \mathcal{D}(\mathsf{S}_{\alpha,0}^*)\,|\,c_1=\omega c_0\}\,.
 \end{equation}

 The short-distance version of \eqref{eq:ginDstarSCHROE} is obtained by implementing in $g\in\mathcal{D}(\mathsf{S}_{\alpha,0}^*)$ the asymptotics \eqref{eq:AsintoticaPhiSSZero}, and \eqref{eq:C17} and reads
 \begin{equation}
	g(r)\overset{(r \downarrow 0)}{=} a_0 r^{-\alpha} + a_1 r^{1+\alpha} + o (r^{\frac{3}{2}})
\end{equation}
 for $c_0,c_1$-dependent constants $a_0,a_1\in\mathbb{C}$ (analogously to \eqref{eq:ShortRangeAdj}). Correspondingly, the self-adjointness condition $c_1=\omega c_0$ takes the form
 \begin{equation}
  a_1\,=\theta a_0
 \end{equation}
  with
  \begin{equation}\label{eq:thal}
	\theta\,:=\,\frac{\,B_{\alpha,\lambda,0}^{(S)}+ 2m \omega C_{\alpha,\lambda,0}^{(S)\,} \Vert \Phi_{\alpha,\lambda}^{(S)} \Vert_{L^2(\mathbb{R}^+;\mathbb{C})}^2}{A_{\alpha,\lambda,0}^{(S)}}
\end{equation}
 (analogously to \eqref{eq:g1gammag0}-\eqref{eq:GammaBeta}). This amounts to re-labelling $\mathsf{S}_{\alpha,0,\omega}\equiv\mathsf{S}_{\alpha,0}^{(\theta)}$ with $\theta=\theta(\alpha)$ given by \eqref{eq:thal}. Actually $\theta$ is monotone increasing with $\alpha$ and $\alpha=\infty$ if and only if $\theta=\infty$. The notation $\mathsf{S}_{\alpha,0}^{(\infty)}$ for the Friedrichs extension is thus justified.

 Last, concerning resolvents, the general theory (see Theorem \ref{thm:KVB2-d1} and formula \eqref{eq:Res11KVB} therein) prescribes that $\mathsf{S}_{\alpha,0,\omega}+\lambda$ is invertible with everywhere defined bounded inverse if and only if $\omega\neq 0$, in which case
 \begin{equation}\label{eq:Sres1}
  (\mathsf{S}_{\alpha,0,\omega}+\lambda)^{-1}\,=\,(\mathsf{S}_{\alpha,0,\infty}+\lambda)^{-1}+\omega^{-1}\Vert \Phi_{\alpha,\lambda}^{(S)} \Vert_{L^2(\mathbb{R}^+;\mathbb{C})}^{-2}\big|\Phi_{\alpha,\lambda}^{(S)}\big\rangle \big\langle \Phi_{\alpha,\lambda}^{(S)} \big|\,.
 \end{equation}
 In the $\theta$-parametrisation we set
 \begin{equation}
   \tau_{\alpha,\lambda,\theta}^{(S)}\,:=\,\big(2m\omega\Vert \Phi_{\alpha,\lambda}^{(S)} \Vert_{L^2(\mathbb{R}^+;\mathbb{C})}^2 )^{-1}
 \end{equation}
 and calculate, by means of \eqref{eq:thal},
 \begin{equation}
  \tau_{\alpha,\lambda,\theta}^{(S)}\,=\, \frac{2}{\pi \sec(\pi \alpha)+2^{\frac12+\alpha}  (m \lambda)^{-\frac12-\alpha} \Gamma(\frac12+\alpha) \Gamma(\frac32+\alpha)\theta}\,.
 \end{equation}
 This turns \eqref{eq:Sres1} into its final form \eqref{eq:Scr0res} below.

 The above discussion is summarised as follows, in analogy to Theorem \ref{prop:Dirac}.

\begin{theorem}\label{prop:k0} Let $\alpha\in(-\frac{1}{2},\frac{1}{2})$.
\begin{itemize}
	\item[(i)] Any $\psi\in\mathcal{D}(\mathsf{S}_{\alpha,0}^*)$ displays the short-distance asymptotics
	\begin{equation}\label{eq:shortdasympt}
		\psi(r) \stackrel{(r\downarrow 0)}{=} a_0 r^{-\alpha}+a_1 r^{1+\alpha} + o(r^{\frac{3}{2}})
	\end{equation}
	for $\psi$-dependent $a_0,a_1\in\mathbb{C}$.
	\item[(ii)] The self-adjoint extensions in $L^2(\mathbb{R}^+;\mathbb{C},\ud r)$ of the operator $\mathsf{S}_{\alpha,0}$ form the one-parameter family $\big(\mathsf{S}_{\alpha,0}^{(\theta)}\big)_{\theta\in \mathbb{R}\cup\{\infty\}}$ defined, with the notation of \eqref{eq:shortdasympt}, by
	\begin{equation}\label{eq:DomSak0}
		\begin{split}
		\mathcal{D}(\mathsf{S}_{\alpha,0}^{(\theta)})\,&:=\,\{ \psi \in \mathcal{D}(\mathsf{S}_{\alpha,0}^*) \, | \, a_1 = \theta a_0 \} \,, \\
			\mathsf{S}_{\alpha,0}^{(\theta)}\,&:=\,\mathsf{S}_{\alpha,0}^* \upharpoonright \mathcal{D}(\mathsf{S}_{\alpha,0}^{(\theta)})\,.
		\end{split}
	\end{equation}
	The case $\theta=\infty$ (Friedrichs extension) selects $a_0=0$ in \eqref{eq:DomSak0}.
	\item[(iii)] For $-\lambda \in \rho(\mathsf{S}_{\alpha,0}^{(\theta)}) \cap \mathbb{R}$, $\mathsf{S}_{\alpha,0}^{(\theta)}$ has resolvent
	\begin{equation}\label{eq:Scr0res}
	 (\mathsf{S}_{\alpha,0}^{(\theta)}+\lambda)^{-1}\,=\,(\mathsf{S}_{\alpha,0}^{(\infty)}+\lambda)^{-1}+2m\tau_{\alpha,\lambda,\theta}^{(S)}\big| \Phi_{\alpha,\lambda}^{(S)} \big\rangle \big\langle \Phi_{\alpha,\lambda}^{(S)}\big|\,,
	\end{equation}
	where
	\begin{equation}
	 \begin{split}
	  (\mathsf{S}_{\alpha,0}^{(\theta)}+\lambda)^{-1}\,&=\,\textrm{integral operator with kernel $G_{\alpha,\lambda,c}^{(S)}(r,\rho)$, \;\;\; $r,\rho>0$}\,, \\
	  G_{\alpha,\lambda,c}^{(S)}(r,\rho)\,&:=\,2m\Big( \Phi_{\alpha,\lambda}^{(S)}(r)F_{\alpha,\lambda}^{(S)}(\rho)\mathbf{1}_{(0,r)}(\rho)+F_{\alpha,\lambda}^{(S)}(r)\Phi_{\alpha,\lambda}^{(S)}(\rho)\mathbf{1}_{(r,+\infty)}(\rho)\Big)\,,
	 \end{split}
	\end{equation}
	and where
\begin{equation}\label{eq:TauS}
	\tau_{\alpha,\lambda,\theta}^{(S)}\;:=\; \frac{2}{\pi \sec(\pi \alpha)+2^{\frac12+\alpha}  (m \lambda)^{-\frac12-\alpha} \Gamma(\frac12+\alpha) \Gamma(\frac32+\alpha)\theta} \, .
\end{equation}
%
\end{itemize}
\end{theorem}

 For later convenience, let us also add the analogous extension classification in the block $k=-1$.

\begin{theorem}\label{prop:k-1} Let $\alpha\in(-\frac{1}{2},\frac{1}{2})$.
	\begin{itemize}
		\item[(i)] Any $\psi \in \mathcal{D}(\mathsf{S}_{\alpha,-1}^*)$ displays the short-distance asymptotics
	\begin{equation}\label{eq:shortdasympt2}
		\psi(r) \stackrel{(r\downarrow 0)}{=} b_0 r^{\alpha}+ b_1 r^{1-\alpha} + o(r^{\frac{3}{2}})
	\end{equation}
	for $\psi$-dependent $b_0,b_1\in\mathbb{C}$.
		\item[(ii)] The self-adjoint extensions in $L^2(\mathbb{R}^+;\mathbb{C},\ud r)$ of the operator $\mathsf{S}_{\alpha,-1}$ form the one-parameter family $\big(\mathsf{S}_{\alpha,-1}^{(\nu)}\big)_{\nu\in \mathbb{R}\cup\{\infty\}}$ defined, with the notation of \eqref{eq:shortdasympt2}, by
	\begin{equation}\label{eq:DomSak-1}
		\begin{split}
		\mathcal{D}(\mathsf{S}_{\alpha,-1}^{(\nu)}) \,&:=\, \{ \psi \in \mathcal{D}(\mathsf{S}_{\alpha,-1}^*) \, | \, b_1 = \nu b_0\} \,, \\
		\mathsf{S}_{\alpha,-1}^{(\nu)}\,&:=\,\mathsf{S}_{\alpha,-1}^* \upharpoonright \mathcal{D}(\mathsf{S}_{\alpha,-1}^{(\nu)})\,.
		\end{split}
	\end{equation}
	The case $\nu=\infty$ (Friedrichs extension) selects $b_0=0$ in \eqref{eq:DomSak-1}.
	\end{itemize}
\end{theorem}

 As anticipated in Remark \ref{rem:squareDAB}, we conclude this Section by analysing the \emph{squares} of the Dirac-AB Hamiltonians in comparison to the Schr\"{o}dinger-AB Hamiltonians.


 \begin{theorem}\label{thm:DiracSquareSchrAB}
  Let $\alpha\in(-\frac{1}{2},\frac{1}{2})$ and $\gamma\in\mathbb{R}\cup\{\infty\}$. Then, with respect to the canonical isomorphism
 \[
  L^2(\mathbb{R}^+;\mathbb{C}^2,\ud r) \,\cong\,L^2(\mathbb{R}^+;\mathbb{C},\ud r) \oplus L^2(\mathbb{R}^+;\mathbb{C},\ud r)\,,
 \]
  one has:
  \begin{align}
    (\overline{\mathsf{h}_{\alpha,k}})^2 \,& =\,  2mc^2\big(\overline{\mathsf{S}_{\alpha,k}}\oplus \overline{\mathsf{S}_{-\alpha,-k}}\big)+m^2 c^4\mathbbm{1}\,,\qquad k\in\mathbb{Z}\setminus\{0,\pm 1\}\,, \label{eq:square01}\\
   (\overline{\mathsf{h}_{\alpha,-1}})^2\,& =\,2mc^2\big(\mathsf{S}_{\alpha,-1}^{(\infty)}\oplus\overline{\mathsf{S}_{-\alpha,1}}\big)+m^2c^4\mathbbm{1}\,,  \label{eq:square02} \\
   (\overline{\mathsf{h}_{\alpha,1}})^2\,& =\,2mc^2\big(\overline{\mathsf{S}_{\alpha,1}}\oplus\mathsf{S}_{-\alpha,-1}^{(\infty)}\big)+m^2c^4\mathbbm{1}\,, \label{eq:square03}
  \end{align}
  as well as
  \begin{equation}\label{eq:square04}
  \begin{split}
    \big(\mathsf{h}_{\alpha,0}^{(\gamma)}\big)^2\,&=\,\big(2mc^2\big(\mathsf{S}_{\alpha,0}^*\oplus\mathsf{S}_{\alpha,0}^*\big)+m^2c^4\mathbbm{1}\big)\upharpoonright\mathcal{D}\big(\big(\mathsf{h}_{\alpha,0}^{(\gamma)}\big)^2\big)\,, \\
    \mathcal{D}\big(\big(\mathsf{h}_{\alpha,0}^{(\gamma)}\big)^2\big)\,&=\, \left\{
    \begin{array}{l}
    g\in\mathcal{D}\big( \mathsf{S}_{\alpha,0}^*\oplus\mathsf{S}_{\alpha,0}^* \big)\,\left|
    \begin{array}{l}
    b_0\,=\,-\ii\gamma a_0\,, \\
	a_1\,=\,-\ii \gamma\,\displaystyle\frac{\,2 \ii a_0 mc+(1-2\alpha)b_1\,}{2\alpha+1}
    \end{array}\right. \\
    \textrm{where $a_0,a_1,b_0,b_1\in\mathbb{C}$ are the $g$-dependent constants} \\
    \textrm{characterised by the asymptotics} \\
     g(r)\,\stackrel{r\downarrow 0}{=}\,
   \begin{pmatrix}
    a_0 r^{-\alpha} + a_1 r^{1+\alpha} \\
    b_0 r^{\alpha} + b_1 r^{1-\alpha}
   \end{pmatrix}+o(r^{\frac{3}{2}})
    \end{array}
    \right\}.
  \end{split}
  \end{equation}
  In particular,
  \begin{align}
    \big(\mathsf{h}_{\alpha,k}^{(\infty)}\big)^2 \,& =\,2mc^2\big(\mathsf{S}_{\alpha,0}^{(\infty)}\oplus \mathsf{S}_{\alpha,0}^{(0)}\big)+m^2c^4\mathbbm{1}\,, \label{eq:squaregammainf}\\
   \big(\mathsf{h}_{\alpha,k}^{(0)}\big)^2 \,& =\, 2mc^2\big(\mathsf{S}_{\alpha,0}^{(0)}\oplus \mathsf{S}_{\alpha,0}^{(\infty)}\big)+m^2c^4\mathbbm{1}\,. \label{eq:squaregammazero}
  \end{align}
 \end{theorem}

 \begin{corollary}\label{cor:squaregammazeroinf}
  Let $\alpha\in(-\frac{1}{2},\frac{1}{2})$ and let $P_{\mathrm{el}}$ denote the orthogonal projection onto the first (`electron'-) component of
  \[
  L^2(\mathbb{R}^2;\mathbb{C}^2,\ud x\,\ud y) \,\cong\,L^2(\mathbb{R}^2;\mathbb{C},\ud x\,\ud y) \oplus L^2(\mathbb{R}^2;\mathbb{C},\ud x\,\ud y)\,.
 \]
  Then:
  \begin{align}
   P_{\mathrm{el}}\big(H_\alpha^{(\infty)}\big)^2 \,& =\, 2mc^2 S_{\alpha,0}^{(\infty)}+m^2c^4\,, \\
   P_{\mathrm{el}}\big(H_\alpha^{(0)}\big)^2 \,& =\, 2mc^2 S_{\alpha,0}^{(0)}+m^2c^4\,.
  \end{align}
 \end{corollary}

 \begin{proof}[Proof of Theorem \ref{thm:DiracSquareSchrAB}]
  On the one hand, observe that the $\overline{\mathsf{h}_{\alpha,k}}\,$'s for $k\neq 0$ and the $\mathsf{h}_{\alpha,0}^{(\gamma)}$'s for $k=0$ are all self-adjoint in $L^2(\mathbb{R}^+;\mathbb{C}^2;\ud r)$, hence their squares are self-adjoint and extend the corresponding $(\mathsf{h}_{\alpha,k})^2$'s (each $\mathsf{h}_{\alpha,k}$ being minimally defined as in \eqref{HalphaBlockDecomp}-\eqref{eq:Hsupersym}). On the other hand, a simple computation shows that, as an identity of differential actions,
 \begin{equation*}
  \begin{split}
	&\begin{pmatrix}
	mc^2 & \ii c \left( -\frac{\ud}{\ud r}+\frac{\alpha+k}{r}\right) \\
	-\ii c \left(\frac{\ud}{\ud r} + \frac{\alpha+k}{r} \right) & -mc^2
	\end{pmatrix}^2\,= \\
	&=\,
	\begin{pmatrix}
	 c^2 \big( -\frac{\ud^2}{\ud r^2} + \frac{(\alpha+k)(1+\alpha+k)}{r^2}  \big)+m^2c^4 & \mathbb{O} \\
	 \mathbb{O} & c^2 \big( -\frac{\ud^2}{\ud r^2} - \frac{(\alpha+k)(1-\alpha-k)}{r^2}  \big)+m^2c^4
	\end{pmatrix}.
  \end{split}
 \end{equation*}
 All together this implies that for each $k\in\mathbb{Z}$ the operator $(\overline{\mathsf{h}_{\alpha,k}})^2$ ($k\neq 0$), as well as the operator $(\mathsf{h}_{\alpha,0}^{(\gamma)})^2$, $\gamma\in\mathbb{R}\cup\{\infty\}$ ($k=0$), is a \emph{self-adjoint restriction}, for the corresponding $k$, of
 \[
  2mc^2\begin{pmatrix}
   \mathsf{S}_{\alpha,k}^* & \mathbb{O} \\
   \mathbb{O} & \mathsf{S}_{-\alpha,-k}^*
  \end{pmatrix}+m^2 c^4 \mathbbm{1}\,.
 \]
 In turn, the above operator is self-adjoint for $k\in\mathbb{Z}\setminus\{0,\pm 1\}$ (Proposition \ref{prop:KernelAdjointSCHROE}). Thus,
 \[
  (\overline{\mathsf{h}_{\alpha,k}})^2\,=\,
  2mc^2\begin{pmatrix}
   \overline{\mathsf{S}_{\alpha,k}} & \mathbb{O} \\
   \mathbb{O} & \overline{\mathsf{S}_{-\alpha,-k}}
  \end{pmatrix}+m^2 c^4 \mathbbm{1}\,,\qquad k\in\mathbb{Z}\setminus\{0,\pm 1\}\,.
 \]

  For the remaining blocks $k\in\{0,\pm 1\}$, recall from Theorems \ref{prop:k0}(i) and \ref{prop:k-1}(i) that
  \[\tag{*}\label{eq:ginstarS}
   g\in\mathcal{D}(\mathsf{S}_{\alpha,k}^*\oplus \mathsf{S}_{-\alpha,-k}^*)\qquad\Rightarrow\qquad g(r)\,\stackrel{r\downarrow 0}{=}\,
   \begin{pmatrix}
    a_0 r^{-\alpha-k} + a_1 r^{1+\alpha+k} \\
    b_0 r^{\alpha+k} + b_1 r^{1-\alpha-k}
   \end{pmatrix}+o(r^{\frac{3}{2}})
  \]
  (for $g$-dependent $a_0,a_1,b_0,b_1\in\mathbb{C}$).

  Take $k=-1$ and $g\in\mathcal{D}((\overline{\mathsf{h}_{\alpha,-1}})^2)$. The leading term of $g$ is square-integrable as $r\downarrow 0$, in view of \eqref{eq:ginstarS}, if and only if $b_0=0$. Furthermore, since $\mathcal{D}((\overline{\mathsf{h}_{\alpha,-1}})^2)$ is the domain of a self-adjoint restriction of $2mc^2(\mathsf{S}_{\alpha,-1}^*\oplus \mathsf{S}_{-\alpha,1}^*)+m^2c^4\mathbbm{1}=2mc^2(\mathsf{S}_{\alpha,-1}^*\oplus \overline{\mathsf{S}_{-\alpha,1}})+m^2c^4\mathbbm{1}$, necessarily also $a_1=\nu a_0$ for some $\nu\in\mathbb{R}\cup\{\infty\}$ (Theorem \ref{prop:k-1}(ii)), with the usual convention that $\nu=\infty$ corresponds to $a_0=0$ for any such $g$, whereas $b_1 r^{2-\alpha}=o(r^{\frac{3}{2}})$. Thus,
  \[
   g\in\mathcal{D}((\overline{\mathsf{h}_{\alpha,-1}})^2)\qquad\Rightarrow\qquad g\,\stackrel{r\downarrow 0}{=}\,
   a_0\begin{pmatrix}
     r^{1-\alpha} + \nu\, r^{\alpha} \\
    0
   \end{pmatrix}+o(r^{\frac{3}{2}})\,.
  \]
 On top of that, it must be $\overline{\mathsf{h}_{\alpha,-1}} \,g\in \mathcal{D}(\overline{\mathsf{h}_{\alpha,-1}})$, i.e. ($\overline{\mathsf{h}_{\alpha,-1}}$ being self-adjoint), $\mathsf{h}_{\alpha,-1}^*\,g\in\mathcal{D}(\mathsf{h}_{\alpha,-1}^*)$. The latter condition is tantamount as the square-integrability of
 \[
  \begin{pmatrix}
	mc^2 & \ii c \left( -\frac{\ud}{\ud r}+\frac{\alpha-1}{r}\right) \\
	-\ii c \left(\frac{\ud}{\ud r} + \frac{\alpha-1}{r} \right) & -mc^2
	\end{pmatrix} g\,\stackrel{r\downarrow 0}{=}\,
	a_0\begin{pmatrix}
	  m c^2(r^{1-\alpha}+\nu r^{\alpha}) \\
	 -\ii c \nu (2\alpha-1)r^{\alpha-1}
	\end{pmatrix}+o(r^{\frac{1}{2}})\,,
 \]
 where we used the preceding asymptotics for $g$. For the above expression to be square-integrable for \emph{all} considered $g$'s, necessarily $a_0=0$. This proves that
 \[
  \mathcal{D}((\overline{\mathsf{h}_{\alpha,-1}})^2)\,\subset\,\mathcal{D}(\mathsf{S}_{\alpha,-1}^{(\infty)}\oplus\overline{\mathsf{S}_{-\alpha,1}})
  \]
 and the inclusion above is an actual identity, by self-adjointness. The conclusion in this case is
 \[
  (\overline{\mathsf{h}_{\alpha,-1}})^2\,=\,2mc^2(\mathsf{S}_{\alpha,-1}^{(\infty)}\oplus\overline{\mathsf{S}_{-\alpha,1}})+m^2c^4\mathbbm{1}\,.
 \]

 The case $k=1$ is completely analogous.

 Last, let $k=0$, $\gamma\in\mathbb{R}\cup\{\infty\}$, and $g\in\mathcal{D}\big((\mathsf{h}_{\alpha,0}^{(\gamma)})^2\big)$. Now \eqref{eq:ginstarS} prescribes
 \[
  g(r)\,\stackrel{r\downarrow 0}{=}\,
   \begin{pmatrix}
    a_0 r^{-\alpha} + a_1 r^{1+\alpha} \\
    b_0 r^{\alpha} + b_1 r^{1-\alpha}
   \end{pmatrix}+o(r^{\frac{3}{2}})
 \]
  (for $g$-dependent $a_0,a_1,b_0,b_1\in\mathbb{C}$). Since $g\in\mathcal{D}(\mathsf{h}_{\alpha,0}^{(\gamma)})$, then $b_0=-\ii\gamma a_0$ (Theorem \ref{prop:Dirac}(iii)).
  The latter asymptotics then yield
  \[
   \mathsf{h}_{\alpha,0}^{(\gamma)}\,g\,=\,\begin{pmatrix}
	mc^2 & \ii c \left( -\frac{\ud}{\ud r}+\frac{\alpha}{r}\right) \\
	-\ii c \left(\frac{\ud}{\ud r} + \frac{\alpha}{r} \right) & -mc^2
	\end{pmatrix} g\,\stackrel{r\downarrow 0}{=}\,
	\begin{pmatrix}
	 \big(a_0 mc^2 + \ii c (2\alpha-1) b_1\big)r^{-\alpha} \\
	 \big(\ii \gamma a_0 mc^2 - \ii c (2\alpha+1) a_1\big)r^{\alpha}
	\end{pmatrix}+o(r^{\frac{1}{2}})\,.
  \]
  Since $\mathsf{h}_{\alpha,0}^{(\gamma)}\,g\in \mathcal{D}(\mathsf{h}_{\alpha,0}^{(\gamma)})$, then
  \[
   \ii \gamma a_0 mc^2 - \ii c (2\alpha+1) a_1\,=\,-\ii\gamma\big(a_0 mc^2 + \ii c (2\alpha-1) b_1\big)
  \]
 (owing again to Theorem \ref{prop:Dirac}(iii)), whence
 \[
  a_1\,=\,-\ii \gamma\,\frac{\,2 \ii a_0 mc+(1-2\alpha)b_1\,}{2\alpha+1}\,.
 \]

 This completes the proof of \eqref{eq:square01}-\eqref{eq:square04}. The identities \eqref{eq:squaregammainf}-\eqref{eq:squaregammazero} are consequence of \eqref{eq:square04}. Indeed, when $\gamma=\infty$ \eqref{eq:square04} yields $a_0=0$, $b_1=0$, and Theorem \ref{prop:k0}(ii) then yields
  \[
   g\,=\,\begin{pmatrix} g_+ \\ g_- \end{pmatrix} \qquad\textrm{with}\; g_+\in\mathcal{D}\big(\mathsf{S}_{\alpha,0}^{(\infty)}\big)\,,\;g_-\in\mathcal{D}\big(\mathsf{S}_{\alpha,0}^{(0)}\big)\,.
  \]
  Conversely, when $\gamma=0$ \eqref{eq:square04} yields $b_0=0$, $a_1=0$, and Theorem \ref{prop:k0}(ii) then yields
  \[
   g\,=\,\begin{pmatrix} g_+ \\ g_- \end{pmatrix} \qquad\textrm{with}\; g_+\in\mathcal{D}\big(\mathsf{S}_{\alpha,0}^{(0)}\big)\,,\;g_-\in\mathcal{D}\big(\mathsf{S}_{\alpha,\infty}^{(0)}\big)\,.
  \]
  The proof is now complete.
 \end{proof}

    \begin{remark}
   As the proof of Theorem \ref{thm:DiracSquareSchrAB} shows, $(\overline{\mathsf{h}_{\alpha,k}})^2$ ($k\neq 0$) or $(\mathsf{h}_{\alpha,0}^{(\gamma)})^2$ ($k=0$), is a self-adjoint restriction of
 \[\tag{*}\label{eq:selfadjrestrof}
  2mc^2\begin{pmatrix}
   \mathsf{S}_{\alpha,k}^* & \mathbb{O} \\
   \mathbb{O} & \mathsf{S}_{-\alpha,-k}^*
  \end{pmatrix}+m^2 c^4 \mathbbm{1}\,.
 \]
  When $k=\pm 1$, the self-adjoint restrictions of \eqref{eq:selfadjrestrof} form a one-real-parameter family, because only one block is non-self-adjoint (Proposition \ref{prop:KernelAdjointSCHROE} and Theorem \ref{prop:k-1}(ii)). Thus, $(\overline{\mathsf{h}_{\alpha,k}})^2$ is always reduced with respect to
  \[
  L^2(\mathbb{R}^+;\mathbb{C}^2,\ud r) \,\cong\,L^2(\mathbb{R}^+;\mathbb{C},\ud r) \oplus L^2(\mathbb{R}^+;\mathbb{C},\ud r)\,,
 \]
 as \eqref{eq:square02}-\eqref{eq:square03} indeed show. Instead, when $k=0$, the self-adjoint restrictions of \eqref{eq:selfadjrestrof} form a four-real-parameter family and are not reduced in general, which is reflected in the domain \eqref{eq:square04}. The cases $\gamma=\infty$ and $\gamma=0$ are the sole case of reducibility of $\mathsf{h}_{\alpha,k}^{(\gamma)}$ -- see \eqref{eq:squaregammainf}-\eqref{eq:squaregammazero}.
  \end{remark}

  \begin{proof}[Proof of Corollary \ref{cor:squaregammazeroinf}]
  Immediate from Theorem \ref{thm:DiracSquareSchrAB}, in particular from \eqref{eq:squaregammainf}-\eqref{eq:squaregammazero}, in view of the general expressions \eqref{eq:familyofDAB} and \eqref{eq:S-AB-Sa-Intro}, respectively, for $H_\alpha^{(\gamma)}$ and $S_\alpha^{(\theta)}$. Indeed, the square is taken block-wise.
 \end{proof}

\section{Proof of Theorem \ref{thm:Main}}\label{sec:Proof3}

Based on the general set-up of Section \ref{sec:intro}, it is clear that it suffices to establish the operator limit \eqref{eq:NonRelLimThmMain} separately in each $k$-block.

 Recall that the block $k=0$ is the one that accommodates non-trivial self-adjoint extensions both before (Dirac-AB) and after (Schr\"{o}dinger-AB) the limit, and the block $k=-1$ has non-trivial self-adjoint extensions after the limit. These blocks then require a special treatment, which would lead to the special scaling of the extension parameter as $c\to +\infty$ in order to connect self-adjoint realisations before and after the limit.

 We start the preparation to the proof of Theorem \ref{thm:Main} with a closer analysis of the operator closure of $\mathsf{h}_{\alpha,k}$ in the non-zero blocks and of the Friedrichs extension in the block $k=0$ (Propositions \ref{cor:Closure_k_neq0} and \ref{prop:hfried0} below). We tacitly exploit as usual the canonical Hilbert space orthogonal direct sum
 \begin{equation}\label{eq:canHspdec}
  L^2(\mathbb{R}^+;\mathbb{C}^2,\ud r) \,\cong\,L^2(\mathbb{R}^+;\mathbb{C},\ud r) \oplus L^2(\mathbb{R}^+;\mathbb{C},\ud r)\,.
 \end{equation}

\begin{proposition}\label{cor:Closure_k_neq0}
 When $k\in\mathbb{Z}\setminus\{0\}$ one has
 \begin{equation}\label{eq:halkclos}
	\overline{\mathsf{h}_{\alpha,k}} \,=\, \begin{pmatrix}
		mc^2 & cB_{\alpha,k}^+ \\
		cB_{\alpha,k} & - mc^2
	\end{pmatrix}\qquad\qquad (k\in\mathbb{Z}\setminus\{0\})
\end{equation}
with respect to the decomposition \eqref{eq:canHspdec}
 and in terms of the two \emph{closed} operators $B_{\alpha,k}$ and $B^+_{\alpha,k}$ in $L^2(\mathbb{R}^+;\mathbb{C},\ud r)$ defined by
 \begin{eqnarray}
  \mathcal{D}(B_{\alpha,k})\,:=\,H^1_0(\mathbb{R}^+;\mathbb{C})\,,& & \quad B_{\alpha,k}\psi\,:=\;-\ii \Big(\frac{\ud}{\ud r}+\frac{\alpha+k}{r} \Big)\psi\,, \label{eq:Bgeneric} \\
  \mathcal{D}(B^+_{\alpha,k})\,:=\,H^1_0(\mathbb{R}^+;\mathbb{C})\,,& & \quad B^+_{\alpha,k}\psi\,:=\;-\ii \Big(\frac{\ud}{\ud r}-\frac{\alpha+k}{r} \Big)\psi\,. \label{eq:B+generic}
 \end{eqnarray}
\end{proposition}

\begin{proof}
 The differential action of $\overline{\mathsf{h}_{\alpha,k}}$ is immediately read out from \eqref{eq:Hsupersym}. Thus, \emph{as an identity between differential actions},
 \[
  \overline{\mathsf{h}_{\alpha,k}}-mc^2\sigma_3\,=\,\begin{pmatrix}
		\mathbb{O} & cB_{\alpha,k}^+ \\
		cB_{\alpha,k} & \mathbb{O}
	\end{pmatrix}.
 \]
 In the l.h.s.~above the domain is
 \[
  \mathcal{D}(\overline{\mathsf{h}_{\alpha,k}}-mc^2\sigma_3)\,=\,\mathcal{D}(\overline{\mathsf{h}_{\alpha,k}})\,=\,H^1_0(\mathbb{R}^+;\mathbb{C}^2)\,,
 \]
 owing to the boundedness of $mc^2\sigma_3$ and to Theorem \ref{prop:Dirac}(i). Then, on this very domain, the r.h.s.~is closed too, and the corresponding domain of both $B_{\alpha,k}$ and $B_{\alpha,k}^+$ as operators in $L^2(\mathbb{R}^+;\mathbb{C},\ud r)$ is $H^1_0(\mathbb{R}^+;\mathbb{C})$.
\end{proof}

 \begin{proposition}\label{prop:hfried0}
  One has
  \begin{equation}\label{eq:halpha0D}
	\mathsf{h}_{\alpha,0}^{(\infty)} \,=\, \begin{pmatrix}
		mc^2 & cB_{\alpha,0}^+ \\
		cB_{\alpha,0} & -mc^2
	\end{pmatrix}
\end{equation}
with respect to the decomposition \eqref{eq:canHspdec}
 and in terms of the two \emph{closed} operators $B_{\alpha,0}$ and $B^+_{\alpha,0}$ in $L^2(\mathbb{R}^+;\mathbb{C},\ud r)$ defined by
\begin{eqnarray}
  \mathcal{D}(B_{\alpha,0})\,:=\,H^1_0(\mathbb{R}^+;\mathbb{C})\,,& & \quad B_{\alpha,0}\psi\,:=\;-\ii \Big(\frac{\ud}{\ud r}+\frac{\alpha}{r} \Big)\psi\,, \label{eq:Bmenozero} \\
  \mathcal{D}(B^+_{\alpha,0})\,:=\,H^1_0(\mathbb{R}^+;\mathbb{C})\,\dot{+}\, \mathrm{span}_{\mathbb{C}}\{\sqrt{r}\, K_{\alpha-\frac{1}{2}}\}\,,& & \quad B^+_{\alpha,0}\psi\,:=\;-\ii \Big(\frac{\ud}{\ud r}-\frac{\alpha}{r} \Big)\psi\,. \label{eq:Bpiuzero}
 \end{eqnarray}
 \end{proposition}

 \begin{proof}
  In the general classification of the self-adjoint extensions of $\mathsf{h}_{\alpha,0}$, see \eqref{eq:ginDstar} and \eqref{eq:halphb} above,
  \[
   \mathcal{D}(\mathsf{h}_{\alpha,0}^{(\infty)})\,=\,\mathcal{D}(\overline{\mathsf{h}_{\alpha,0}})\,\dot{+}\,\mathrm{span}_{\mathbb{C}}\{\Psi_{\alpha,\lambda,c}^{(D)}\}\,,
  \]
   where (Proposition \ref{prop:hD}(iv))
  \[
   \Psi_{\alpha,\lambda,c}^{(D)}\,=\,(\mathsf{h}_{\alpha,0}^{(\infty)}-mc^2+\lambda)^{-1}\Phi_{\alpha,\lambda,c}^{(D)}
  \]
  and (Proposition \ref{prop:Closure}(i))
  \[
   \mathcal{D}(\overline{\mathsf{h}_{\alpha,0}})\,=\,H^1_0(\mathbb{R}^+;\mathbb{C}^2)	\,.
  \]

  Let us split
  \[
   \Psi_{\alpha,\lambda,c}^{(D)}\,=\,\begin{pmatrix} 0 \\ q_\alpha \sqrt{r} K_{\alpha-\frac{1}{2}} \end{pmatrix}+\Xi_{\alpha,\lambda,c}
  \]
 with
  \[
   \begin{split}
    q_\alpha\,&:=\,-2^{\frac{1}{2}+\alpha}\,\frac{\lambda}{\,c^2} \frac{  C_{\alpha,\lambda,c}}{\, \Gamma(\frac{1}{2}-\alpha)} \Vert \Phi_{\alpha,\lambda,c}^{(D)} \Vert_{L^2(\mathbb{R}^+;\mathbb{C}^2)}^2 \,.
   \end{split}
  \]
 Since \cite[Eq.~(9.6.1) and (9.6.10)]{Abramowitz-Stegun-1964}
\begin{equation*}
	\sqrt{r} K_{\alpha-\frac{1}{2}}(r)\, \stackrel{(r\downarrow 0)}{=}\,
	2^{-\alpha
   -\frac{1}{2}} \Gamma \left({\textstyle\frac{1}{2}-\alpha}
   \right) r^{\alpha}+2^{\alpha -\frac{3}{2}} \Gamma \left(\alpha
   -\frac{1}{2}\right) r^{1-\alpha} + O(r^{2+\alpha}) \, ,
   \end{equation*}
 then $q_\alpha \sqrt{r} K_{\alpha-\frac{1}{2}}$ has the same $r^{\alpha}$-leading term as the asymptotics \eqref{eq:AsintoticaPsiZero} for $\Psi_{\alpha,\lambda,c}^{(D)}$, whence
 \[
  \Xi_{\alpha,\lambda,c} \, \stackrel{(r\downarrow 0)}{=}\, o(r^{\frac{1}{2}})\,.
 \]

 In fact, we shall now argue that
 \[
  \Xi_{\alpha,\lambda,c} \,\in\, H^1_0(\mathbb{R}^+;\mathbb{C}^2)\,.
 \]
 Clearly, $\Xi_{\alpha,\lambda,c}\in C^\infty(\mathbb{R}^+;\mathbb{C}^2)$, since both $\Psi_{\alpha,\lambda,c}^{(D)}$ and $K_{\alpha-\frac{1}{2}}(r)$ are smooth functions on $\mathbb{R}^+$. Moreover, in the notation $\Psi_{\alpha,\lambda,c}^{(D)}\equiv\begin{pmatrix} \Psi_+ \\ \Psi_-\end{pmatrix}$, $\Phi_{\alpha,\lambda,c}^{(D)}=\begin{pmatrix} \Phi_+ \\ \Phi_-\end{pmatrix}$, the identity $\Psi_{\alpha,\lambda,c}^{(D)}=(\mathsf{h}_{\alpha,0}^{(\infty)}-mc^2+\lambda)^{-1}\Phi_{\alpha,\lambda,c}^{(D)}$ reads
 \begin{equation*}
 \begin{cases}
  \;- \ii c \frac{\ud \Psi_-}{\ud r} \,=\, mc^2 \Psi_+- \ii c \frac{\alpha}{r} \Psi_- + \Phi_+\,, \\
  \;-\ii c \frac{\ud \Psi_+}{\ud r} \,=\, \ii c \frac{\alpha}{r} \Psi_+ -mc^2 \Psi_-+ \Phi_-\,.
 \end{cases}
\end{equation*}
 In the right-hand sides above, all functions are square-integrable on any interval $[\varepsilon,+\infty]$, $\varepsilon>0$ (the $\alpha/r$-terms are bounded as $r\geqslant \varepsilon$), implying
   \[
   \Psi_{\alpha,\lambda,c}^{(D)}\big|_{[\varepsilon,+\infty)}\,\in\,H^1([\varepsilon,+\infty);\mathbb{C}^2)\qquad\forall\varepsilon>0\,.
  \]
  As a matter of fact, also
  \[
   K_{\alpha-\frac{1}{2}}\big|_{[\varepsilon,+\infty)}\,\in\,H^1([\varepsilon,+\infty);\mathbb{C})\qquad\forall\varepsilon>0\,.
  \]
  Indeed, $\sqrt{r} K_{\alpha-\frac12} \in H^1_{\text{loc}}(\mathbb{R}^+;\mathbb{C})$ because $K_{\alpha-\frac12} \in C^\infty(\mathbb{R}^+;\mathbb{C})$, and the exponential decay of $K_{\nu}$ at infinity  \cite[Eq.~(9.7.2)]{Abramowitz-Stegun-1964} together with the recurrence relation \cite[Eq.~(9.6.26)-(iv)]{Abramowitz-Stegun-1964}
\begin{equation*}
	K_{\alpha-\frac{1}{2}}'(r)\,=\,-\frac{1}{2}\Big( K_{-\frac{1}{2}-\alpha}(r)+K_{\frac{3}{2}-\alpha}(r) \Big) \, ,
\end{equation*}
 imply that $\sqrt{r} K_{\alpha-\frac12}$ is an $H^1$-function at infinity. Since both $\Psi_{\alpha,\lambda,c}^{(D)}\big|_{[\varepsilon,+\infty)}$ and $K_{\alpha-\frac{1}{2}}\big|_{[\varepsilon,+\infty)}$ are $H^1$-functions, so is $\Xi_{\alpha,\lambda,c}\big|_{[\varepsilon,+\infty)}$. But since $\Xi_{\alpha,\lambda,c} \stackrel{(r\downarrow 0)}{=} o(r^{\frac{1}{2}})$, then $\Xi_{\alpha,\lambda,c} \in H^1_0(\mathbb{R}^+;\mathbb{C}^2)$ as claimed.

  Therefore,
\begin{equation*}
	\begin{split}
		\mathcal{D}(\mathsf{h}_{\alpha,0}^{(\infty)})&=H^1_0(\mathbb{R}^+;\mathbb{C}^2)\,\dot{+} \,\mathrm{span}_{\mathbb{C}}\{\Psi_{\alpha,\lambda,c}^{(D)}\} \\
		&=H^1_0(\mathbb{R}^+;\mathbb{C}^2)\,\dot{+}\, \mathrm{span}_{\mathbb{C}}\bigg\{q_\alpha \sqrt{r} K_{\alpha-\frac{1}{2}}(r)\begin{pmatrix} 0 \\ 1 \end{pmatrix}+R_\alpha(r)\bigg\}\\
		&=H^1_0(\mathbb{R}^+;\mathbb{C}^2)\,\dot{+}\, \mathrm{span}_{\mathbb{C}}\bigg\{\sqrt{r} K_{\alpha-\frac{1}{2}}(r)\begin{pmatrix} 0 \\ 1 \end{pmatrix}\bigg\} \, .
	\end{split}
\end{equation*}

Now, \eqref{eq:halpha0D}, \emph{as an identity of differential actions}, follows directly from \eqref{eq:Hsupersym} and is equivalent to
 \[
  \mathsf{h}_{\alpha,0}^{(\infty)}-mc^2\sigma_3 \,=\, \begin{pmatrix}
		\mathbb{O} & cB_{\alpha,0}^+ \\
		cB_{\alpha,0} & \mathbb{O}
	\end{pmatrix}.
 \]
 Since the domain of the l.h.s.~is
 \[
  \mathcal{D}(\mathsf{h}_{\alpha,0}^{(\infty)}-mc^2\sigma_3)\,=\,\mathcal{D}(\mathsf{h}_{\alpha,0}^{(\infty)})\,=\,H^1_0(\mathbb{R}^+;\mathbb{C}^2)\,\dot{+}\, \mathrm{span}_{\mathbb{C}}\bigg\{\sqrt{r} K_{\alpha-\frac{1}{2}}(r)\begin{pmatrix} 0 \\ 1 \end{pmatrix}\bigg\}\,,
 \]
 then the domains of the non-zero blocks in the r.h.s.~must be the ones indicated in \eqref{eq:Bmenozero}-\eqref{eq:Bpiuzero}.
 \end{proof}

 As a matter of fact, the above $B^+_{\alpha,k}$'s are nothing but the adjoints in $L^2(\mathbb{R}^+;\mathbb{C},\ud r)$ of the corresponding $B_{\alpha,k}$'s.

 \begin{proposition}\label{prop:B+=B*}
  $B^+_{\alpha,k}=(B_{\alpha,k})^*$ $\;\forall k \in \mathbb{Z}$, where
   \begin{equation}\label{eq:Balphakstar}
	\begin{split}
		\mathcal{D}((B_{\alpha,k})^*)\,&=\,\Big\{\psi \in L^2(\mathbb{R}^+;\mathbb{C},\ud r) \, \Big| \, \Big(\frac{\ud}{\ud r}-\frac{\alpha+k}{r}\Big) \psi \in L^2(\mathbb{R}^+;\mathbb{C},\ud r) \Big\} \,,\\
		(B_{\alpha,k})^* \psi \,&=\, -\ii \Big(\frac{\ud}{\ud r}-\frac{\alpha+k}{r} \Big)\psi\,.
	\end{split}
\end{equation}
 \end{proposition}

  \begin{proof}
   Formula \eqref{eq:Balphakstar} follows by general facts (see, e.g., \cite[Lemma 4.3]{Grubb-DistributionsAndOperators-2009}) directly from \eqref{eq:Bgeneric} and \eqref{eq:Bmenozero}.

   As a follow-up of (the proof of) Proposition \ref{prop:hfried0}, one sees that on the domain
    \[
  \mathcal{D}(\mathsf{h}_{\alpha,0}^{(\infty)}-mc^2\sigma_3)\,=\,\mathcal{D}(\mathsf{h}_{\alpha,0}^{(\infty)})\,=\,H^1_0(\mathbb{R}^+;\mathbb{C}^2)\,\dot{+}\, \mathrm{span}_{\mathbb{C}}\bigg\{\sqrt{r} K_{\alpha-\frac{1}{2}}(r)\begin{pmatrix} 0 \\ 1 \end{pmatrix}\bigg\}
 \]
  the left-hand side of
     \[
  \mathsf{h}_{\alpha,0}^{(\infty)}-mc^2\sigma_3 \,=\, \begin{pmatrix}
		\mathbb{O} & cB_{\alpha,0}^+ \\
		cB_{\alpha,0} & \mathbb{O}
	\end{pmatrix}.
 \]
 is self-adjoint. Then necessarily $B_{\alpha,0}$ and $B_{\alpha,0}^+$ are mutually adjoint.
  \end{proof}

With the above preparations, we now control the non-relativistic limit $c\to +\infty$ \emph{of the operator closures} in each $k$-block.

\begin{proposition}\label{prop:LimitFriedrichs}
	\textbf{}\begin{itemize}
		\item[(i)] If $k\in\mathbb{Z}\setminus\{-1,0\}$, then $\forall c > 0$ $\rho(\overline{\mathsf{h}_{\alpha,k}}-mc^2) \subset \rho(\overline{\mathsf{S}_{\alpha,k }})$ and $\forall z \in \rho(\overline{\mathsf{h}_{\alpha,k}}-mc^2)$,
		\begin{equation}\label{eq:clim-non0nonm1}
			\lim_{c \to +\infty} (\overline{\mathsf{h}_{\alpha,k}}-mc^2-z)^{-1} \;=\; (\overline{\mathsf{S}_{\alpha,k}}-z)^{-1} \oplus_{\mathbb{C}^2} \mathbb{O} \, .
		\end{equation} 
		\item[(ii)] $\forall c > 0$ $\rho(\mathsf{h}_{\alpha,0}^{(\infty)}-mc^2) \subset \rho(\mathsf{S}_{\alpha,0}^{(\infty)})$ and for any $z \in \rho(\mathsf{h}_{\alpha,0}^{(\infty)}-mc^2)$, we have
		\begin{equation}\label{eq:LimitResolventFriedrichs1}
			\lim_{c \to +\infty} (\mathsf{h}_{\alpha,0}^{(\infty)} - mc^2-z)^{-1} = (\mathsf{S}_{\alpha,0}^{(\infty)}-z)^{-1}\oplus_{\mathbb{C}^2} \mathbb{O} \, .
		\end{equation}
		\item[(iii)] $\forall c > 0$ $ \rho(\overline{\mathsf{h}_{\alpha,-1}}-mc^2) \subset \rho(\mathsf{S}_{\alpha,-1}^{(\infty)})$ and for any $z \in \rho(\overline{\mathsf{h}_{\alpha,-1}}-mc^2)$, we have
		\begin{equation}\label{eq:clim-m1}
			\lim_{c \to +\infty} (\overline{\mathsf{h}_{\alpha,-1}}-mc^2-z)^{-1} = (\mathsf{S}_{\alpha,-1}^{(\infty)}-z)^{-1} \oplus_{\mathbb{C}^2} \mathbb{O}\,.
		\end{equation}
	\end{itemize}
	All limits above are meant in operator norm.
\end{proposition}

\begin{proof} The result in all three cases relies on the operator-norm limit, $\forall k\in\mathbb{Z}$,
\begin{equation}\tag{a}\label{eq:Thallerlimit}
	\lim_{c \to +\infty} \left[\begin{pmatrix}
		mc^2 & (cB_{\alpha,k})^* \\
		cB_{\alpha,k} & - mc^2
	\end{pmatrix}  - mc^2 - z\right]^{-1} = \left[\frac{1}{2m} (B_{\alpha,k})^*B_{\alpha,k}-z\right]^{-1} \oplus \mathbb{O} \, ,
\end{equation}
a simple but fundamental classical result discussed, e.g., in \cite[Corollary 6.2]{Thaller-Dirac-1992}.

Observe also that the differential action of $(B_{\alpha,k})^*B_{\alpha,k}$ (namely $B^+_{\alpha,k}B_{\alpha,k}$, in view of Proposition \ref{prop:B+=B*}) and of $\mathsf{S}_{\alpha,k }^*$ is the same and amounts to
\[
\frac{1}{2m} \Big( -\frac{\ud^2}{\ud r^2}+\frac{(\alpha+k)(\alpha+k+1)}{r^2}\Big)\,,
\]
as one deduces from \eqref{eq:Bgeneric}-\eqref{eq:B+generic} and \eqref{eq:Bmenozero}-\eqref{eq:Bpiuzero}.
Moreover, since both $B_{\alpha,k}$ and $(B_{\alpha,k})^*$ are closed operators, their product $(B_{\alpha,k})^*B_{\alpha,k}$ is self-adjoint.

Owing to \eqref{eq:halkclos}, \eqref{eq:halpha0D}, and Proposition \ref{prop:B+=B*}, the left-hand side of \eqref{eq:Thallerlimit} has exactly the form of the left-hand side of \eqref{eq:clim-non0nonm1}, of \eqref{eq:LimitResolventFriedrichs1}, and of \eqref{eq:clim-m1}. Therefore, in order to establish such three formulas, one has to show that
\begin{align}
  \frac{1}{2m} (B_{\alpha,k})^*B_{\alpha,k}\,&=\,\overline{\mathsf{S}_{\alpha,k }}\,,\qquad k\in\mathbb{Z}\setminus\{0,-1\}\,, \tag{b}\label{eq:BBc1} \\
  \frac{1}{2m} (B_{\alpha,0})^*B_{\alpha,0}\,&=\,\mathsf{S}_{\alpha,0}^{(\infty)}\,, \tag{c}\label{eq:BBc0} \\
   \frac{1}{2m} (B_{\alpha,-1})^*B_{\alpha,-1}\,&=\,\mathsf{S}_{\alpha,-1}^{(\infty)}\,. \tag{d}\label{eq:BBc2}
\end{align}

For all cases $k\in\mathbb{Z}\setminus\{0,-1\}$, the self-adjoint operator $(2m)^{-1}(B_{\alpha,k})^*B_{\alpha,k}$, having the same differential action as $\mathsf{S}_{\alpha,k }$, must be a self-adjoint extension of $\mathsf{S}_{\alpha,k }$ itself. Then \eqref{eq:BBc1} follows because $\mathsf{S}_{\alpha,k }$ is essentially self-adjoint.

Concerning $k=0$ and $k=-1$, also in these cases $(2m)^{-1}(B_{\alpha,k})^*B_{\alpha,k}$ is a self-adjoint extension of $\mathsf{S}_{\alpha,k}$. In order to prove that this extension is precisely $\mathsf{S}_{\alpha,k}^{(\infty)}$, and hence to establish \eqref{eq:BBc0} and \eqref{eq:BBc2}, we now show that the only self-adjoint extension of $\mathsf{S}_{\alpha,k}$ whose domain is \emph{contained} in $\mathcal{D}((B_{\alpha,k})^*B_{\alpha,k})$ is actually $\mathsf{S}_{\alpha,k}^{(\infty)}$.

So, when $k=0$, take $\psi\in\mathcal{D}(\mathsf{S}_{\alpha,0}^{(\theta)})$, a function in the domain of a generic member of the family of extensions described in Theorem \ref{prop:k0}. Due to \eqref{eq:shortdasympt}-\eqref{eq:DomSak0},
\[
 \psi(r) \stackrel{(r\downarrow 0)}{=} \theta^{-1} a_1 r^{-\alpha}+ a_1 r^{1+\alpha} + o(r^{\frac{3}{2}})\,.
\]
On the other hand, due to \eqref{eq:Bmenozero}-\eqref{eq:Bpiuzero},
\[
 \mathcal{D}((B_{\alpha,0})^*B_{\alpha,0})\;=\;\Big\{ f\in H^1_0(\mathbb{R}^+)\,\Big|\,\Big(\frac{\ud}{\ud r}-\frac{\alpha}{r}\Big) f \in H^1_0(\mathbb{R}^+)\dot{+}\, \mathrm{span}_{\mathbb{C}}\{\sqrt{r}\, K_{\alpha-\frac{1}{2}}\}\Big\}\,.
\]
As argued already, the short-distance asymptotics for $H^1_0(\mathbb{R}^+)$-functions is $\sim o(r^{\frac{1}{2}})$, and for the function $\sqrt{r} K_{\alpha-\frac12}$ it is

\begin{equation*}
 \sqrt{r} K_{\alpha-\frac12}(r)\,\stackrel{(r\downarrow 0)}{=}\,d_{\alpha,k}\,\Big( 2^{\alpha-\frac{1}{2}}\,r^{-\alpha+1}+2^{-\alpha+\frac{1}{2}}\,r^{\alpha} \Big)\big(1+o(r^2))
\end{equation*}
(see the proof of Proposition \ref{prop:B+=B*}). One thus sees that, in order for the above $\psi$ to belong to $\mathcal{D}((B_{\alpha,0})^*B_{\alpha,0})$, the part of $\psi$ behaving as $r^{-\alpha}$ as $r\downarrow 0$ must be absent, for otherwise the function $(\frac{\ud}{\ud r}-\frac{\alpha}{r})\psi$ would display an incompatible $r^{-\alpha-1}$ singularity. Requiring this to hold for all $\psi$'s in $\mathcal{D}(\mathsf{S}_{\alpha,0}^{(\theta)})$, implies $\theta=\infty$, which selects the extension $\mathsf{S}_{\alpha,0}^{(\infty)}$. Then also \eqref{eq:BBc0} is established.

The remaining case $k=-1$ is treated analogously. Take now $\psi\in\mathcal{D}(\mathsf{S}_{\alpha,-1}^{(\nu)})$, a function in the domain of a generic member of the family of extensions described in Theorem \ref{prop:k-1}. Due to \eqref{eq:shortdasympt2}-\eqref{eq:DomSak-1},
\[
 \psi(r) \stackrel{(r\downarrow 0)}{=} \nu^{-1} b_1 r^{\alpha}+ b_1 r^{1-\alpha} + o(r^{\frac{3}{2}})\,,
\]
and due to \eqref{eq:Bgeneric}-\eqref{eq:B+generic},
\[
 \mathcal{D}((B_{\alpha,-1})^*B_{\alpha,-1})\;=\;\Big\{ f\in H^1_0(\mathbb{R}^+)\,\Big|\,\Big(\frac{\ud}{\ud r}-\frac{\alpha-1}{r}\Big) f \in H^1_0(\mathbb{R}^+)\Big\}\,.
\]
Thus, in order for $\psi$ to belong to $\mathcal{D}((B_{\alpha,-1})^*B_{\alpha,-1})$, the part of $\psi$ behaving as $r^\alpha$ as $r\downarrow 0$ must be absent, otherwise the function $(\frac{\ud}{\ud r}-\frac{\alpha-1}{r})\psi$ would display an incompatible $r^{\alpha-1}$ singularity. Requiring this to hold for all $\psi$'s in $\mathcal{D}(\mathsf{S}_{\alpha,-1}^{(\nu)})$, implies $\nu=\infty$, which selects the extension $\mathsf{S}_{\alpha,-1}^{(\infty)}$. Then also \eqref{eq:BBc2} is established.

All three limits \eqref{eq:clim-non0nonm1}-\eqref{eq:clim-m1} are proved.
\end{proof}

 Two more bits of information are needed to identify the limits of generic self-adjoint realisations in the non-trivial block $k=0$. The first concerns the connection in the limit between the two pre-factors $\tau_{\alpha,\lambda,\gamma,c}^{(D)}$ and $\tau_{\alpha,\lambda,\theta}^{(S)}$ introduced in \eqref{eq:TauDGamma} and \eqref{eq:TauS}, respectively.

\begin{lemma}\label{lem:tauprefactors}
 Let $\alpha\in(-\frac{1}{2},\frac{1}{2})$. For given $\theta\in\mathbb{R}\cup\{\infty\}$, let $\gamma=\gamma(c)$ scale with $c$ according to \eqref{eq:LimitGammaTheta}, that is,
 \begin{equation*}
		\lim_{c \to +\infty} \frac{\,2\, m\, c \,\gamma(c)\,}{1+ 2\alpha} = \theta \, .
	\end{equation*}
 Then
 \begin{equation}
   \lim_{c \to +\infty}\tau_{\alpha,\lambda,\gamma(c),c}^{(D)}\,=\,\tau_{\alpha,\lambda,\theta}^{(S)}\,.
 \end{equation}
\end{lemma}

\begin{proof}
 In the expression
 \[
  (\tau_{\alpha,\lambda,\gamma(c),c}^{(D)})^{-1}\;=\; \frac{1}{2}\Big( \pi \sec(\pi \alpha)+4^\alpha \mu_{c,\lambda}^2  \big(\mu_{c,\lambda} \sqrt{\lambda} \big)^{-1-2\alpha} \,\Gamma^2({\textstyle\frac{1}{2}+\alpha}) c \gamma(c)\Big)
 \]
 we observe that $\mu_{c,\lambda}\xrightarrow{c\to+\infty}\sqrt{2m}$ and hence
 \begin{equation*}
	\lim_{c \to +\infty}4^\alpha \mu_{c,\lambda}^2  \big( \mu_{c,\lambda}\sqrt{\lambda} \big)^{-1-2\alpha} \,= \,2^{\frac{1}{2}+\alpha} m^{\frac{1}{2}-\alpha} \lambda^{-\frac{1}{2}-\alpha}\,.
\end{equation*}
Moreover, since $(\frac{1}{2}+\alpha)\Gamma(\frac{1}{2}+\alpha)=\Gamma(\frac{3}{2}+\alpha)$ (\cite[Eq.~(6.1.15)]{Abramowitz-Stegun-1964}),
\[
 \Gamma^2\Big(\frac{1}{2}+\alpha\Big)\,=\,\Gamma\Big(\frac{1}{2}+\alpha\Big)\Gamma\Big(\frac{3}{2}+\alpha\Big)\,\frac{2}{1+2\alpha}\,.
\]
Combining all together,
\[
 \begin{split}
 \lim_{c \to +\infty} (\tau_{\alpha,\lambda,\gamma(c),c}^{(D)})^{-1}\,&=\,\frac{1}{2} \Big(\pi \sec (\pi \alpha)+2^{\frac{1}{2}+\alpha} (m \lambda)^{-\frac{1}{2}-\alpha}\, \Gamma\big({\textstyle\frac{1}{2}+\alpha}\big)\Gamma\big({\textstyle\frac{3}{2}+\alpha}\big)\Big(\lim_{c \to +\infty} \frac{2 m c \gamma(c)}{1 +2\alpha}\Big)\Big) \\
 &=\,\frac{1}{2} \Big(\pi \sec (\pi \alpha)+2^{\frac{1}{2}+\alpha} (m \lambda)^{-\frac{1}{2}-\alpha}\, \Gamma\big({\textstyle\frac{1}{2}+\alpha}\big)\Gamma\big({\textstyle\frac{3}{2}+\alpha}\big)\,\theta\Big)\,=\,(\tau_{\alpha,\lambda,\theta}^{(S)})^{-1}\,,
 \end{split}
\]
having implemented \eqref{eq:LimitGammaTheta} in the second equality and \eqref{eq:TauS} in the third.
\end{proof}

The final ingredient is the proof that eventually as $c\to +\infty$ all resolvent sets $\rho\big(\mathsf{h}_{\alpha,0}^{(\gamma(c))}-mc^2\big)$ along the trajectory $\gamma\equiv\gamma(c)$ satisfying the scaling \eqref{eq:LimitGammaTheta}, do contain a common, fixed (i.e., $c$-independent) sub-interval of the negative real half-line.

 \begin{proposition}\label{prop:myprop}
  Let $\alpha\in(-\frac{1}{2},\frac{1}{2})$. For given $\theta\in\mathbb{R}\cup\{\infty\}$, let $\gamma=\gamma(c)$ scale with $c$ according to \eqref{eq:LimitGammaTheta}, that is,
 \begin{equation*}
		\lim_{c \to +\infty} \frac{\,2\, m\, c \,\gamma(c)\,}{1+ 2\alpha} = \theta \, .
	\end{equation*}
 Then, eventually as $c\to +\infty$, say, $\forall c>c_0$ for some sufficiently high threshold $c_0>0$,
 the whole interval $(-2mc^2,0)$, with the possible exception of an open, $c$-independent, proper sub-interval $I_{\alpha,\theta,c_0}$, is entirely contained in $\rho\big(\mathsf{h}_{\alpha,0}^{(\gamma(c))}-mc^2\big)$.
 \end{proposition}

 As a consequence of Proposition \ref{prop:myprop}, there is a convenient threshold $c_0>0$ and a non-empty region
  \begin{equation}\label{eq:magicJ}
  J_{\alpha,\theta,c_0}\,:=\,(-2mc_0^2)\setminus I_{\alpha,\theta,c_0}
 \end{equation}
 such that
 \begin{equation}
  J_{\alpha,\theta,c_0}\,\subset\,\bigcap_{c>c_0}\rho\big(\mathsf{h}_{\alpha,0}^{(\gamma(c))}-mc^2\big)\,.
 \end{equation}
 For the present purposes we do not quite need to characterise $J_{\alpha,\theta,c_0}$ further, nor need we identify for which values of $\theta$ there is no interval $I_{\alpha,\theta,c_0}$ to remove from $(-2mc_0^2,0)$. 
 We only need the crucial property that either $J_{\alpha,\theta,c_0}$ is a non-empty open interval, or a non-empty difference of an open interval minus an open sub-interval. Owing to such features, \emph{it is in this very $J_{\alpha,\theta,c_0}\subset\mathbb{R}^-$ that we shall take $-\lambda$ in order to compute the limits}
 \[
  \lim_{c\to +\infty}(\mathsf{h}_{\alpha,0}^{(\gamma(c))}-mc^2+\lambda)^{-1}\,.
 \]

 \begin{proof}[Proof of Proposition \ref{prop:myprop}]
 From Theorem \ref{prop:Dirac}(iv), $\rho(\mathsf{h}_{\alpha,0}^{(\infty)}-mc^2)$ contains the interval $(-2mc^2,0)$, for $\gamma(c)\neq\infty$ $\rho(\mathsf{h}_{\alpha,0}^{(\gamma(c))}-mc^2)$ contains $(-2mc^2,0)$ up to possibly one single point $E_{\alpha,\theta,c,\gamma(c)}$, the actual eigenvalue of $\mathsf{h}_{\alpha,0}^{(\gamma(c))}-mc^2$ inside the gap, and for $\lambda\in(-2mc^2,0)\setminus \{E_{\alpha,\theta,c,\gamma(c)}\}$ one has
  \begin{equation*}
(\mathsf{h}_{\alpha,0}^{(\gamma(c))}-mc^2+\lambda)^{-1} \,=\, (\mathsf{h}_{\alpha,0}^{(\infty)}-mc^2+\lambda)^{-1}+ \tau_{\alpha,\lambda,\gamma(c),c}^{(D)}\,\frac{\lambda}{c^2}  \big| \Phi^{(D)}_{\alpha,\lambda,c} \big\rangle \big\langle \Phi_{\alpha,\lambda,c}^{(D)} \big|\,.
\end{equation*}
 Consider now the two limits
 \begin{equation*}
\lim_{c \to +\infty}\frac{\lambda}{c^2} \big| \Phi^{(D)}_{\alpha,\lambda,c} \big\rangle \big\langle \Phi_{\alpha,\lambda,c}^{(D)} \big|\,=\,2m \big| \Phi_{\alpha,\lambda}^{(S)} \big\rangle \big\langle \Phi_{\alpha,\lambda}^{(S)}\big|\,\oplus_{\mathbb{C}^2}\,\mathbb{O}\,,
\end{equation*}
 following from the explicit expression of the spinor function $\Phi^{(D)}_{\alpha,\lambda,c}$ and of the scalar function $\Phi_{\alpha,\lambda}^{(S)}$ -- see \eqref{eq:PhiD} and \eqref{Phi-F-Schr-0} above, respectively -- and
 \begin{equation*}
   \lim_{c \to +\infty}\tau_{\alpha,\lambda,\gamma(c),c}^{(D)}\,=\,\tau_{\alpha,\lambda,\theta}^{(S)}\,,
 \end{equation*}
 already proved with Lemma \ref{lem:tauprefactors}. In view of such limits and of the above expression for the resolvent of $\mathsf{h}_{\alpha,0}^{(\gamma(c))}-mc^2$, one concludes that for arbitrary $\varepsilon>0$ there is $c_\varepsilon>0$ such that for $c_1,c_2>c_\varepsilon$ and for $-\lambda\in(-2mc_\varepsilon^2,0)\setminus\{E_{\alpha,\theta,c_1,\gamma(c_1)},E_{\alpha,\theta,c_2,\gamma(c_2)}\}$, one has
 \[
  \begin{split}
     \|
     (\mathsf{h}_{\alpha,0}^{(\gamma(c_1))}-mc_1^2&+\lambda)^{-1}-(\mathsf{h}_{\alpha,0}^{(\gamma(c_2))}-mc_2^2+\lambda)^{-1}\big\|_{L^2(\mathbb{R}^+;\mathbb{C}^2)\to L^2(\mathbb{R}^+;\mathbb{C}^2)} \,\leqslant\,\varepsilon\,.
  \end{split}
 \]
 As a further consequence \cite[Theorem VIII.23(b)]{rs2}, the eigenvalues $E_{\alpha,\theta,c_1,\gamma(c_1)}$ and $E_{\alpha,\theta,c_2,\gamma(c_2)}$ are $\varepsilon$-close, since they are the unique poles, respectively, of the resolvent of $\mathsf{h}_{\alpha,0}^{(\gamma(c_1))}-mc_2^2$ and of $\mathsf{h}_{\alpha,0}^{(\gamma(c_2))}-mc_2^2$. This establishes that eventually in $c$ the eigenvalues $E_{\alpha,\theta,c,\gamma(c)}$ inside the gap, \emph{if any}, are asymptotically close to a $c$-independent limit, and hence fall all within an open sub-interval $I_{\alpha,\theta,c_0}$ of $(-2mc_0^2,0)$.
 \end{proof}

\begin{proof}[Proof of Theorem \ref{thm:Main}] With respect to the block-decomposition \eqref{HoplusDecomp}-\eqref{eq:spinHk}, namely
\[
L^2(\mathbb{R}^2;\mathbb{C}^2,\ud x\,\ud y)\,\cong\,\bigoplus_{k\in\mathbb{Z}} L^2(\mathbb{R}^+;\mathbb{C}^2,\ud r)\,,
\]
respectively, with respect to the block-decomposition \eqref{SchrHoplusDecomp}-\eqref{eq:SchrHk}, namely
\[
 L^2(\mathbb{R}^2;\mathbb{C},\ud x\,\ud y)\,\cong\,\bigoplus_{k\in\mathbb{Z}}L^2(\mathbb{R}^+;\mathbb{C},\ud r)\,,
\]
we write $H_\alpha^{(\gamma)}$ and $S_\alpha^{(\theta)}$ in reduced form as
\[
 \begin{split}
  H_\alpha^{(\gamma)}\,&=\,\bigg(\bigoplus_{\substack{ k\in\mathbb{Z} \\ k\leqslant -2}} \overline{\mathsf{h}_{\alpha,k}}\;\bigg)\oplus\; \overline{\mathsf{h}_{\alpha,-1}}  \; \oplus \;\mathsf{h}_{\alpha,0}^{(\gamma)} \;\oplus  \bigg(\bigoplus_{\substack{ k\in\mathbb{Z} \\ k\geqslant 1}} \overline{\mathsf{h}_{\alpha,k}}\;\bigg), \\
  S_\alpha^{(\theta)}\,&=\,\bigg(\bigoplus_{\substack{ k\in\mathbb{Z} \\ k\leqslant -2}} \overline{S_{\alpha,k}}\;\bigg)\oplus\; S_{\alpha,-1}^{(F)} \; \oplus \;S_{\alpha,0}^{(\theta)} \;\oplus  \bigg(\bigoplus_{\substack{ k\in\mathbb{Z} \\ k\geqslant 1}} \overline{S_{\alpha,k}}\;\bigg).
 \end{split}
\]

In the blocks $k\in\mathbb{Z}\setminus\{0\}$ there is no dependence on $\gamma$ or $\theta$: the block-wise limit $c\to +\infty$, upon subtracting the rest energy $mc^2$, and in the norm resolvent sense, is a direct consequence of Proposition \ref{prop:LimitFriedrichs} (i) and (iii). Here resolvents are taken at the point $z$ in the gap $(-2mc^2,0)$ and hence automatically also at $z$ with non-zero imaginary part, owing to self-adjointness.

Concerning the block $k=0$, it follows from Theorem
\ref{prop:Dirac}(iv) and Proposition \ref{prop:myprop} that
there is $c_0>0$ and a non-empty $c$-independent union $J_{\alpha,\theta,c_0}$ of two intervals of the negative real line ($J_{\alpha,\theta,c_0}$ is actually defined in \eqref{eq:magicJ}) such that
 \begin{equation*}
  J_{\alpha,\theta,c_0}\,\subset\,\bigcap_{c>c_0}\rho\big(\mathsf{h}_{\alpha,0}^{(\gamma(c))}-mc^2\big)
 \end{equation*}
 and
 \begin{equation*}\tag{i}\label{eq:resres}
 \begin{split}
  (\mathsf{h}_{\alpha,0}^{(\gamma(c))}-mc^2+\lambda)^{-1} \,&=\, (\mathsf{h}_{\alpha,0}^{(\infty)}-mc^2+\lambda)^{-1}+ \tau_{\alpha,\lambda,\gamma(c),c}^{(D)}\,\frac{\lambda}{c^2}  \big| \Phi^{(D)}_{\alpha,\lambda,c} \big\rangle \big\langle \Phi_{\alpha,\lambda,c}^{(D)} \big|\,, \\
  & \qquad \forall (-\lambda)\in J_{\alpha,\theta,c_0}\, .
 \end{split}
\end{equation*}
 We already proved that
	\begin{equation*}\tag{ii}\label{eq:limfin-part1}
			\lim_{c \to +\infty} (\mathsf{h}_{\alpha,0}^{(\infty)} - mc^2+\lambda)^{-1} \,=\, (\mathsf{S}_{\alpha,0}^{(\infty)}+\lambda)^{-1}\oplus_{\mathbb{C}^2} \mathbb{O}
		\end{equation*}
 (Proposition \ref{prop:LimitFriedrichs}(ii)), as well as
  \begin{equation*}\tag{iii}\label{eq:limfin-part2}
\lim_{c \to +\infty}\frac{\lambda}{c^2} \big| \Phi^{(D)}_{\alpha,\lambda,c} \big\rangle \big\langle \Phi_{\alpha,\lambda,c}^{(D)} \big|\,=\,2m \big| \Phi_{\alpha,\lambda}^{(S)} \big\rangle \big\langle \Phi_{\alpha,\lambda}^{(S)}\big|\,\oplus_{\mathbb{C}^2}\,\mathbb{O}
\end{equation*}
 (following from \eqref{eq:PhiD} and \eqref{Phi-F-Schr-0} above), and
 \begin{equation*}\tag{iv}\label{eq:limfin-part3}
   \lim_{c \to +\infty}\tau_{\alpha,\lambda,\gamma(c),c}^{(D)}\,=\,\tau_{\alpha,\lambda,\theta}^{(S)}
 \end{equation*}
 (Lemma \ref{lem:tauprefactors}).
  Plugging \eqref{eq:limfin-part1}, \eqref{eq:limfin-part2}, and \eqref{eq:limfin-part3} into \eqref{eq:resres} yields
\[
\begin{split}
 \lim_{c \to +\infty}(\mathsf{h}_{\alpha,0}^{(\gamma(c))} - mc^2+\lambda)^{-1}\,&=\,\Big( (\mathsf{S}_{\alpha,0}^{(\infty)}+\lambda)^{-1}+2m\tau_{\alpha,\lambda,\theta}^{(S)}\big| \Phi_{\alpha,\lambda}^{(S)} \big\rangle \big\langle \Phi_{\alpha,\lambda}^{(S)}\big|\Big)\oplus_{\mathbb{C}^2} \mathbb{O} \\
 &\qquad \forall (-\lambda)\in J_{\alpha,\theta,c_0}\, .
\end{split}
\]
 In view of Theorem \ref{prop:k0}(iii), this limit is precisely $(\mathsf{S}_{\alpha,0}^{(\theta)}+\lambda)^{-1}\oplus_{\mathbb{C}^2} \mathbb{O}$. The extension of the above limit from $(-\lambda)\in J_{\alpha,\theta,c_0}$ to $(-\lambda)\in \mathbb{C}^\pm$ is automatic by holomorphicity.

 This establishes \eqref{eq:NonRelLimThmMain} in all $k$-blocks and completes the proof.
\end{proof}

\appendix

{
\section{Elements of Kre\u{\i}n-Vi\v{s}ik-Birman-Grubb self-adjoint extension theory}\label{sec:proofOfMain}

 We extract from the general discussion available, e.g., in
 \cite{GMO-KVB2017}, \cite[Chapter 2]{GM-SelfAdj_book-2022}, \cite[Chapter 13]{Grubb-DistributionsAndOperators-2009}, and \cite{KM-2015-Birman}, a concise summary of the tools of the Kre\u{\i}n-Vi\v{s}ik-Birman-Grubb self-adjoint extension theory needed in the present work.

 The main result is the following (see, e.g., \cite[Theorem 2.17]{GM-SelfAdj_book-2022}).

\begin{theorem}\label{thm:KVB-General-App}
	Let $S$ be a densely defined symmetric operator on a complex Hilbert space $\mathcal{H}$, which admits a self-adjoint extension $S_D$ that has everywhere defined bounded inverse on $\mathcal{H}$. Then
	\begin{equation}
	 \mathcal{D}(S^*)\,=\,\left\{\varphi+S_D^{-1}z+v \,\Big| \begin{array}{c} \varphi \in \mathcal{D}(\overline{S}),  \\
			z,v \in \ker S^* \end{array} \right\} \, .
	\end{equation}
    Moreover, there is a one-to-one correspondence between the family of the self-adjoint extensions of $S$ in $\mathcal{H}$ and the family of self-adjoint operators on Hilbert subspaces of $\ker S^*$. If $T$ is any such operator, in the correspondence $T \leftrightarrow S_T$ each self-adjoint extension $S_T$ of $S$ is given by
	\begin{equation}\label{eq:KVB-generalissima}
		\begin{split}
			S_T \,&=\, S^*\upharpoonright \mathcal{D}(S_T)\,, \\
			\mathcal{D}(S_T)\,&=\, \left\{\varphi+S_D^{-1}(Tv +w)+v \,\Big| \begin{array}{c} \varphi \in \mathcal{D}(\overline{S}), \, v \in \mathcal{D}(T) \\
			w \in \ker S^* \cap \mathcal{D}(T)^{\perp} \end{array} \right\} \, .
		\end{split}
	\end{equation}

\end{theorem}

Theorem \ref{thm:KVB-General-App}, specialised to the case of interest of the present work, namely deficiency indices $(1,1)$, reads as follows.

 \begin{theorem}\label{thm:KVB-General-App11}
	Let $S$ be a densely defined symmetric operator on a complex Hilbert space $\mathcal{H}$, with deficiency indices $(1,1)$, and admitting a self-adjoint extension $S_D$ that has everywhere defined bounded inverse on $\mathcal{H}$. Define
    \[
     \Psi\,:=\, S_D^{-1}\Phi
	\]
	for a choice $\Phi\in\cH\setminus\{0\}$ such that
	\[
	 \ker S^* = \mathrm{span}_{\mathbb{C}}\{\Phi\}\,.
	\]
	Then
	\begin{equation}\label{eq:domainstar11}
	 \mathcal{D}(S^*)\,=\,\left\{\varphi+c_1\Psi+c_0\Phi \,\Big| \begin{array}{c} \varphi \in \mathcal{D}(\overline{S}),  \\
			c_0,c_1\in\mathbb{C} \end{array} \right\} \, .
	\end{equation}
    Moreover, the self-adjoint extensions of $S$ in $\cH$ form the family $(S^{(\beta)})_{\beta\in\mathbb{R}\cup\{\infty\}}$ where $S^{(\infty)}=S_D$ and where, for $\beta\in\mathbb{R}$,
	\begin{equation}\label{eq:KVB-11}
		\begin{split}
			S^{(\beta)}\,&=\,S^*\upharpoonright \mathcal{D}(S^{(\beta)})\,, \\
			\mathcal{D}(S^{(\beta)})\,&=\, \left\{\varphi+c(\beta\Psi+\Phi) \,\Big| \begin{array}{c} \varphi \in \mathcal{D}(\overline{S})\,, \\
			c\in\mathbb{C} \end{array} \right\} \, .
		\end{split}
	\end{equation}

\end{theorem}

 The operator $T$ is the `\emph{Birman parameter}' of the extension $S_T$. In the special case of deficiency indices $(1,1)$, $T$ is the multiplication by $\beta\in\mathbb{R}$ on $\mathrm{span}_{\mathbb{C}}\{\Phi\}$, and each self-adjoint extension $S^{(\beta)}$ of $S$ is identified by restricting $S^*$ to the sub-domain obtained from \eqref{eq:domainstar11} with the condition of self-adjointness
 \begin{equation}
  c_1\,=\beta c_0\,.
 \end{equation}
 The extension $S^{(\infty)}\equiv S_D$ is characterised by $c_0=0$.

 Theorem \ref{thm:KVB-General-App11} shows that in order to find the domain \eqref{eq:KVB-11}, one needs the following data: the operator closure domain $\mathcal{D}(\overline{S})$, a spanning element $\Phi$ of $\ker S^*$, and the action of $S_D^{-1}$ on $\Phi$ for a distinguished self-adjoint extension $S_D$ of $S$ having everywhere defined bounded inverse on $\mathcal{H}$.

 The Birman parameter $T$ also determines the expression of $S_T^{-1}$ for those extensions $S_T$ that are invertible with everywhere defined and bounded inverse in $\cH$. In the special case of deficiency indices $(1,1)$ such expression is given as follows (see, e.g., \cite[Theorems 2.21 and 2.27]{GM-SelfAdj_book-2022}).

\begin{theorem}\label{thm:KVB2-d1}
	Under the assumptions of Theorem \ref{thm:KVB-General-App11}, the self-adjoint extension $S^{(\beta)}$ is invertible with everywhere defined and bounded inverse in $\cH$ if and only if $\beta\neq 0$, in which case
	\begin{equation}\label{eq:Res11KVB}
		(S^{(\beta)})^{-1}=S_D^{-1}+\beta^{-1}\|\Phi\|_\cH^{-2}|\Phi\rangle\langle\Phi|\,.
	\end{equation}
\end{theorem}

\section{Canonical decomposition of the AB-operators}\label{app:A}

\subsection{Dirac operator} Let us consider the following unitary operators:
\begin{equation}
	\begin{array}{ccccc}
		V_1& :& L^2(\mathbb{R}^2;\mathbb{C}^2, \ud x \ud y) & \longrightarrow & L^2(\mathbb{R}^+ \times \mathbb{S}^1; \mathbb{C}^2,  r \ud r \ud \vartheta) \\
		& & \psi(x,y)=\begin{pmatrix}
			\psi_+(x,y) \\ \psi_-(x,y)
		\end{pmatrix} & \mapsto & \begin{pmatrix}
			e^{\ii \frac{\vartheta}{2}} \psi_+(x(r,\vartheta),y(r,\vartheta)) \\
			e^{-\ii \frac{\vartheta}{2}} \psi_-(x(r,\vartheta),y(r,\vartheta))
		\end{pmatrix}
	\end{array}
\end{equation}
and then
\begin{equation}
	V_2=U_2 \oplus U_2, \qquad V_3=\widetilde{\mathscr{F}}_{\vartheta} \oplus \widetilde{\mathscr{F}}_\vartheta
\end{equation}
\begin{equation}
\begin{array}{ccccc}
	V_3 & : & L^2(\mathbb{R}^+ \times \mathbb{S}^1; \mathbb{C}^2, \ud r \ud \vartheta) & \longrightarrow & \bigoplus_{k \in \mathbb{Z}} L^2(\mathbb{R}^+; \mathbb{C}^2, \ud r) \\
	& & \psi(r,\vartheta) & \mapsto & \frac{1}{2 \pi} \int_0^{2\pi} e^{\ii (k-\frac{1}{2}) \vartheta} \psi(r, \vartheta) \, \ud \vartheta\ .
\end{array}
\end{equation}
Then for the Dirac operator with Aharonov-Bohm vector potential $H_\alpha$ defined in \eqref{magneticDirac} one has
\begin{equation}
	V_3 V_2 V_1 H_\alpha V_1^* V_2^* V_3^* = \bigoplus_{k \in \mathbb{Z}} \mathsf{h}_{\alpha,k} \, .
\end{equation}

\subsection{Schr\"{o}dinger operator}

Let us consider the following unitary operators:
\begin{equation}
\begin{array}{ccccc}
	U_1 &: & L^2(\mathbb{R}^2, \ud x \ud y) &\longrightarrow & L^2(\mathbb{R}^+\times \mathbb{S}^1, r \ud r \ud \vartheta) \\
	& & \psi(x,y) & \mapsto & (U_1 \psi)(r,\vartheta)=\psi(x(r,\vartheta),y(r,\vartheta))
\end{array}
\end{equation}
\begin{equation}
\begin{array}{ccccc}
	U_2 &: & L^2(\mathbb{R}^+ \times \mathbb{S}^1, r \ud r \ud \vartheta) &\longrightarrow & L^2(\mathbb{R}^+\times \mathbb{S}^1, \ud r \ud \vartheta) \\
	& & \psi(r,\vartheta) & \mapsto & (U_2 \psi)(r,\vartheta)=\sqrt{r}\psi(r,\vartheta)
\end{array}
\end{equation}
\begin{equation}
\begin{array}{ccccc}
	\mathscr{F}_\vartheta &: & L^2(\mathbb{R}^+ \times \mathbb{S}^1, \ud r \ud \vartheta) &\longrightarrow & \bigoplus_{k \in \mathbb{Z}}L^2(\mathbb{R}^+, \ud r) \\
	& & \psi(r,\vartheta) & \mapsto & (\mathscr{F}_\vartheta \psi)_k(r)=\frac{1}{\sqrt{2 \pi}} \int_0^{2 \pi} e^{\ii k \vartheta} \psi(r,\vartheta) \, \ud \vartheta\ .
\end{array}
\end{equation}
Then for the Schr\"{o}dinger operator with Aharonov-Bohm vector potential $S_\alpha$ defined in \eqref{eq:2DSochroedinger}, one has
\begin{equation}
	\mathscr{F}_\vartheta U_2 U_1 S_\alpha U_1^* U_2^* \mathscr{F}_\vartheta^* = \bigoplus_{k \in \mathbb{Z}} \mathsf{S}_{\alpha,k}
\end{equation}


\def\cprime{$'$}

\end{document}